\documentclass{article}
\usepackage[utf8]{inputenc}
\usepackage{amsmath}
\usepackage{amssymb}
\usepackage{braket}
\usepackage{graphicx}
\usepackage{float}
\usepackage{braket}
\usepackage{bbold}
\usepackage[margin=1in]{geometry}
\usepackage{amsthm}
\usepackage[dvipsnames]{xcolor}
\usepackage{comment}
\usepackage{mathtools}
\usepackage{hyperref}
\usepackage{thm-restate}
\usepackage{thmtools}
\usepackage[capitalise]{cleveref}
\usepackage{tikz}
\usepackage{standalone}
\usepackage{textcomp}
\usepackage{cancel}
\usepackage{stmaryrd}
\usepackage{subcaption}
\usepackage[sorting=none]{biblatex}
\usepackage{float}
\bibliography{a}

\def\dbra#1{\mathinner{\langle\!\langle{#1}|}}
\def\dket#1{\mathinner{|{#1}\rangle\!\rangle}}
    
	\newcommand{\proj}[2]{\ket{#1}\bra{#2}}

\DeclareMathOperator{\tr}{tr}
\newcommand{\C}{\mathcal}

\newtheorem{defi}{Definition}
\newtheorem{lemma}{Lemma}

\newtheorem{example}{Example}
\newtheorem{remark}{Remark}

\Crefname{prop}{Proposition}{Propositions}
\crefname{prop}{Prop.}{Props.}
\Crefname{defi}{Definition}{Definitions}
\crefname{defi}{Def.}{Defs.}
\Crefname{prop}{Proposition}{Propositions}
\crefname{lemma}{Lemma}{Lemmas}
\crefname{coro}{Cor.}{Cors.}
\Crefname{coro}{Corollary}{Corollaries}
\crefname{exampe}{Ex.}{Exs.}
\Crefname{exampe}{Example}{Examples}
\crefname{remark}{Remark}{Remarks}
\crefname{conj}{Conj.}{Conjs.}
\Crefname{conj}{Conjecture}{Conjectures}

\hypersetup{
    colorlinks=true,
    linkcolor=blue,
    filecolor=magenta,      
    urlcolor=cyan,
    citecolor=green,
    }

\DeclareFieldFormat{title}{\mkbibquote{#1}}

\title{Connecting indefinite causal order processes to composable quantum protocols in a spacetime}
\author{Matthias Salzger}
\date{\today}

\begin{document}
\pagenumbering{gobble}

\begin{titlepage}
\begin{flushleft}
\includegraphics[]{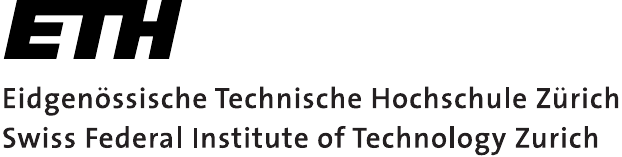}
\end{flushleft}
\hrule
\centering
\vspace*{3cm}

\begin{huge}
\textbf{Connecting indefinite causal order processes to composable quantum protocols in a spacetime}
\end{huge}

\vspace{2cm}
\begin{Large}            
Master thesis

\vspace{1cm}

Matthias Salzger

\vspace{2cm}

Supervised by:

Dr. V. Vilasini

Prof. Dr. Renato Renner

\vspace{2cm}

April 2022
\end{Large}
\vfill

\begin{large}
Quantum Information Theory Group

Institute for Theoretical Physics

ETH Zürich
\end{large}

\end{titlepage}

\begin{abstract}

Process matrices provide a framework to model quantum information processing protocols in the absence of a well-defined acyclic causal order. While they are defined without reference to a background spacetime, it is an open question to characterize the subset of processes that are physically realizable in a fixed background spacetime. A related problem is that process matrices are known to be non-composable while composability is a basic property of physical processes. On the other hand, so-called causal boxes define a framework that allows for arbitrary composition and model physical protocols defined on a fixed background space-time, which include scenarios where quantum states may be in a superposition of different spacetime locations.

To address these questions, we compare quantum circuits with quantum control of causal order (QC-QC), a subset of process matrices, which can be interpreted as generalized quantum circuits, and process boxes, a subset of causal boxes, which can be interpreted as processes. We analyze their state spaces and define a notion of operational equivalence between QC-QCs and process boxes. We then explicitly construct for each QC-QC an operationally equivalent process box. This allows us to define composition of QC-QCs in terms of composition of causal boxes which is well-defined. In doing so, we also find that resolving the composability issue requires a framework that goes beyond process matrices, in particular one that allows multiple messages and spacetime information.

We further show that process boxes admit a unitary extension which is itself a process box, and conjecture that each process box can be reduced to a simpler, but operationally equivalent process box, where only the time information is relevant. Based on this conjecture, we construct an operationally equivalent QC-QC for each process box.

Our results indicate that the only class of processes that can be physically implemented in a fixed background spacetime are those that can be interpreted as quantum circuits with quantum controlled superpositions of orders. Further, they also reveal that the composability issue can be resolved by embedding processes in a spacetime structure. This in turn sheds light on the connection between physical realizability in a spacetime and composability.
\end{abstract}

\newpage
\newgeometry{left=2in,right=2in, top=2in}
\section*{Acknowledgements}

I would like to thank Vilasini for going above and beyond as a supervisor, for the many long discussions and for so often bringing a new perspective that I had not considered before to the problem. I am grateful to Prof. Renato Renner and the QIT group for giving me the opportunity to work on this thesis. Finally, I would like to thank my family, especially my mother and Maeki and my grandparents, who always treated their support as a matter of course, even though it was what allowed me to study and succeed at ETH in the first place.


\newpage
\restoregeometry
\pagenumbering{arabic}

\tableofcontents

\newpage

\part{Introduction}\label{sec:intro}

Causal relations are one of the most fundamental notions in science. Carrying out an experiment can ultimately be viewed as asking what the causal relationship between some input variable and some output variable is. The classical view on causality is that of Reichenbach's principle \cite{reichenbach1991direction}, which states that given two correlated variables, either one of the two is the cause of the other or they share a common cause given by a third variable. The principle furthermore states that if the latter is the case, then conditioning on the common cause makes the two correlated variables independent. 

However, when considering quantum mechanics, we find that Reichenbach's principle (in particular, the latter part) cannot be so easily applied \cite{Allen_2017}. Bell's theorem \cite{Bell} tells us that there are correlations which cannot be explained by any local hidden variables (i.e. a common cause). Giving up the intuitive notion of causality in the classical picture does come with an interesting consolation prize, namely the violation of the Bell inequalities, which provides computational advantages in various tasks \cite{Brunner2014}.

The failure of Reichenbach's principle in quantum mechanics has led to significant research into quantum causal models that generalize the principles of classical causality \cite{Barrett2019, Henson_2014}. The result is that there now exist frameworks that can characterize and causally explain not only the quantum correlations arising from the Bell scenario, but those of general multipartite quantum networks, as long as the operations of the parties have a well-defined acyclic order.


One could now ask do the quantum causal models introduced in the wake of Bell's theorem account for all possible causal relations? In recent years, modeling causal relations in the absence of a well-defined acyclic causal order has been an area of interest \cite{Oreshkov_2012, Hardy2005,Zych_2019}. An example where research in this direction could be relevant is the intersection between quantum mechanics and gravity \cite{Oreshkov_2012, Hardy2005,Zych_2019}. As the gravitational effects of matter determine the spacetime geometry, then if that matter is quantum, one might expect superpositions of different geometries. We then also expect superpositions of causal orders, as causal relationships are determined by the geometry in relativistic physics. 

However, even outside the regime of quantum gravity more general causal structures can have relevance as exemplified by the quantum switch \cite{Chiribella2013}. This process involves two agents, Alice and Bob, who each apply a unitary to a target system, where the order in which they apply their unitaries depends coherently on the state of a quantum control system. This process can be used to determine, more efficiently than using quantum circuits with a fixed, acyclic order, whether a black box unitaries is commuting or anti-commuting \cite{Chiribella_2012} among a number of other tasks \cite{Colnaghi_2012, Ara_jo_2014} and has been implemented experimentally \cite{Procopio_2015,Rubino_2017,Goswami_2018}. 

There are several frameworks capable of describing the quantum switch and other general causal structures. One such framework is the process matrix framework \cite{Oreshkov_2012} which drops the assumption of a fixed background spacetime altogether. As it turns out, this allows some of these process matrices to violate so-called causal inequalities \cite{Oreshkov_2012}, which are bounds set by (convex combinations of) fixed causal orders and can be viewed as an analogue to the Bell inequalities. A common feature, which is in a sense the analogue to entanglement in the Bell scenario, of process matrices that violate these causal inequalities is that they do not admit a causal explanation in terms of a fixed, acyclic causal structure, even when allowing quantum controlled superpositions of orders \cite{Oreshkov_2012, Ara_jo_2015, Branciard_2015}. It remains an open question whether such structures are physically realizable \cite{Oreshkov_2016,Wechs_2019}. Some research in this direction has adopted a top-down approach, for example by showing that not all process matrices can be unitarized \cite{Ara_jo_2017}, something we would expect is possible for physical processes due to the connection between unitarity and the reversibility of quantum theory \cite{Chiribella_2010}. It has also been shown that process matrices, even relatively simple ones, are not composable \cite{Gu_rin_2019}, even though composability is a basic property that we expect in an experimental setting \cite{Procopio_2015, Rubino_2017}-- we can combine two physical experiments to form a new physical experiment. 

On the other hand, a bottom-up approach to the characterization of process matrices is given by the framework of quantum circuits with quantum control of causal orders (QC-QC) \cite{Wechs_2021}. As the name suggests, this framework attempts to capture such processes which can be interpreted as circuits where the causal order is determined coherently or dynamically. As it turns out, QC-QCs cannot violate causal inequalities and so whether this class captures all physical processes is an important open question \cite{Wechs_2021}.

The above approaches all operate within an information theoretic understanding of causality, which is based on the flow of information between systems. In order to address the physicality question, we must also consider embedding such an information-theoretic causal structure in a spacetime, and the relativistic notion of causality that arises therein which can be viewed as a compatibility condition between the two notions of causation \cite{Paunkovi__2020, Vilasini_2022}. One approach to instantiating the information-theoretic process framework with time information, is by considering quantum reference frames \cite{Paunkovi__2020,Castro_Ruiz_2020,Baumann_2022}. Quantum reference frames are also formulated without reference to a background spacetime, with temporal information modeled by quantum clocks which are themselves quantum systems. Here, we consider the open question of physical realizability of processes consistent with relativistic causality within a fixed spacetime. Understanding what is possible or impossible within the most general scenarios embeddable in a fixed background spacetime would also provide insights into how physics in more exotic spacetimes, physical regimes or when characterized by quantum reference frames might fundamentally differ.


A framework that models composable information processing protocols on a fixed background spacetime is given by that of causal boxes \cite{Portmann_2017}. Causal boxes can model scenarios where quantum states may be sent or received at a superposition of different spacetime locations in the background spacetime. It was originally developed for quantum cryptographical purposes \cite{Vilasini_2019,Laneve2021} and therefore an emphasis was put on the composability of the protocols described by the framework. The composition of two or more causal boxes is again a causal box \cite{Portmann_2017}. A recent work \cite{Vilasini_2022} adopts a top-down approach to disentangle information-theoretic and spacetime notions of causality and characterizes their compatibility in general scenarios where quantum states may be in a superposition of spacetime locations. 
This work suggests that causal boxes are the most general description of quantum protocols that can be implemented in a fixed spacetime such that relativistic causality is satisfied.

Now, the question of physical realizability and composability of processes in a fixed spacetime reduces to asking which subset of processes can be modeled as causal boxes. However, these two frameworks, causal boxes and process matrices, while sharing some common features, are quite different. For example, causal boxes are composable and allow multiple rounds of information processing and superpositions of different number of messages while process are not composable and only consider agents in closed labs acting once on single messages, although possible at a superposition of different times. 

In order to bridge some of the differences between the frameworks, the framework of process boxes was developed in \cite{Vilasini_2020}. This framework imposes additional constraints on causal boxes, which are on the one hand needed to restrict to single-round protocols, which correspond to the kind of setting described by process matrices. And on the other hand, these constraints are also needed to rule out trivial violations of causal inequalities, using simple causally ordered strategies (analogous to how one demands non-communicating parties to exclude trivial violations of Bell inequalities \cite{CHSH}). 
Like QC-QCs, process boxes cannot violate causal inequalities \cite{Vilasini_2020}. 

In this work, we show that process boxes and QC-QCs are equivalent in an operational sense if an additional assumption is imposed on the process box framework that ensures that agents always receive and send exactly one message. Such an assumption avoids situations where an agent is in a superposition of participating and not participating in the protocol, which is a priori allowed in the process box picture but not in the QC-QC framework. We analyze the state spaces of both frameworks and define a notion of equivalence between the operations that agents can apply in each framework. This then allows us to define operationally what it means for a QC-QC and a process box to be equivalent. Based on this definition, we show that there exists for each QC-QC an operationally equivalent process box and vice versa. We can then understand composition of QC-QCs in terms of composition of process boxes, which is well-defined. Furthermore, our results unify the top-down approach of \cite{Vilasini_2022, Vilasini_2020} and the bottom-up approaches of \cite{Wechs_2021, Purves_2021} towards the physicality problem for processes in a fixed spacetime, and suggests that the largest set of physically realizable processes in a fixed background spacetime corresponds to QC-QCs, up to slight generalizations that may be possible by relaxing assumptions of the process box framework. 

\section{Summary of contributions}\label{sec:contributions}

We begin this thesis with a review of the mathematical tools that we will need for the remainder of the thesis in \cref{sec:tools}, followed by a review of the process matrix framework in \cref{sec:pm} and the causal box framework in \cref{sec:cb}. In the respective sections, we also review the QC-QC framework (\cref{sec:circuitpm}) and the process box framework (\cref{sec:pb}).

We present our results in three parts. 

\begin{itemize}
    \item In \cref{sec:qcqctopb}, we discuss how QC-QCs can be given a process box description, starting with a discussion of what this means in the first place in and then constructing two possible ``extensions'', that is we extend the action of QC-QCs to the Fock spaces that define the overall causal box. This then allows us to resolve the composability problem of process matrices for the case of QC-QCs by composing their causal box description. 
    \item In \cref{sec:characterizing}, we characterize process boxes, showing that they admit a unitary extension and using redundancy in the framework to conjecture that we can simplify the spacetime configuration of a general process box to one where only temporal information is relevant.
    \item Finally, in \cref{sec:pbtoqcqc} we first discuss what it means to have a controlled process and argue that all process boxes are controlled. We then use this and the simplified spacetime configuration from \cref{sec:characterizing} to construct a QC-QC description of process boxes.
\end{itemize}

\newpage

\part{Review}\label{sec:review}
\section{Mathematical tools}\label{sec:tools}
\subsection{Notation}\label{sec:notation}

Throughout this thesis, we will differentiate Hilbert spaces with superscripts, e.g. $\C{H}^X$. We will use this superscript on states $\ket{\psi}^X$ to denote that it is an element of $\C{H}^X$. If it is clear from context which space a state belongs to, we sometimes drop the superscript on the state in order to avoid clutter. Additionally, we will denote the tensor product of two Hilbert spaces $\C{H}^X \otimes \C{H}^Y$ as $\C{H}^{XY}$. Given a set $\C{K}$ and an element of that set $k \in \C{T}$, we will usually write $\C{K} \backslash \{k\}$ as $\C{K} \backslash k$ to avoid clutter.

\subsection{Choi isomorphism}\label{sec:choi}

Frequently, it will be convenient for us to express channels as matrices. For this purpose, we use the Choi isomorphism. 

\begin{defi}[Choi matrix \cite{CHOI1975285}]
Given some Hilbert spaces $\C{H}^X$ with some basis $\{\ket{i}^X \}_i$ and $\C{H}^Y$ with some basis $\{\ket{i}^Y \}_i$, let $\C{M}: \C{L}(\C{H}^X) \rightarrow \C{L}(\C{H}^Y)$ be some linear map. We then define the Choi matrix of $\C{M}$ as

\begin{equation}
    M \coloneqq \C{I}^X \otimes \C{M} \dket{\mathbb{1}} \dbra{\mathbb{1}}^{XX}
\end{equation}

where $\dket{\mathbb{1}}^{XX} = \sum_i \ket{i}^X \ket{i}^X$ denotes the maximally entangled state.

\end{defi}

If $\C{M}$ is pure, i.e. there exists a linear operator $V: \C{H}^X \rightarrow \C{H}^Y$ such that $\C{M}(\rho) = V \rho V^\dagger$ for any $\rho \in \C{L}(\C{H}^X)$, we can instead work with the simpler Choi vector.

\begin{defi}[Choi vector  \cite{CHOI1975285, Royer1991}]
Given some Hilbert spaces $\C{H}^X$ with some basis $\{\ket{i}^X \}_i$ and $\C{H}^Y$ with some basis $\{\ket{i}^Y \}_i$, let $V: \C{H}^X \rightarrow \C{H}^Y$ be a linear operator. We then define the Choi vector of $V$ as 

\begin{equation}
    \dket{V} \coloneqq \mathbb{1}^X \otimes V \dket{\mathbb{1}}^{XX}.
\end{equation}
\end{defi}

Note that the maximally entangled state is the Choi vector of the identity, which justifies our use of the notation $\dket{\mathbb{1}}^{XX}$.

\subsection{Link product}\label{sec:link}

The link product \cite{Chiribella_2008, Chiribella_2009} is a useful tool when calculating Choi matrices and vectors. It allows us to calculate the Choi matrix/vector of a composition of maps from the Choi matrices/vectors of the maps themselves.

\begin{defi}[Link product for vectors \cite{Wechs_2021}]\label{def:linkvec}
Let $\C{H}^X, \C{H}^Y, \C{H}^Z$ be non-overlapping Hilbert spaces. For two vectors $\ket{a} \in \C{H}^{XY}, \ket{b} \in \C{H}^{YZ}$ we define their link product

\begin{equation}
    \ket{a} * \ket{b} \coloneqq \sum_i \ket{a_i} \otimes \ket{b_i}
\end{equation}

where $\ket{a_i} = (\mathbb{1}^X \otimes \bra{i}^Y) \ket{a}$ and $\ket{b_i} = (\bra{i}^Y \otimes \mathbb{1}^Z) \ket{b}$.

\end{defi}

Let us make a few remarks about the above definition. If the shared space $\C{H}^Y$ is trivial, the link product simplifies to the tensor product $\ket{a}^X * \ket{b}^Z = \ket{a}^X \otimes \ket{b}^Z$. If the shared space $\C{H}^Y$ is the only non-trivial space, the link product simplifies to the inner product, $\ket{a}^Y * \ket{b}^Y = \braket{\bar{a}|b}$ where the entries of $\ket{\bar{a}}$ are the complex conjugated entries of $\ket{a}$ \cite{Wechs_2021}. 

\begin{defi}[Link product for matrices \cite{Chiribella_2008, Chiribella_2009}]\label{def:linkmat}
Let $\C{H}^X, \C{H}^Y, \C{H}^Z$ be non-overlapping Hilbert spaces. For two operators $A \in \C{L}(\C{H}^{XY}), B \in \C{L}(\C{H}^{YZ})$ we define their link product

\begin{gather}
\begin{aligned}
    A * B &\coloneqq \tr_Y ((A^{T_Y} \otimes \mathbb{1}^Z)(\mathbb{1}^X \otimes B)) \\
    &= \sum_{ij} A_{ij} \otimes B_{ij}
\end{aligned}
\end{gather}
where $T_Y$ denotes the partial transpose on $\C{H}^Y$ and $A_{ij} = (\mathbb{1}^X \otimes \bra{i}^Y)A(\mathbb{1}^X \otimes \ket{j}^Y), B_{ij} = (\bra{i}^Y \otimes \mathbb{1}^Z )A(\ket{j}^Y \otimes \mathbb{1}^Z)$.
\end{defi}

Once again, let us consider how this simplifies for trivial Hilbert spaces. If $\C{H}^Y$ is trivial, we obtain the tensor product $A*B = A \otimes B$. If $\C{H}^{XZ}$ is trivial, the link product becomes the trace $A*B = \tr(A^T B)$ \cite{Chiribella_2008, Chiribella_2009}. 

The link product is commutative for both vectors and operators (up to reordering of the resulting tensor product). The $n$-fold link product over vectors $\ket{a_k} \in \C{H}^{A_k}$ or operators $M_k \in \C{L}(\C{H}^{A_k})$ is also associative provided that each Hilbert space appears at most twice, i.e. $\C{H}^{A_k} \cap \C{H}^{A_l} \cap \C{H}^{A_m} = \emptyset$ for all $k \neq l \neq m \neq k$. The link product of hermitian and/or positive semi-definite maps is hermitian and/or positive semi-definite again \cite{Chiribella_2008, Chiribella_2009}.

Let us now discuss the reasons for introducing the link product. We will often want to calculate the Choi vector of a linear operator of the form $V = (\mathbb{1}^{X'} \otimes V_2)(V_1 \otimes \mathbb{1}^{Z})$, where $V_1: \C{H}^{X} \rightarrow \C{H}^{X'Y}$ and $V_2: \C{H}^{YZ} \rightarrow \C{H}^{Z'}$. It then holds that \cite{Wechs_2021}

\begin{equation}
    \dket{V} = \dket{V_1} * \dket{V_2}.
\end{equation}

A similar statement can be made for the Choi matrix of a map $\C{M} = (\C{I}^{X'} \otimes \C{M}_2)(\C{M}_1 \otimes \C{I}^{Z})$ with $\C{M}_1: \C{L}(\C{H}^X) \rightarrow \C{L}(\C{H}^{X'Y})$ and $\C{M}_2: \C{L}(\C{H}^{YZ}) \rightarrow \C{L}(\C{H}^{Z'})$. In this case we have \cite{Chiribella_2008, Chiribella_2009}

\begin{equation}
    M = M_1 * M_2.
\end{equation}

\begin{figure}
    \centering
    \includegraphics[width=0.5\textwidth]{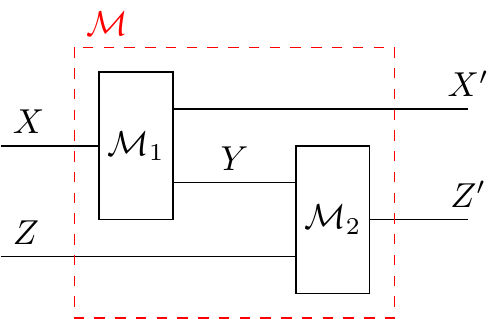}
    \caption{Diagrammatic representation of the action of the link product. The link product turns the circuits defined by the two maps $\C{M}_1, \C{M}_2$ into a new circuit $\C{M}$ from $\C{H}^{XZ}$ to $\C{H}^{X'Z'}$. The wire $\C{H}^Y$ becomes an internal wire.}
    \label{fig:link_circuit}
\end{figure}

Diagrammatically, the action of the link product can be represented by \cref{fig:link_circuit}. The circuits corresponding to the two maps are linked by connecting an output wire of one to the input wire of the other. These two wires correspond to the shared space $\C{H}^Y$.

\section{Process matrices}\label{sec:pm}
\subsection{The general process matrix framework}\label{sec:generalpm}

Standard quantum theory is not well suited to describing processes that lack a definite acyclic causal order (which are often also referred to as processes with indefinite causal order in the literature (cf. for example \cite{Oreshkov_2012, Wechs_2021, Barrett_2021})). The process matrix framework \cite{Oreshkov_2012} attempts to rectify that. The idea is to formulate a quantum theory without a notion of background spacetime. This means that, while we assume that regular quantum theory holds locally for what we will call local quantum laboratories or local agents, there is no fixed global ordering between these agents.


We consider these laboratories to be separate in the sense that a laboratory is isolated from all others and no signaling can occur while it is carrying out its experiment. As we assume that local quantum theory holds in these labs, the agents inside them can carry out any operation consistent with that. This can be either a unitary evolution of the state they receive, or a measurement carried out on this state. Both types of operation can be described mathematically with quantum instruments \cite{davies1970operational}. We can formalize this idea in the following definition:


\begin{defi}[Local agents and local operations \cite{Ara_jo_2015}]\label{def:locallabs}
A local quantum laboratory or local agent consists of an input Hilbert space $\C{H}^{A^I}$ and an output Hilbert space $\C{H}^{A^O}$. The local agent applies a local operation which is a quantum instrument. A quantum instrument is a set of completely positive (CP) maps $\C{M}_A^a = \{\C{M}_A^{x|a}\}_x$ with $\C{M}_A^{x|a}: \C{L}(\C{H}^{A^I}) \rightarrow \C{L}(\C{H}^{A^O})$ labeled by the measurement setting $a$ and the measurement outcome $x$ such that $\sum_x \C{M}_A^{x|a}$ is a completely positive trace-preserving (CPTP) map.\footnote{We could also define quantum instruments without the measurement setting $a$. Instead of choosing a measurement setting, the local agent then chooses a quantum instrument directly. The two formulations are equivalent. Picking a measurement setting $a$ out of all the possible measurement settings is equivalent to picking a quantum instrument out of all possible quantum instruments.}
\end{defi}

We will also refer to the input and output Hilbert spaces as input and output wires.

In the end, we will be interested in probabilities to obtain certain outcomes or also final states. Given a specific CP map $\C{M}_A^{x|a}$ that is part of some quantum instrument, the probability that it is realized is independent of the remaining CP maps making up the quantum instrument \cite{Oreshkov_2012, Wechs_2021}. As such, it is enough to consider a single generic CP map $\C{M}_A$ when setting up the framework. In other words, we can consider each possible measurement outcome separately. The quantum instrument can then be recovered by plugging in each $\C{M}_A^{x|a}$ of the quantum instrument for the generic CP map $\C{M}_A$ and calculating all the probabilities and/or final states separately. In light of this, we will use the term ``local operation" to refer to the generic CP map $\C{M}_A$ in addition to quantum instruments.




The process matrix then describes the outside environment that connects these local agents. It contains all the information on which agents can signal which other agents and under what circumstances. Let us try to understand how this works by considering a simple example. Alice (who applies an operation $\C{M}_A: \C{L}(\C{H}^{A^I}) \rightarrow \C{L}(\C{H}^{A^O})$) and Bob ($\C{M}_B: \C{L}(\C{H}^{B^I}) \rightarrow \C{L}(\C{H}^{B^O})$) are two local agents. For simplicity, we take their input and output spaces to be isomorphic. Alice receives a quantum state $\rho$ which is sent to Bob after Alice has carried out her measurement. The state at the end conditioned on the measurement outcomes corresponding to $\C{M}_A$ and $\C{M}_B$ can be described by

\begin{equation}
    \C{M}_B \circ \C{I}^{A^O \rightarrow B^I} \circ \C{M}_A (\rho)
\end{equation}

where $\C{I}^{A^O \rightarrow B^I}$ is the identity from $\C{L}(\C{H}^{A^O})$ to $\C{L}(\C{H}^{B^I})$ with respect to the computational bases of each Hilbert space. Taking the Choi isomorphism of the above with the help of the link product, one obtains

\begin{equation}
    M_B * \mathbb{1}^{A^O B^I} * M_A * \rho.
\end{equation}

Using commutativity of the link product, we can put everything corresponding to the local agents on one side and everything else (the outside environment) on the other, yielding

\begin{equation}\label{eq:pm_simple_example}
    (M_A \otimes M_B) * (\mathbb{1}^{A^O B^I} \otimes \rho).
\end{equation}

The process matrix is then $\mathbb{1}^{A^O B^I} \otimes \rho$ for this familiar setup. Note that \cref{eq:pm_simple_example} can also be viewed as the definition of a quantum supermap \cite{Chiribella_2008_Supermaps}. Quantum supermaps are maps that take CP maps as arguments and take them to some new CP map. In our case, the CP maps $\C{M}_A$ and $\C{M}_B$ are taken to the map whose Choi matrix is given by \cref{eq:pm_simple_example}. The new map has trivial input, and its output lies in $\C{H}^{B^O}$. 

The full process matrix framework then entails allowing a more general set of matrices in \cref{eq:pm_simple_example} than would be possible under standard quantum mechanics. 

Before formalizing the general process matrix framework, let us make a small modification to the above example. Note that the final state lies in $\C{H}^{B^O}$. This makes sense for the example as Bob is always the last person who acts on the quantum state. However, in more general examples without a definite order there might be multiple agents who could act last. We thus define the Hilbert space of the global future $\C{H}^F$ \cite{Wechs_2021} and demand that the final state is an element of that space (or when considering density matrices, of $\C{L}(\C{H}^F)$). The process matrix from our example would then be $\mathbb{1}^{B^O F} \otimes \mathbb{1}^{A^O B^I} \otimes \rho$. We can also define a Hilbert space of the global past $\C{H}^P$ \cite{Wechs_2021} which prepares the quantum state before sending it to the agents. We can, however, always take $\C{H}^P$ to be trivial (which corresponds to a fixed prepared state). The global future can also be taken to be trivial if one is not interested in the final state but, for example, only the probabilities of the classical measurement results. The global past and future can also be viewed as local quantum laboratories with trivial input respectively output.

With the definitions up to this point and some intuition in hand, we formalize the process matrix framework.

\begin{figure}
\centering
\includegraphics[width=0.5\textwidth]{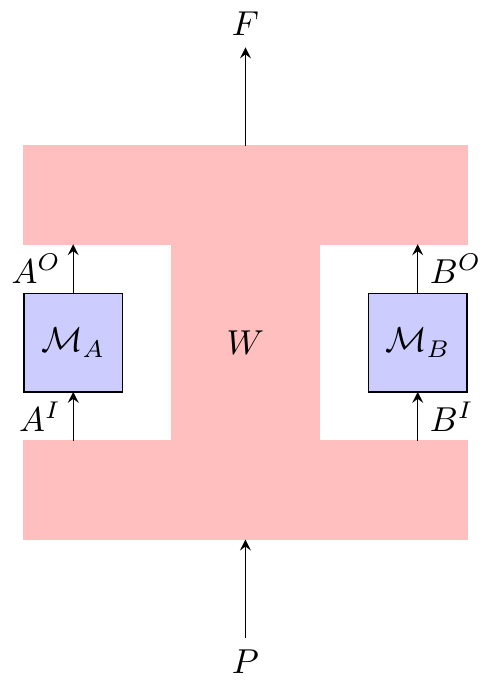}
\caption{An illustration of the process matrix framework in the bipartite case. Local operations (in blue) can be slotted in leading to an induced map from the global past $P$ to the global future $F$ depending on both the local operations and the ``outside environment" described by the process matrix $W$ (here in pink). The ordering of the local operations is not necessarily pre-defined.}
\label{fig:fig_general_W}
\end{figure}

\begin{defi}[Process matrices \cite{Oreshkov_2012}]
For $k=1,...,N$ let $\C{H}^{A^I_k}$ be the input wire and $\C{H}^{A^O_k}$ the output wire of a local agent $A_k$. The operator $W \in \C{L}(\C{H}^{P A^{IO}_{\C{N}} F})$, where $A^{IO}_{\C{N}} = A^{IO}_1 ... A^{IO}_N$ and $A^{IO}_k = A^I_k A^O_k$, is a process matrix if \cite{Ara_jo_2015, Ara_jo_2017, Wechs_2021}

\begin{gather}\label{eq:pm_cond}
\begin{aligned}
    W &\geq 0 \\
    \tr W &= d_P \prod_k d^O_k \\
    W &= L_V(W)
\end{aligned}
\end{gather}
where $L_V(W) = {}_{1 - F\prod_k (1 - A^O_k + A^{IO}_k) + PF A^{IO}_{\C{N}}} W$ and ${}_X W = \frac{1}{d_X} \mathbb{1}^X \otimes \tr_X W$.

The process matrix $W$ defines a quantum supermap \cite{Ara_jo_2017} taking local operations $\C{M}_{A_1},...,\C{M}_{A_N}$ belonging to the above local quantum laboratories to a map $\C{L}(\C{H}^P) \rightarrow \C{L}(\C{H}^F)$. The action of the supermap can be expressed in terms of the Choi matrix of the resulting map as

\begin{equation}\label{eq:supermap}
    W(\C{M}_{A_1},...,\C{M}_{A_N}) \coloneqq (M_{A_1} \otimes ... \otimes M_{A_N}) * W \in \C{L}(\C{H}^{PF}).
\end{equation}

The probability to observe an outcome corresponding to the set of local operations $\C{M}_{A_1},...,\C{M}_{A_N}$ if the input state coming from the global past is $\rho^P$ is given by the generalized Born rule

\begin{equation}\label{eq:pmprob}
    \tr_{F} (\rho^P \otimes M_{A_1} \otimes ... \otimes M_{A_N}) * W.
\end{equation}

\end{defi}

The conditions in \cref{eq:pm_cond} come from the requirement that the process matrix produces only positive and normalized probabilities for all possible local operations \cite{Oreshkov_2012, Ara_jo_2015}.

Frequently, we will consider process vectors \cite{Ara_jo_2015} instead of process matrices. Just like the process matrix is essentially the Choi matrix of the environment, the process vector can be viewed as the corresponding Choi vector. If the process vector is given by $\ket{w}$, then the process matrix is simply $W = \ket{w} \bra{w}$. Instead of \cref{eq:supermap}, we can then use

\begin{equation}
    (\dket{A_1} \otimes ... \otimes \dket{A_N}) * \ket{w}.
\end{equation}

where $\C{M}_{A_k}(\rho) = A_k \rho A_k^\dagger$.\footnote{Note that we use $A_k$ to refer to both the local agent and the single Kraus operator describing the pure operation that this agent applies. This will allow us to keep equations compact while it should be clear from context whether the agent or the Kraus operator is meant.}

Note that the process vector description requires both the process matrix $W$ and the local operations to be pure. However, it is not always possible to purify a process \cite{Ara_jo_2017}. We will discuss this in more detail in \cref{sec:unitary}.


\subsection{Device-dependent and independent notions of causality}\label{sec:causalinequalities}


The fact that process matrices are not necessarily compatible with some acyclic order of the agents allows them to violate so-called ``causal inequalities'' \cite{Oreshkov_2012}. These are bounds set by those processes that are compatible with some global order. This situation is analogous to Bell inequalities where allowing a more general set of shared states, specifically entangled states, allows agents to violate bounds set in the classical setting \cite{Bell}. 

Before we begin let us explicitly spell out a number of assumptions \cite{Oreshkov_2012, Vilasini_2022} that go into the process matrix framework that we have used implicitly before. These can be viewed as analogous to the various assumptions that rule out loopholes in the Bell scenario which allow for trivial violation of the Bell inequalities. 

\paragraph{Free choice:} The agents are free to choose the measurement setting for their quantum instrument or, equivalently, the agents are free to choose their quantum instrument.

\paragraph{Local order:} The event corresponding to the output of an agent causally precedes the event of the input. 

\paragraph{Closed laboratories:} The agents can only interact with the outside environment via the singular events that correspond to receiving an input and sending an output.

\bigskip

These assumptions are necessary to avoid trivializing any communication task. For example, if Alice and Bob each need to communicate a bit to the other, then without the local order assumption they could simply send their bit at time $t=1$ which the other would receive at some later time, say $t=2$.

We continue by dividing process matrices into several subsets which will give us the analogues to separable and entangled states in the Bell setting. We will focus on the bipartite case with trivial global past and future, for simplicity. The concepts discussed here readily generalize to the $N$-partite case.

The simplest case is given by fixed order processes. As the name suggests, these are processes where the agents act in a fixed order and as such if Alice can signal to Bob, then Bob cannot signal to Alice. Our example of Alice sending a state to Bob that we used to introduce the general process matrix framework is one such fixed order process. More formally, we have the following definition.

\begin{defi}[Fixed-order processes \cite{Ara_jo_2015,Oreshkov_2016}]\label{def:fixed_order}
Consider a bipartite process matrix $W$ with agents $A$ and $B$ who apply quantum instruments $\{\C{M}_A^{x|a}\}_x$ and $\{\C{M}_B^{y|b}\}_y$. Defining the probability for $A$ to obtain the outcome $x$ given setting $a$ and Bob to obtain the outcome $y$ given setting $b$, $P(xy|ab) \coloneqq \tr((M_A^{x|a} \otimes M_B^{y|b})*W)$, we say that $W$ is a fixed order process compatible with the order $A$ before $B$ if the probability for $A$ to obtain $x$ given settings $a$ and $b$

\begin{equation}
    P(x|ab) \coloneqq \sum_y P(xy|ab)
\end{equation}

is independent of $b$. Analogously, we say that $W$ is a fixed order process compatible with the order $B$ before $A$ if the probability for $B$ to obtain $y$ given settings $a$ and $b$

\begin{equation}
    P(y|ab) \coloneqq \sum_x P(xy|ab)
\end{equation}

is independent of $a$.\footnote{Note that instead of considering settings, we could also say that $W$ has a fixed causal order if the probability for $A$ to obtain outcome $x$ is independent of agent $B$'s choice of quantum instrument or vice versa.}

\end{defi}

There is an equivalent condition that uses just the form of the process matrix \cite{Oreshkov_2016}. We say that a process is a fixed order process with $A$ before $B$ and denoting it as $W^{A \prec B}$ if

\begin{equation}
    W^{A \prec B} = W^{A^I A^O B^I} \otimes \mathbb{1}^{B^O}
\end{equation}

where $W^{A^I A^O B^I}$ is a process matrix where the output of $B$ is trivial, $\dim \C{H}^{B^O} = 1$. We analogously define $W^{B \prec A}$.

We could now imagine each of these process matrices, $W^{A \prec B}$ and $W^{B \prec A}$, being realized with some probability. The result is a called a causally separable process.

\begin{defi}[Causal separability \cite{Ara_jo_2015,Oreshkov_2016}]\label{def:causalsep}
We call a  bipartite process matrix $W$ causally separable if 

\begin{equation}
    W = p W^{A \prec B} + (1-p) W^{B \prec A}.
\end{equation}
\end{defi}

Such process matrices can be viewed as the analogue to separable states in the Bell scenario. 

One can then show \cite{Oreshkov_2012} that the probability distributions generated from causally separable process matrices obey certain linear inequalities.\footnote{Alternatively, one could generate these inequalities from just the set of fixed order processes. In this case, the causally non-separable processes would automatically satisfy them as well because the inequalities are linear and causally separable processes are by definition convex mixtures of fixed order processes.} These linear inequalities are the causal inequalities we mentioned in the beginning of this section, and they are obtained from a ``game" similar to the CHSH game \cite{CHSH}, which yields the Bell inequalities.

The probability distributions obeying causal inequalities then give us a device-independent notion of causality instead of the device-dependent one from \cref{def:causalsep}.

\begin{defi}[Causal processes \cite{Oreshkov_2012}]
A bipartite process matrix $W$ is causal if for all choices of local operations

\begin{equation}
    P(xy|ab) = p P^{A \prec B}(xy|ab) + (1-p) P^{B \prec A}(xy|ab)
\end{equation}

where $P^{A \prec B}$ is the probability distribution of a process $W^{A \prec B}$ and analogously for $P^{B \prec A}$.
\end{defi}

The class of process matrices that are not causally separable (usually, they are called causally non-separable process matrices) can then in principle violate the causal inequalities, in which case we call them non-causal \cite{Oreshkov_2012}. This is similar to how entangled states can in principle violate Bell inequalities. However, just like not all entangled states violate Bell inequalities \cite{Werner}, there are causally non-separable process matrices that are still causal \cite{Feix_2016}. The set of causal process matrices is thus strictly larger than the set of causally separable process matrices. 


An open question is whether non-causal processes are actually physically realizable. This question is motivated by the fact that no such process has been physically implemented yet and that there is some theoretical evidence that at least some non-causal processes are non-physical \cite{Ara_jo_2017}.

\subsection{Subset of processes modeled by generalized quantum circuits}\label{sec:circuitpm}

In general, interpreting the causal structures described by process matrices is difficult. As we discussed in the previous section, the framework can model very general scenarios as it does not assume a background spacetime and a long-standing open question is to understand which process matrices can be physically realized and under what assumptions and physical regimes. The framework of quantum circuits with quantum control of causal order (QC-QC), which was developed in \cite{Wechs_2021}, adopts a bottom-up approach to this problem and aims to define a very general class of physically realizable quantum circuits, which can be mapped to a subset of process matrices. As we will see, they can be viewed as circuits in the sense that the local operations can be ``plugged in" with their causal order being coherently or classically controlled. 

A similar framework with similar results was also developed in \cite{Purves_2021} at around the same time as the QC-QC framework. We will, however, focus on QC-QCs in this section.

For the rest of the section, we consider $N$ agents, $A_1,...,A_N$ with agent $A_k$ having an input space $\C{H}^{A^I_k}$ and an output space $\C{H}^{A^O_k}$.

\subsubsection{Quantum circuits with fixed causal order}\label{sec:qcfo}

The simplest QC-QCs are so-called quantum circuits with fixed causal order (QC-FOs), also known as quantum combs \cite{Chiribella_2008, Chiribella_2009}. These are essentially generalizations of our example \cref{eq:pm_simple_example} and the fixed order processes we discussed in \cref{sec:causalinequalities} to $N$ agents. The setup is depicted in \cref{fig:QCFO}. In QC-FOs, an internal operation (a CPTP map $\C{M}_1$) sends the system to the first agent who applies their local operation. Afterwards, another internal operation $\C{M}_2$ sends the system to a second agent and this pattern repeats until all the agents have acted. As the name QC-FO suggests, the order of the agents is fixed. The internal operations may also have ancillary input and output spaces (indicated by the wires $\alpha_i$ in the figure) which act as internal memories.

\begin{figure}
    \centering
    \includegraphics[width=\textwidth]{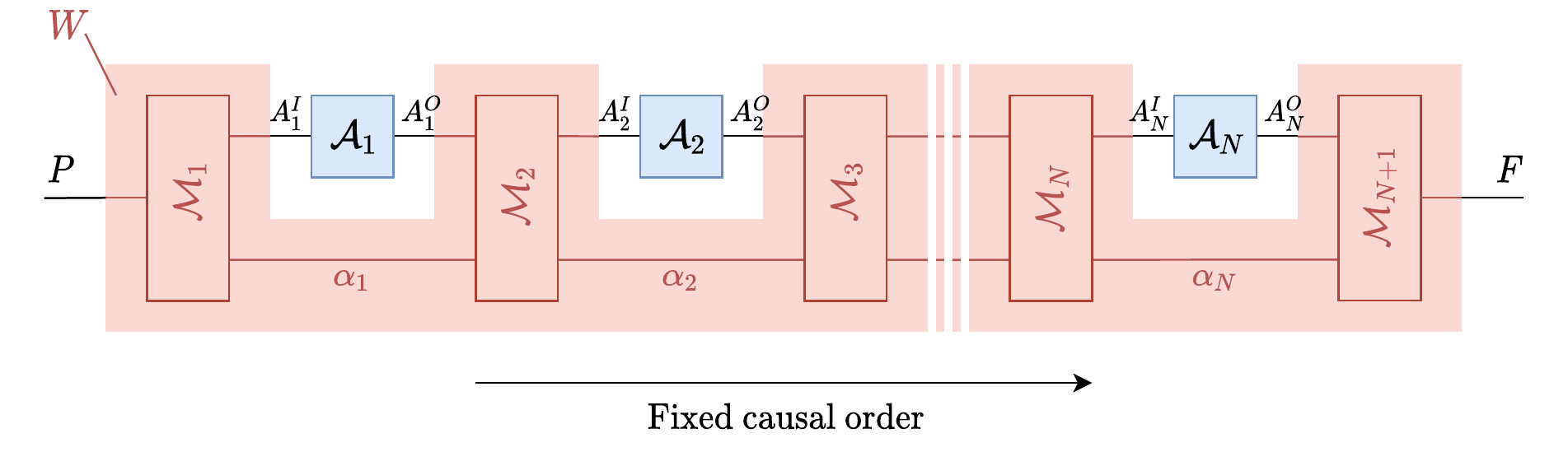}
    \caption{(Fig. 4 \protect \cite{Wechs_2021}) A general depiction of a quantum circuit with fixed causal order (QC-FO), also known as a quantum comb. The agents apply their operations in a fixed order with internal channels $\C{M}_n$ taking the target system from one agent to the next.}
    \label{fig:QCFO}
\end{figure}

\subsubsection{Quantum circuits with classical control of causal order}\label{sec:qccc}

Before going to the general case of QC-QCs, we first describe another class, quantum circuits with classical control of causal order (QC-CC). These already feature many of the properties and ideas that we will again encounter when we formulate the general QC-QC framework, but they are conceptually a bit simpler. A scheme of QC-CCs is depicted in \cref{fig:QCCC}. The order in which the local operations are applied is not pre-defined in a QC-CC but is established dynamically. The idea is that after each local operation the circuit applies a measurement internally. The outcome of this measurement dictates which agent acts next. Which measurement is applied must depend on which agents have already acted so that we can guarantee that each agent acts at most once. We can thus say that if the agents $A_{k_1}, ..., A_{k_n}$ already acted in that order, the circuit applies a quantum instrument $\{\C{M}^{\rightarrow k_{n+1}}_{(k_1,...,k_n)}\}_{k_{n+1} \in \C{N} \backslash \{k_1,...,k_n\}}$ with $\C{M}^{\rightarrow k_{n+1}}_{(k_1,...,k_n)}: \C{L}(\C{H}^{A^O_{k_n} \alpha_n}) \rightarrow \C{L}(\C{H}^{A^I_{k_{n+1}} \alpha_{n+1}})$ where $\C{N} = \{1,...,N\}$ denotes the set of all agents. If all agents acted already the circuit simply sends the system to the global future with a CPTP map $\C{M}^{\rightarrow F}_{(k_1,...,k_N)}: \C{L}(\C{H}^{A^O_{k_N} \alpha_N}) \rightarrow \C{L}(\C{H}^F)$. The causal order is thus still well-defined as it is given by the outcomes of these measurements during each time step even if it is not pre-defined because the outcomes are a priori unknown.

Formally, we can achieve the above by introducing for each time step, denoted by $t_n$ in \cref{fig:QCCC}, $1 \leq n \leq N$, a control system which is an element of a Hilbert space $\C{H}^{C_n}$ with basis states $\ket{(k_1,...,k_n)}^{C_n}$. This control system records which agents have already acted as well as their order, i.e. $\ket{(k_1,...,k_n)}^{C_n}$ means that the agents $A_{k_1},..., A_{k_n}$ already acted in that order. The control acts classically (hence the name QC-CC) so we will adopt the notation

\begin{equation}
    \llbracket (k_1,...,k_n) \rrbracket^{C_n} \coloneqq \proj{(k_1,...,k_n)}{(k_1,...,k_n)}^{C_n}.
\end{equation}

This system then incoherently controls which quantum instrument $\{\C{M}^{\rightarrow k_{n+1}}_{(k_1,...,k_n)}\}_{k_{n+1} \in \C{N} \backslash \{k_1,...,k_n\}}$ is applied as well as which agent acts next. For this purpose, we wish to formally embed the input and output spaces of each agent at a time step $t_n$ in the same Hilbert space. We assume that the input dimensions of all agents are the same, $\dim \C{H}^{A^I_k} = d^I$ for all $k$, and do the same for the output dimensions, $\dim \C{H}^{A^O_k} = d^O$ for all $k$. This can be achieved by introducing additional ancillary systems that can later be traced out. 

The various input/output spaces are now isomorphic to each other and we can introduce for each time step $t_n$ a generic input space $\C{H}^{\tilde{A}^I_n}$ and a generic output space $\C{H}^{\tilde{A}^O_n}$ which are isomorphic to $\C{H}^{A^I_{k_n}}$ and respectively $\C{H}^{A^O_{k_n}}$. 

Using this isomorphism, we can then write the local operations as maps over these generic spaces $\tilde{\C{A}}_{k_n}: \C{L}(\C{H}^{\tilde{A}^I_n}) \rightarrow \C{L}(\C{H}^{\tilde{A}^O_n})$ while the CP maps of the quantum instruments that the circuit applies can be written as $\tilde{\C{M}}^{\rightarrow k_{n+1}}_{(k_1,...,k_n)}: \C{L}(\C{H}^{\tilde{A}^I_n \alpha_n}) \rightarrow \C{L}(\C{H}^{\tilde{A}^O_n \alpha_{n+1}})$. 

We can then embed $\tilde{\C{A}}_{k_n}$ in a conditional operation $\tilde{\C{A}}_n$ such that

\begin{equation}
    \tilde{\C{A}}_n \coloneqq \sum_{(k_1,...,k_n)} \tilde{\C{A}}_{k_n} \otimes \C{P}^{C_n \rightarrow C_n}_{(k_1,...,k_n)}
\end{equation}

where $\C{P}^{C_n \rightarrow C_n}_{(k_1,...,k_n)}$ is the projector on $\llbracket (k_1,...,k_n) \rrbracket^{C_n}$. 

Similarly, for the internal circuit operations, we can write

\begin{equation}
    \tilde{\C{M}}_{n+1} \coloneqq \sum_{(k_1,...,k_n, k_{n+1})} \tilde{\C{M}}^{\rightarrow k_{n+1}}_{(k_1,...,k_n)} \otimes \C{P}^{C_n \rightarrow C_{n+1}}_{(k_1,...,k_n), k_{n+1}}
\end{equation}

where $\C{P}^{C_n \rightarrow C_{n+1}}_{(k_1,...,k_n), k_{n+1}}$ projects on $\llbracket (k_1,...,k_n) \rrbracket^{C_n}$ and then updates the control to $\llbracket (k_1,...,k_n, k_{n+1})\rrbracket^{C_{n+1}}$. Note that the control system allows us to compactly write how the internal operation chooses the right quantum instrument while also allowing us to write the measurement as a CPTP map.

Using the above definitions, their Choi matrices and the link product, we can write the resulting supermap as

\begin{gather}
\begin{aligned}
    M &= \tilde{M}_1 * \tilde{A}_1 * ... * \tilde{A}_N * \tilde{M}_{N+1} \\
    &= \sum_{(k_1,...,k_N)} \tilde{M}^{\rightarrow k_1}_{\emptyset, k_1} * \tilde{A}_{k_1} * ... * \tilde{M}^{\rightarrow k_N}_{(k_1,...,k_{N-1}} * \tilde{A}_{k_N} * \tilde{M}^{F}_{(k_1,...,k_N)} \\
    &\cong (A_1 \otimes ... \otimes A_N) * \sum_{(k_1,...,k_N)} M^{\rightarrow k_1}_{\emptyset, k_1} * ... * \tilde{M}^{\rightarrow k_N}_{(k_1,...,k_{N-1}} * \tilde{M}^{F}_{(k_1,...,k_N)}.
\end{aligned}
\end{gather}

In the second line, we contracted over the control system, using that $\llbracket (k_1,...,k_n)\rrbracket^{C_n} * \llbracket(k'_1,...,k'_n) \rrbracket^{C_n} = \delta_{k_1, k'_1} ... \delta_{k_n, k'_n}$ and in the last line we transformed back to the original input and output spaces, using the isomorphism while also using commutativity of the link product to pull the local operations out of the sum. Note that the different orders $(k_1,...,k_N)$ are in an incoherent superposition, justifying the ``classical" in the name QC-CC. 

From the above, we can readily read off the process matrix

\begin{equation}\label{eq:pmqccc}
    W = \sum_{(k_1,...,k_N)} M^{\rightarrow k_1}_{\emptyset, k_1} * ... * M^{\rightarrow k_N}_{(k_1,...,k_{N-1}} * M^{F}_{(k_1,...,k_N)}.
\end{equation}

From this equation, we can also see that QC-FOs are a subset of QC-CCs. A QC-FO is simply a QC-CC where all the terms in the above process matrix are 0 except for one. 

QC-CCs are causally separable \cite{Wechs_2021}. This can be intuitively understood as a consequence of each term in \cref{eq:pmqccc} corresponding to a definite causal order.

\paragraph{Purifying the operations:}Before we continue with QC-QCs, note that instead of considering CP maps $\C{M}^{\rightarrow k_{n+1}}_{(k_1,...,k_n)}$, we can purify the internal operations to obtain linear operators $V^{\rightarrow k_{n+1}}_{(k_1,...,k_n)}: \C{H}^{A_{k_n} \alpha_n} \rightarrow \C{H}^{A_{k_{n+1}} \alpha_{n+1}}$ such that $\C{M}^{\rightarrow k_{n+1}}_{(k_1,...,k_n)}(\rho) = V^{\rightarrow k_{n+1}}_{(k_1,...,k_n)} \rho V^{\rightarrow k_{n+1} \dagger}_{(k_1,...,k_n)}$. This can be done with the help of the ancillaries $\alpha_n$ as these are essentially arbitrary, except in the case of the final internal operation $\C{M}^{\rightarrow F}_{(k_1,...,k_N)}$ for which an additional wire $\alpha_F$ going to the global future needs to be added. This wire can simply be traced out to recover the original process. On the other hand, while we do not have ancillaries to purify the local operations of the agents, we can recover the general case of multiple Kraus operators by summing up what we obtain for single Kraus operators. We can therefore assume that the action of the local operations consists of applying a single Kraus operator $A_{k_n}$. The result of these purifications is that we can work with process and state vectors instead of process and density matrices. Additionally, it will make it easier to introduce coherent (i.e. quantum) control.

\begin{figure}
    \centering
    \includegraphics[width=\textwidth]{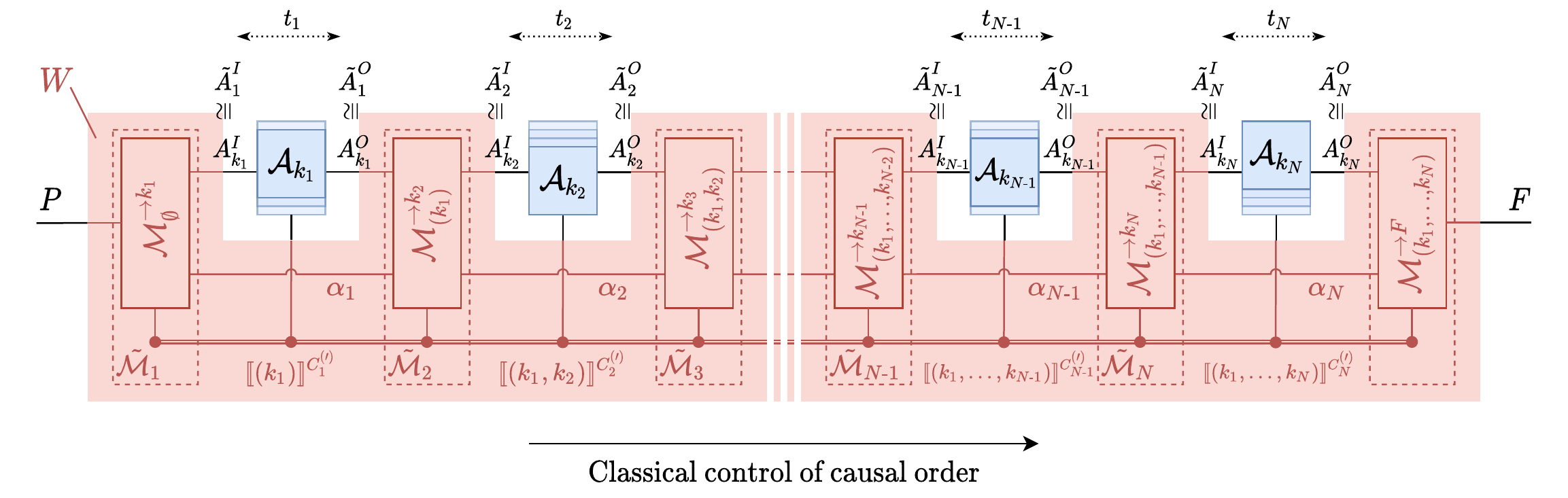}
    \caption{(Fig. 9 \protect \cite{Wechs_2021}) A general depiction of a quantum circuit with classical control of causal order (QC-CC). The causal order is established dynamically by the outcomes $k_n$ of the measurement carried out by the internal circuit operations $\C{M}^{\rightarrow k_{n+1}}_{(k_1,...,k_n)}$. The classical control system $\llbracket(k_1,...,k_n)\rrbracket^{C_n}$ records which agents have already acted and thus ensures that the correct measurement is applied by the circuit during a given time step. The input and output spaces of the agents are formally identified with generic spaces $\C{H}^{\tilde{A}^I_n}$ and $\C{H}^{\tilde{A}^O_n}$ which can be viewed as the input and output space during time step $t_n$.}
    \label{fig:QCCC}
\end{figure}

\subsubsection{Quantum circuits with quantum control of causal order}\label{sec:qcqc}

While QC-CCs control the causal order incoherently, QC-QCs do so coherently. We can achieve this with a small modification to the control system from the previous section. For QC-CCs, the control system recorded the complete order of the operations previously applied to the target system. We will relax this now such that the basis states of the control Hilbert space $\C{H}^{C_n}$ are $\ket{\{k_1,...,k_{n-1}\}, k_n}^{C_n}$. This can be written more compactly as $\ket{\C{K}_{n-1}, k_n}^{C_n}$ where $\C{K}_{n-1} = \{k_1,...,k_{n-1}\}$ will stand generically for some subset of $\C{N}$ with $n-1$ elements. Given the control $\ket{\C{K}_{n-1}, k_n}^{C_n}$, we know that the agent $A_{k_n}$ acted most recently or is about to act and also what agents acted before $A_{k_n}$, namely $A_{k_1}$,..., $A_{k_{n-1}}$ but not their exact order. This allows different orders to coherently interfere which makes the overall order of operations indefinite. At the same time, the fact that the control system does record which agents have already acted ensures that no agent acts more than once.

The control system then coherently controls the application of the local and internal operations, similarly to how the classical control system incoherently controls these operations in QC-CCs. We now take the operations to be linear maps as discussed at the end of the previous section. Additionally, due to the new form of the control system, the internal operations now take the form $V^{\rightarrow k_{n+1}}_{\{k_1,...,k_{n-1}\}, k_n}: \C{H}^{A^O_{k_n} \alpha_n} \rightarrow \C{H}^{A^I_{k_{n+1}} \alpha_{n+1}}$.

We can then embed these operations in the generic Hilbert spaces $\C{H}^{\tilde{A}^I_n}$ and $\C{H}^{\tilde{A}^O_n}$ similar to how we did this for QC-CCs. The general setup is then the following: During the first time step, the circuit applies the operation $\tilde{V}_1$ to the target system coming from the global past

\begin{equation}
    \tilde{V}_1 = \sum_{k_1} \tilde{V}^{\rightarrow k_1}_{\emptyset, \emptyset} \otimes \ket{\emptyset, k_1}
\end{equation}

with $\tilde{V}^{\rightarrow k_1}_{\emptyset, \emptyset}: \C{H}^P \rightarrow \C{H}^{\tilde{A}^I_1} \otimes \C{H}^{\alpha_1}$. We see that it sends something to every lab in a coherent superposition (instead of carrying out a measurement) and appends the appropriate control system.

The local operations then act on the state they receive. We can write the overall action as

\begin{equation}
    \tilde{A}_1 = \sum_{k_1} \tilde{A}_{k_1} \otimes \ket{\emptyset, k_1} \bra{\emptyset, k_1}.
\end{equation}

This pattern then repeats. During time step $t_n$, the internal operation is given by

\begin{equation}\label{eq:defiso}
    \tilde{V}_{n+1} = \sum_{\substack{\C{K}_{n-1},\\ k_n, k_{n+1}}} \tilde{V}^{\rightarrow k_{n+1}}_{\C{K}_{n-1}, k_n} \otimes \ket{\C{K}_{n-1} \cup k_n, k_{n+1}} \bra{\C{K}_{n-1}, k_n}
\end{equation}

and is followed by another application of the local operations

\begin{equation}
    \tilde{A}_n = \sum_{\C{K}_n, k_{n+1}} \tilde{A}_{k_{n+1}} \otimes \ket{\C{K}_n, k_{n+1}} \bra{\C{K}_n, k_{n+1}}.
\end{equation}

The sums in the first of the above equations goes over $\C{K}_{n-1} \subset \C{N}$ and $k_n, k_{n+1} \in \C{N}\backslash \C{K}_{n-1}, k_n \neq k_{n+1}$, while in the second equation it goes over $\C{K}_{n-1} \in \C{N}$ and $k_{n+1} \in \C{N} \backslash \C{K}_n$. We will assume this convention for all such sums.

Finally, after all local agents have acted, there is one last internal operation which sends the state to the global future

\begin{equation}
    \tilde{V}_{N+1} = \sum_{k_N} \tilde{V}^{\rightarrow F}_{\C{K}_{N-1}, k_N} \otimes \bra{\C{N} \backslash k_N, k_N}.
\end{equation}


We can now write down the map $V: \C{H}^P \rightarrow \C{H}^{F \alpha_F}$ which we obtain from the above procedure in terms of its Choi vector

\begin{gather}
\begin{aligned}
    \dket{V} &= \dket{\tilde{V}_1} * \dket{\tilde{A}_1} * ... * \dket{\tilde{V}_N} * \dket{\tilde{A}_N} * \dket{\tilde{V}_{N+1}} \\
    &= \sum_{(k_1,..., k_N)} \dket{\tilde{V}^{\rightarrow k_1}_{\emptyset, \emptyset}} * \dket{\tilde{A}_{k_1}} * ... * \dket{\tilde{V}^{\rightarrow k_N}_{\{k_1,...,k_{N-2}\}, k_{N-1}}} * \dket{\tilde{A}_{k_N}} * \dket{\tilde{V}^{\rightarrow F}_{\C{N}\backslash k_N, k_N}} \\
    &\cong \sum_{(k_1,..., k_N)} (\dket{A_1} \otimes ... \otimes \dket{A_N}) * \ket{w_{(k_1,...,k_N, F)}} \\
    &= (\dket{A_1} \otimes ... \otimes \dket{A_N}) * \ket{w_{\C{N}, F}}.
\end{aligned}
\end{gather}

To obtain the second equality, we contracted over the control system. The sum goes over all possible orders of the agents. In the third equality, we used commutativity of the link product and used the isomorphism between the generic input/output spaces and the local agents' input/output spaces. Finally, we introduced the vectors $\ket{w_{(k_1,...,k_N, F)}} = \dket{V^{\rightarrow k_1}_{\emptyset, \emptyset}} * ... * \dket{V^{\rightarrow F}_{\C{N}\backslash k_N, k_N}}$ and the overall process vector $\ket{w_{\C{N}, F}} = \sum_{(k_1,..., k_N)} \ket{w_{(k_1,...,k_N, F)}}$. The process matrix can then be written as

\begin{equation}
    W = \tr_{\alpha_F} \ket{w_{\C{N}, F}} \bra{w_{\C{N}, F}}.
\end{equation}

The requirement that the internal operations $V_1,..., V_{N+1}$ are isometries yields a useful characterization of QC-QCs.

\begin{defi}[Characterization of QC-QCs \cite{Wechs_2021}]\label{def:qcqcchar}
The operator $W \in \C{L}(\C{H}^{P A^{IO}_{\C{N}} F})$ is the process matrix of a QC-QC if and only if there exist positive semidefinite matrices $W_{(\C{K}_{n-1}, k_n)} \in \C{L}(\C{H}^{P A^{IO}_{\C{K}_{n-1}} A^I_{k_n}})$ for all $\C{K}_{n-1} \subsetneq \C{N}$ and $k_n \in \C{N} \backslash \C{K}_{n-1}$ that fulfill

\begin{gather}\label{eq:qcqccond}
\begin{aligned}
& \sum_{k_1 \in \C{N}} \tr_{A_{k_1}^I} W_{(\emptyset,k_1)} = \mathbb{1}^P \\
& \forall \emptyset \subsetneq \C{K}_n \subsetneq \C{N}, \sum_{k_{n+1} \in \C{N} \backslash \C{K}_n} \tr_{A_{k_{n+1}}^I} W_{(\C{K}_n,k_{n+1})} = \sum_{k_n \in \C{K}_n} W_{(\C{K}_n \backslash k_n,k_n)}\otimes \mathbb{1}^{A_{k_n}^O}, \\
& \textup{and} \tr_F W = \sum_{k_N \in \C{N}} W_{(\C{N} \backslash k_N,k_N)}\otimes \mathbb{1}^{A_{k_N}^O}.
\end{aligned}
\end{gather}

\end{defi}

The positive semidefinite matrices can be written as $W_{(\C{K}_n,k_{n+1})} = \tr_{\alpha_n} \ket{w_{(\C{K}_{n-1}, k_n}} \bra{w_{(\C{K}_{n-1}, k_n}}$ with $\ket{w_{(\C{K}_{n-1}, k_n}} = \sum_{(k_1,...,k_{n-1})} \dket{V^{\rightarrow k_1}_{\emptyset, \emptyset}} * ... * \dket{V^{\rightarrow k_n}_{\{k_1,...,k_{n-2}\}, k_{n-1}}}$ where the sum goes over all orders of $\C{K}_{n-1}$.

Finally, we note that QC-QCs cannot violate causal inequalities despite the fact that they are not causally separable in general \cite{Wechs_2021}.

\begin{figure}
    \centering
    \includegraphics[width=\textwidth]{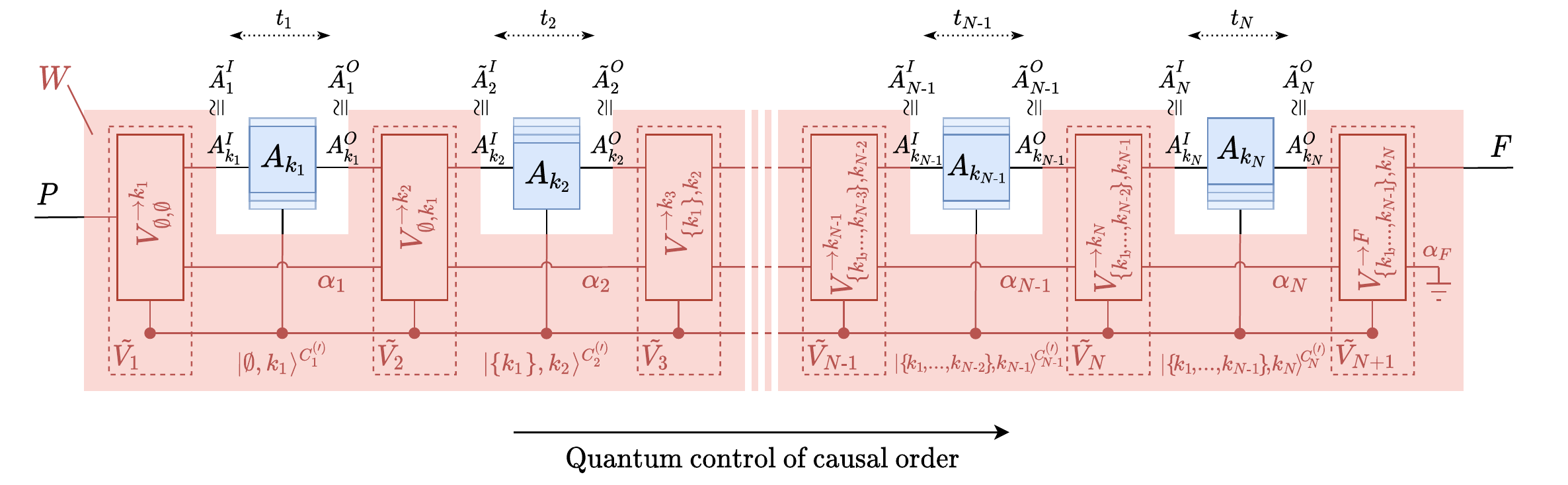}
    \caption{(Fig. 10 \protect \cite{Wechs_2021}) A general depiction of a quantum circuit with quantum control of causal order (QC-QC). The control system $\ket{\{k_1,...,k_{n-1}\}, k_n}$ now only records which agents have already acted but not their order and which agent is about to act or just acted. The causal order is now established coherently via linear operators $V^{\rightarrow k_{n+1}}_{\{k_1,...,k_{n-1}\}, k_n}$ instead of the measurement used in QC-CCs.}
    \label{fig:QCQC}
\end{figure}

\subsection{An example of a QC-QC with dynamical causal order}\label{sec:dynamicalswitch}

In this section, we will discuss a process, the dynamical switch, that was first described in \cite{Wechs_2021} as an illustration of their QC-QC framework. Unlike the frequently discussed standard quantum switch, the dynamical switch features dynamical control of causal order. This means that the causal order depends not just on some fixed control bit but on the outputs of the local agents. Additionally, the dynamical switch is not causally separable and as such it cannot simply be viewed as the superposition of processes with well-defined causal orders.

\begin{figure}
\centering
\includegraphics[width=0.8\textwidth]{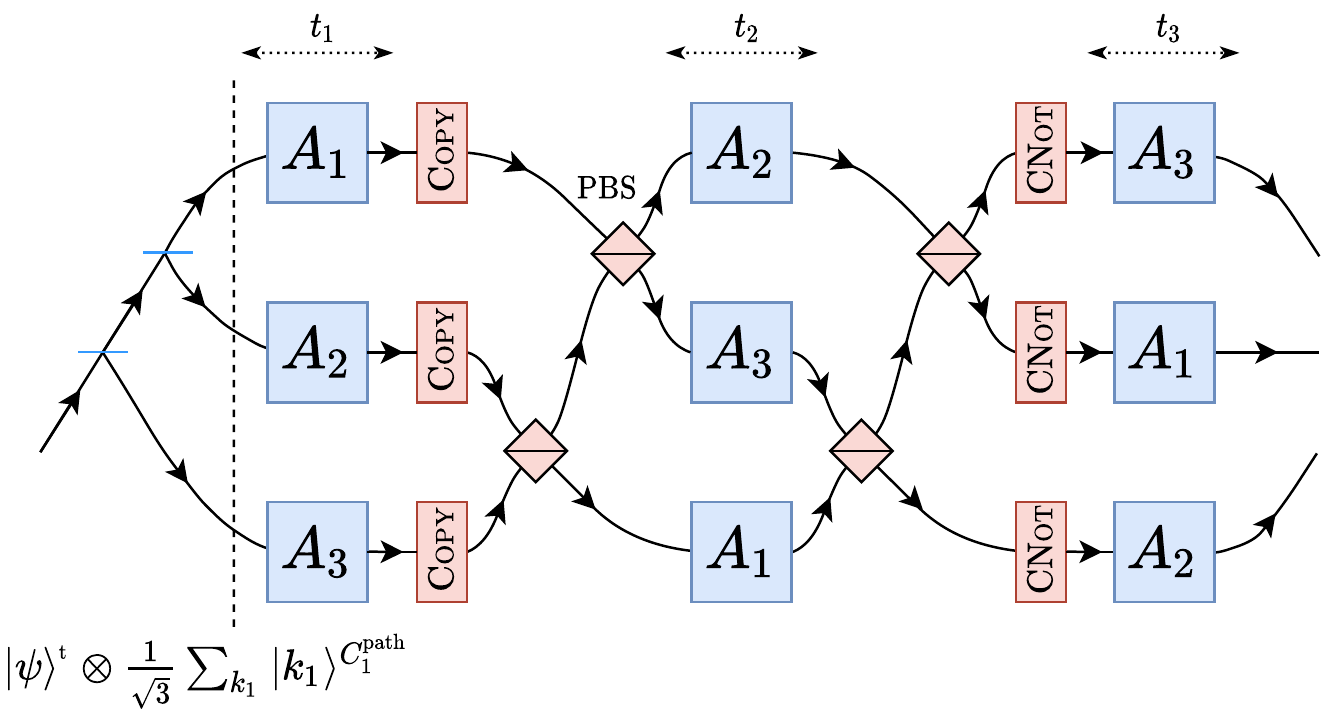}
\caption{(Fig. 12 \protect \cite{Wechs_2021}) Schematic of an experimental implementation of the dynamical switch using photons. Initially, the photon is sent to all agents in superposition with the path acting as the control bit. In the next step, the state of the target is copied onto the polarization via the COPY gates. This then allows the polarizing beam splitters (PBS) to guide the photon to the correct agent. Finally, PBS along with CNOT gates acting jointly on the target and the polarization implement the internal operation $V_3$.}
\label{fig:new_QCQC}
\end{figure}

The dynamical switch is a process between three agents. The target system under consideration is simply a qubit with computational basis $\ket{0}, \ket{1}$ and therefore the input and output spaces of the agents are $\C{H}^{A^{I/O}_n} = \mathbb{C}^2$ for all $n$. The global past and future are taken to be trivial, $\dim \C{H}^P = \dim \C{H}^F = 1$. The process is then defined in the QC-QC framework via the operators

\begin{gather}\label{eq:switchkraus}
\begin{aligned}
    &V^{\rightarrow k_1}_{\emptyset, \emptyset} = \frac{1}{\sqrt{3}} \ket{\psi}^{A^I_{k_1}} \\
    &V^{\rightarrow k_2}_{\emptyset, k_1} = \begin{cases}  \ket{0}^{A^I_{k_2}} \bra{0}^{A^O_{k_1}}, k_2 = k_1 + 1 \pmod 3 \\ \ket{1}^{A^I_{k_2}} \bra{1}^{A^O_{k_1}}, k_2 = k_1 + 2 \pmod 3 \end{cases} \\
    &V^{\rightarrow k_3}_{\{k_1\}, k_2} = \begin{cases}  \ket{0}^{A^I_{k_3}} \ket{0}^{\alpha_3} \bra{0}^{A^O_{k_2}} + \ket{1}^{A^I_{k_3}} \ket{1}^{\alpha_3} \bra{1}^{A^O_{k_2}}, k_2 = k_1 + 1 \pmod 3 \\ \ket{0}^{A^I_{k_3}} \ket{1}^{\alpha_3} \bra{0}^{A^O_{k_2}} + \ket{1}^{A^I_{k_3}} \ket{0}^{\alpha_3} \bra{1}^{A^O_{k_2}}, k_2 = k_1 + 2 \pmod 3 \end{cases}\\
    &V^{\rightarrow F}_{\{k_1, k_2\}, k_3} = \mathbb{1}^{A^O_{k_3} \alpha_3 \rightarrow \alpha_F^{(1)}} \otimes \ket{k_3}^{\alpha^{(2)}_F}.
\end{aligned}
\end{gather}

We introduced several ancillary systems, $\alpha_3$ which is two-dimensional, $\alpha_F^{(1)}$ which is four-dimensional and $\alpha_F^{(2)}$ which is three-dimensional. Additionally, we define $\C{H}^{\alpha_F} = \C{H}^{\alpha_F^{(1)} \alpha_F^{(2)}}$. 

The action of this QC-QC on the target system can be viewed as follows: the operations $V^{\rightarrow k_1}_{\emptyset, \emptyset}$ prepare some state $\ket{\psi}$ and send it to the agent $A_{k_1}$. Next, the operations $V^{\rightarrow k_2}_{\emptyset, k_1}$ send the output of that agent coherently to one of the remaining two agents. It sends the component in the state $\ket{0}$ to $A_{k_1 + 1}$ and the component in the state $\ket{1}$ to $A_{k_1 + 2}$. The operations $V^{\rightarrow k_3}_{\{k_1\}, k_2}$ then send the state to $A_{k_3}$, attaching an ancillary state $\ket{0}^{\alpha_3}$ if $k_2 = k_1 + 1$ or $\ket{1}^{\alpha_3}$ if $k_2 = k_1 + 2$ and applying a CNOT gate to the ancillary with the target system acting as the control. Finally, the operations $V^{\rightarrow F}_{\{k_1, k_2\}, k_3}$ send the output of $A_{k_3}$ along with $\ket{k_3}$ to the ancillary $\alpha_F$. Both the second and third step in this process represent dynamical control of the causal order.

The process matrix and process vector of this QC-QC are

\begin{gather}
\begin{aligned}
    W &= \tr_{\alpha_F} \ket{w} \bra{w} \\
    \ket{w} &= \sum_{(k_1, k_2, k_3)} \dket{V^{\rightarrow k_1}_{\emptyset, \emptyset}} * \dket{V^{\rightarrow k_2}_{\emptyset, k_1}} * \dket{V^{\rightarrow k_3}_{\{k_1\}, k_2}} * \dket{V^{\rightarrow F}_{\{k_1, k_2\}, k_3}}
\end{aligned}
\end{gather}

A possible experimental implementation of this QC-QC using photons is depicted in \cref{fig:new_QCQC}. As all input and output spaces are two-dimensional, the target system can be taken to belong to some two-dimensional Hilbert space $\C{H}^t$ which describes some quantum state of the photon. The controls $C_1$ and $C_3$ can be encoded via the path the photon takes with $\ket{\emptyset, k_1}^{C_1} = \ket{k_1}^{C_1^{\text{path}}}$ and $\ket{\{k_1, k_2\}, k_3}^{C_3} = \ket{k_3}^{C_3^{\text{path}}}$. On the other hand, $C_2$ is completely determined if we know $k_2$ and whether $k_2 = k_1 + 1$ or $k_2 = k_1 + 2$ and so it can be encoded in the path and some two-dimensional system, $\ket{\{k_1\}, k_2}^{C_2} = \ket{\alpha}^{\alpha} \otimes \ket{k_2}^{C_2^{\text{path}}}$. We can encode $\alpha$ in the polarization degree of freedom of the photon and take $\ket{\alpha}^{\alpha} = \ket{0}^\alpha$ if $k_2 = k_1 + 1$ and $\ket{\alpha}^{\alpha} = \ket{1}^\alpha$ if $k_2 = k_1 + 2$.

The isometry $V_2$ is then achieved as depicted in \cref{fig:new_QCQC}. First a COPY gate $V_{\text{COPY}} = \sum_{i=0,1} \ket{i}^{t} \ket{i}^{\alpha} \bra{i}^t $ is applied which copies the state of the target onto the polarization, contributing to the control system $C_2$. Polarizing beam splitters (PBS) then use the polarization to send the system to the correct agent $A_{k_2}$. Meanwhile, $V_3$ is implemented by first applying PBS as depicted in the figure and then a CNOT gate $V_{\text{CNOT}} = \sum_{ij} \ket{i}^t \ket{i \oplus j}^\alpha \bra{i}^t \bra{j}^\alpha$ and the system $\alpha$ switches from being part of the control to being the ancillary $\alpha_3$.

\section{Causal boxes}\label{sec:cb}

The causal box framework \cite{Portmann_2017} models information processing protocols satisfying relativistic causality in a fixed background spacetime (such as Minkowski spacetime), where quantum messages may be exchanged in superpositions of different spacetime locations. This allows for quantum systems to be delocalized over a fixed background spacetime and can model physical scenarios where quantum operations are in a superposition of being applied at different spacetime locations. The framework is, however, very different from the process matrix or QC-QC approaches as it allows for multiple rounds of information processing, is closed under arbitrary composition and does not divide protocols into local operations coupled with processes. Additionally, it explicitly models sending ``nothing" with a vacuum state $\ket{\Omega}$. The existence of such a state and the possibility of sending superpositions of ``something" and ``nothing", $\alpha \ket{\psi} + \beta \ket{\Omega}$, has physical relevance. For example, the controlled application of an unknown unitary to a target system is not possible without the vacuum state \cite{Friis_2014, Zhou_2011}.

\subsection{Messages and Fock spaces}\label{sec:fock}

Let us now begin to formally define the causal box framework. We begin by defining the space of ordered messages. 

\begin{defi}[Messages \cite{Portmann_2017}]
Let $\C{H}^A$ be a Hilbert space and $\C{T}$ a partially ordered set. A message is a state $\ket{\psi}^A \otimes \ket{t}$ with $\ket{\psi}^A \in \C{H}^A$ and $t \in \C{T}$. The space of a single message is $\C{H}^A \otimes l^2(\C{T})$ where $l^2(\C{T})$ is the sequence space with bounded 2-norm of $\C{T}$. 
\end{defi}

The set $\C{T}$ can be physically interpreted as containing spacetime points (or simply the time information if the spatial positions of the agents are fixed) and we will also call it the set of positions. We see that, unlike in the process matrix framework, we model a message as a pair $\ket{\psi}^A \otimes \ket{t}$ where $\ket{\psi}^A$ essentially contains the content of the message while $\ket{t}$ contains the position information of when and possibly where the message was sent or received. This approach allows us to ensure causality is respected and is therefore also a necessary ingredient for the composability of causal boxes. We will frequently write $\ket{\psi, t}^A$ instead of $\ket{\psi}^A \otimes \ket{t}$.

For the process matrix, we modeled the state space of a wire as a Hilbert space. As mentioned earlier, the causal box framework allows the sending of multiple messages or, in other words, there can be any number of messages on a wire. To capture this, we model the state space of the wire in the causal box framework as the symmetric Fock space of the single message space 

\begin{equation}
    \C{F}(\C{H}^A \otimes l^2(\C{T})) = \bigoplus_{n=0}^\infty \vee^n (\C{H}^A \otimes l^2(\C{T}))
\end{equation}

where $\vee^n (\C{H}^A \otimes l^2(\C{T}))$ is the symmetric subspace of $(\C{H}^A \otimes l^2(\C{T}))^{\otimes n}$. The one-dimensional space $(\C{H}^{A} \otimes l^2(\C{T}))^{\otimes 0}$ corresponds to the vacuum state $\ket{\Omega}$.

The reason we use the symmetric subspace is that all the ordering information is already contained in the position label $t \in \C{T}$. There is thus no difference between the state $\ket{0, t_1} \otimes \ket{1, t_2}$ and $\ket{1, t_2} \otimes \ket{0, t_1}$ for $t_1, t_2 \in \C{T}$.

The above notation is quite cumbersome so we will often abbreviate it by writing $\C{F}^{\C{T}}_{A} = \C{F}(\C{H}^A \otimes l^2(\C{T}))$.

Let us also explicitly define the inner product in a Fock space, as this will make the discussion in the results section and several of the proofs of the propositions and lemmas there easier to follow. In general, the spaces $\vee^n (\C{H}^A \otimes l^2(\C{T}))$ for different $n$ are orthogonal to each other. Physically, this means we can always perfectly distinguish between states with different numbers of messages. We can therefore define the inner product separately for each $n$ which is induced by the inner product over $\vee^n (\C{H}^A \otimes l^2(\C{T}))$ and then linearly extend that to the entire Fock space. In turn, we can consider $\vee^n (\C{H}^A \otimes l^2(\C{T}))$ to be a subspace $(\C{H}^A \otimes l^2(\C{T}))^{\otimes n}$ and simply use the inner product on the general product space restricted to the symmetric subspace. Let us consider how this works out explicitly. A state in the space $\vee^n (\C{H}^A \otimes l^2(\C{T}))$ can be written as

\begin{equation}
    \bigodot_{k=1}^n \ket{\psi_k, t_k} \coloneqq \frac{1}{\sqrt{n!}} \sum_{\pi \in S^n} \bigotimes_{k=1}^n \ket{\psi_{\pi(k)}, t_{\pi(k)}}
\end{equation}

with $\psi_k \in \C{H}$, $t_k \in \C{T}$ and $S^n$ is the set of permutations of $n$ elements. Given another state $\bigodot_{l=1}^n \ket{\phi_l, t'_l} \in \vee^n (\C{H}^A \otimes l^2(\C{T}))$, we can then write their inner product as an inner product over $(\C{H}^A \otimes l^2(\C{T}))^{\otimes n}$

\begin{gather}\label{eq:fock_inner_n}
\begin{aligned}
    \bigodot_{k=1}^n \bra{\psi_k, t_k} \bigodot_{l=1}^n \ket{\phi_l, t'_l} &= \frac{1}{\sqrt{n!}} \frac{1}{\sqrt{n!}} \sum_{\pi \in S^n} \bigotimes_{k=1}^n \bra{\psi_{\pi(k)}, t_{\pi(k)}} \sum_{\pi' \in S^n} \bigotimes_{l=1}^n \ket{\phi_{\pi'(l)}, t'_{\pi'(l)}} \\
    &= \frac{1}{n!} \sum_{\pi, \pi' \in S^n} \prod_{k=1}^n \braket{\psi_{\pi(k)}, t_{\pi(k)}|\phi_{\pi'(k)}, t'_{\pi'(k)}} \\
    &= \sum_{\pi \in S^n} \prod_{k=1}^n \braket{\psi_k, t_k|\phi_{\pi(k)}, t'_{\pi(k)}} \\
    &= \sum_{\pi \in S^n} \prod_{k=1}^n \braket{\psi_k|\phi_{\pi(k)}} \delta_{t_k, t'_{\pi(k)}}.
\end{aligned}
\end{gather}

We can summarize the above in the following definition.

\begin{defi}[Inner product in Fock spaces \cite{bhatia1996matrix}]
Let $\ket{\psi} = a \ket{\Omega} + \sum_{n=1}^\infty \bigodot_{k=1}^n \ket{\psi_{n,k}, t_{n,k}}, \ket{\phi} = b \ket{\Omega} + \sum_{n=1}^\infty \bigodot_{k=1}^n \ket{\phi_{n,k}, t'_{n,k}} \in \C{F}(\C{H}^A \otimes l^2(\C{T}))$. Then their inner product is given by 

\begin{equation}\label{eq:fock_inner}
    \braket{\psi|\phi} = \overline{a} b + \sum_{n=1}^\infty \sum_{\pi \in S^n} \prod_{k=1}^n \braket{\psi_{n,k}|\phi_{n, \pi(k)}} \delta_{t_{n,k}, t'_{n,\pi(k)}}.
\end{equation}

\end{defi}

\paragraph{Wire isomorphisms:} There exist two useful isomorphisms regarding the splitting of wires. For any $\C{H}^A = \C{H}^B \oplus \C{H}^C$ and any $\C{T}' \subseteq \C{T}$, we have

\begin{gather}\label{eq:wireiso}
\begin{aligned}
    \C{F}^{\C{T}}_{A} &\cong \C{F}^{\C{T}}_{B} \otimes \C{F}^{\C{T}}_{C} \\
    \C{F}^{\C{T}}_{A} &\cong \C{F}^{\C{T}'}_{A} \otimes \C{F}^{\C{T} \backslash \C{T}'}_{A}.
\end{aligned}
\end{gather}

The first isomorphism implies that two wires, one carrying messages from one Hilbert space $\C{H}^B$ and the other carrying messages from another Hilbert space $\C{H}^C$, is equivalent to a single wire carrying messages from the direct sum of the two Hilbert spaces. This will allow us to formally define causal boxes with just a single input and a single output wire. 

The second isomorphism states that it does not matter whether there is one wire that carries all messages or a wire for each $t \in \C{T}$.

\begin{example}[Norm of a Fock space state]
Consider a wire which has a superposition of no messages and two messages at time $t$, one in state $\ket{0}$ and the other in state $\ket{1}$, on it such that the probability to obtain either outcome when measuring the number of messages is equal. The overall state on the wire can be written as $\ket{\psi} = \frac{1}{\sqrt{2}} \ket{\Omega} + \frac{1}{\sqrt{2}} \ket{0, t} \odot \ket{1, t} = \frac{1}{\sqrt{2}} \ket{\Omega} + \frac{1}{2} (\ket{0, t} \otimes \ket{1, t} + \ket{1, t} \otimes \ket{0, t})$. We can check that this state is normalized either by using the form of $\ket{\psi}$ where we expanded the symmetric tensor product

\begin{gather}
\begin{aligned}
    \braket{\psi|\psi} &= (\frac{1}{\sqrt{2}} \bra{\Omega} + \frac{1}{2} (\bra{0, t} \otimes \bra{1, t} + \bra{1, t} \otimes \bra{0, t}))(\frac{1}{\sqrt{2}} \ket{\Omega} + \frac{1}{2} (\ket{0, t} \otimes \ket{1, t} + \ket{1, t} \otimes \ket{0, t})) \\
    &= \frac{1}{2} \braket{\Omega|\Omega} + \frac{1}{4} (\braket{0, t|0,t} \braket{1,t|1,t} + \braket{1, t|0,t} \braket{0,t|1,t} + \braket{0, t|1,t} \braket{1,t|0,t} + \braket{1, t|1,t} \braket{0,t|0,t}) \\
    &= \frac{1}{2} + \frac{1}{4} 2 = 1
\end{aligned}
\end{gather}
or by using \cref{eq:fock_inner} with $a = b = \frac{1}{\sqrt{2}}, \ket{\psi_{1, 1}} = \ket{\phi_{1, 1}} = \frac{1}{\sqrt{2}} \ket{0}$ and $\ket{\psi_{1, 2}} = \ket{\phi_{1, 2}} = \frac{1}{\sqrt{2}} \ket{1}$

\begin{gather}
\begin{aligned}
    \braket{\psi|\psi} &= \frac{1}{2} \delta_{t, t}+ \frac{1}{2} (\braket{0|0} \braket{1|1} \delta_{t,t} + \braket{0|1} \braket{1|0} \delta_{t,t} \\
    &= \frac{1}{2}  + \frac{1}{2} = 1.
\end{aligned}
\end{gather}

One can also check that the probabilities to find the state to be the vacuum $\ket{\Omega}$ or $\ket{0,t} \odot \ket{1,t}$ are $\frac{1}{2}$ each as required by taking the square of the inner product between these states and $\ket{\psi}$.

\end{example}

\begin{example}[Fock space inner product of a multi-system state]
Consider two wires $A$ and $B$. Let the state on wire $A$ be $\ket{0, t}^A \odot \ket{1, t}^A$ and the state on wire $B$ be $\ket{0, t}^B \odot \ket{0, t'}^B$ with $t \neq t'$. When writing the overall state it does not matter whether we write the state on $A$ or $B$ first, $\ket{0, t}^A \odot \ket{1, t}^A \otimes \ket{0, t}^B \odot \ket{0, t'}^B = \ket{0, t}^B \odot \ket{0, t'}^B \otimes \ket{0, t}^A \odot \ket{1, t}^A$. We just have to make sure that when we take the inner product of such multi-wire states, we separately calculate the inner product of the states on each wire and then take the product over the wires. For example, consider the inner product of the above state with the state where the wires are exchanged

\begin{gather}
\begin{aligned}
    (\bra{0, t}^A& \odot \bra{1, t}^A \otimes \bra{0, t}^B \odot \bra{0, t'}^B)(\ket{0, t}^B \odot \ket{1, t}^B \otimes \ket{0, t}^A \odot \ket{0, t'}^A) \\
    &= (\bra{0, t}^A \odot \bra{1, t}^A)(\ket{0, t}^A \odot \ket{0, t'}^A)(\bra{0, t}^B \odot \bra{0, t'}^B)(\ket{0, t}^B \odot \ket{1, t}^B) = 0
\end{aligned}
\end{gather}
and not 1 as one might expect if one just looked at the order and not the system labels. It is thus important to be aware what system a message belongs to when considering states over multiple wires. 

\end{example}

\begin{remark}[Different conventions for Fock space states]
When doing calculations and in particular when trying to interpret the inner product as a tool to find the probability of obtaining a certain measurement result, one needs to be a bit careful. This is because the symmetric tensor product of two normalized vectors $\ket{\psi} \odot \ket{\phi}$ in the way we defined it here is not necessarily normalized itself. For example, $\ket{0} \odot \ket{0} = \frac{1}{\sqrt{2}} (\ket{0} \otimes \ket{0} + \ket{0} \otimes \ket{0}) = \sqrt{2} \ket{0} \otimes \ket{0}$. The symmetric tensor product has norm 1 only when the individual tensor factors are orthogonal to each other. We could define the symmetric tensor product $\odot$ in such a way that the norm is always the product of the norms of the individual tensor factors (for example, by normalizing the symmetric tensor products of all basis states). However, this would make the expanded expression of a general symmetric tensor product $\bigodot_{k=1}^n \ket{\psi_n}$ more complicated. It will be much more useful later to keep this expression simple than it is to be able to easily write down a normalized multi-message state. Additionally, note that the distinction is physically meaningless. The expression $\ket{\psi} \odot \ket{\phi}$ is ultimately just a label for an element of the Fock space and as such has no physical meaning except for that which is induced by said element.
\end{remark}

\begin{remark}[Inner product and the permanent of a matrix \cite{bhatia1996matrix}]
The symmetric tensor product in \cref{eq:fock_inner_n} is also the permanent of the matrix $A$ with elements $A_{ij}:=\braket{\psi_i, t_i|\phi_j, t_j}$.
\end{remark}

\subsection{Cuts and causality}\label{sec:causality}

Given what we discussed up until now, a causal box can be viewed as a map from an input wire with Fock space $\C{F}^{\C{T}}_{X}$ to an output wire with Fock space $\C{F}^{\C{T}}_{Y}$. Defining such a map could, however, turn out difficult if the set $\C{T}$ under consideration is infinite. For example, we could consider a map that outputs a fixed state for each $t \in \C{T}$ which would yield an infinite tensor product of this state as the overall output. However, if we consider the output up to a certain point, it is finite and there are no issues. So instead of viewing the causal box as a single map we should view it as a collection of maps, each of which describes the behavior up to a certain point. Additionally, the causal box should respect causality, that is for $t_1, t_2 \in \C{T}$ with $t_1 \leq t_2$ an input at $t_2$ should not influence an output at $t_1$.

In order to formalize these requirements, we introduce the concept of cuts of $\C{T}$ which formalize our notion of ``up to a certain point".

\begin{defi}[Cuts \cite{Portmann_2017}]
A cut of $\C{T}$ is a subset $\C{C} \subseteq \C{T}$ such that $\C{C} = \bigcup_{t \in \C{C}} \C{T}^{\leq t}$ where $\C{T}^{\leq t} = \{p \in \C{T}: p \leq t\}$. A cut $\C{C}$ is bounded if there exists a point $t \in \C{T}$ such that $\C{C} \subseteq \C{T}^{\leq t}$.
The set of all cuts of $\C{T}$ is denoted as $\mathfrak{C}(\C{T})$ and the set of all bounded cuts as $\overline{\mathfrak{C}}(\C{T})$.
\end{defi}

The requirement of causality can be formalized by requiring that for each causal box, there exists a causality function $\chi$ which tells us what set of points in $\C{T}$ can have an influence on the output ``up to a certain point". 

\begin{defi}[Causality function \cite{Portmann_2017}]
A function $\chi: \mathfrak{C}(\C{T}) \rightarrow \mathfrak{C}(\C{T})$ is a causality function if it satisfies the following conditions:

\begin{gather}
\begin{aligned}
    &\forall \C{C}, \C{D} \in \mathfrak{C}(\C{T}), \chi(\C{C} \cup \C{D}) = \chi(\C{C}) \cup \chi(\C{D}) \\
    &\forall \C{C}, \C{D} \in \mathfrak{C}(\C{T}), \C{C} \subseteq \C{D} \implies \chi(\C{C}) \subseteq \chi(\C{D}) \\
    &\forall \C{C} \in \overline{\mathfrak{C}}(\C{T}) \backslash \{\emptyset\}, \chi(\C{C}) \subsetneq \C{C} \\
    &\forall \C{C} \in \overline{\mathfrak{C}}(\C{T}), \forall t \in \C{C}, \exists n \in \mathbb{N}, t \not \in \chi^n(\C{C})
\end{aligned}
\end{gather}
\end{defi}

The first three conditions capture our intuition that only inputs coming before the output can influence the output. The meaning of the last condition is less immediately apparent. This condition, however, is necessary to avoid another case of infinite outputs. Consider a system that outputs $\ket{0}$ at $1-t/2, t \in [0, 1)$ if it receives an input $\ket{0}$ at time $1-t$ and additionally it outputs $\ket{0}$ at $t=0$. If we loop its output back to its input, the system outputs a message at $t=1/2, 3/4, 7/8,...$. However, the outputs in the cut $[0,1]$ then depend on inputs that are arbitrarily close to $t=1$. This means that $\chi([0,1]) = [0,1)$ and also $\chi([0,1)) = [0,1)$. The fourth condition is not fulfilled. Hence, we see that the fourth condition allows us to exclude systems that produce an infinite number of outputs in a finite amount of time.

\subsection{Definition of causal boxes}\label{sec:defcb}

Having defined the message space as well as cuts and the causality function, we can now finally formally define causal boxes.

\begin{defi}[Causal boxes \cite{Portmann_2017}]\label{def:causalbox}
A causal box is a system with an input wire $X$ and an output wire $Y$, defined by a set of mutually consistent, CPTP maps

\begin{equation}
    \Phi = \{\Phi^{\C{C}}: \mathfrak{T}(\C{F}^{\chi(\C{C})}_X) \rightarrow \mathfrak{T}(\C{F}^{\C{C}}_Y)\}_{\C{C} \in \overline{\mathfrak{C}}}
\end{equation}

where $\chi$ is a causality function and $\mathfrak{T}(\C{F})$ denotes the trace class operators over the space $\C{F}$. Furthermore, these maps must fulfill

\begin{gather}\label{eq:cbreqs}
\begin{aligned}
    \Phi^{\C{C}} &= \tr_{\C{D} \backslash \C{C}} \circ \Phi^{\C{D}} \\
    \Phi^{\C{C}} &= \Phi^{\C{C}} \circ \tr_{\C{T} \backslash \chi(\C{C})}
\end{aligned}
\end{gather}
where $\tr_{\C{D} \backslash \C{C}}$ corresponds to tracing out the messages in positions in $\C{D} \backslash \C{C}$ and analogously for $\tr_{\C{T} \backslash \chi(\C{C})}$.

\end{defi}

The first of \cref{eq:cbreqs} states that the various maps making up the causal box must be consistent with each other. The second equation encodes the causality condition, i.e. that we can calculate the outputs in $\C{C}$ from the inputs in $\chi(\C{C})$ which is equivalent to saying that inputs after $\chi(\C{C})$ cannot influence outputs in $\C{C}$.

\subsection{Representations of causal boxes}\label{sec:repcb}

The causal box admits two alternative representations \cite{Portmann_2017}, the Choi representation and the sequence representation. These are often easier to deal with than \cref{def:causalbox}. In particular, we will use the Choi representation to define composition of causal boxes while the sequence representation offers us a convenient way to check whether something is a causal box.

\subsubsection{Choi representation}\label{sec:choicb}

As the maps $\{\Phi^{\C{C}}\}_{\C{C} \in \overline{\mathfrak{C}}}$ are CPTP maps, they admit a Choi representation

\begin{equation}\label{eq:CB_choi}
    R_{\Phi^{\C{C}}} = \sum_{i,j} \proj{i}{j} \Phi^{\C{C}}(\proj{i}{j}).
\end{equation}

with some appropriate basis $\ket{i}$ of the Fock space. The causal box given by the set of CPTP maps $\{\Phi^{\C{C}}\}_{\C{C} \in \overline{\mathfrak{C}}}$ can thus be equivalently represented by the set of Choi matrices of this set, $\{R_{\Phi^{\C{C}}}\}_{\C{C} \in \overline{\mathfrak{C}}}$. However, note that the Fock space is infinite-dimensional which means that the above operator can be unbounded. This can be dealt with by using a more general definition of the Choi representation which simplifies to \cref{eq:CB_choi} whenever \cref{eq:CB_choi} is bounded. We will, however, not discuss this here as we will later only consider causal boxes where we restrict to some finite-dimensional subspace of the Fock space. 

\subsubsection{Sequence representation}\label{sec:seqrep}

The sequence representation is based on the Stinespring dilation \cite{Stinespring_1955} which states that for a CPTP map $\Phi: \C{L}(\C{H}^A) \rightarrow \C{L}(\C{H}^B)$, there exists an isometry $U: \C{H}^A \rightarrow \C{H}^{B \alpha}$ for some Hilbert space $\C{H}^\alpha$ such that $\Phi (\rho^A) = \tr_\alpha (U \rho^A U^\dagger)$. 

This can be adapted to the case of causal boxes \cite{Portmann_2017}. One can find Stinespring representations $\{U^{\C{C}}: \C{F}^{\chi(\C{C})}_X \rightarrow \C{F}^{\C{C}}_Y \otimes \mathcal{H}^{\alpha_\C{C}}\}_{\C{C} \in \mathfrak{C}(\C{T})}$ for the maps making up a causal box where for any $\C{C}, \C{D} \in \overline{\mathfrak{C}}(\C{T}), \C{C} \subseteq \C{D}$ there exists an isometry $V: \C{H}^{\alpha_{\C{C}}} \otimes \C{F}^{\C{D} \backslash \chi(\C{C})}_X \rightarrow \C{F}^{\C{D} \backslash \C{C}}_Y \otimes \C{H}^{\alpha_{\C{D}}}$ such that 

\begin{equation}\label{eq:stinespring}
    U^{\C{D}} = (\mathbb{1}^{\C{C}}_Y \otimes V)(U^{\C{C}} \otimes \mathbb{1}^{\C{D} \backslash \chi(\C{C})}_X)
\end{equation}

where $\mathbb{1}^{\C{C}}_Y$ is the identity on states on wire $Y$ with position in $\C{C}$.

\begin{figure}
    \centering
    \includegraphics[width=\textwidth]{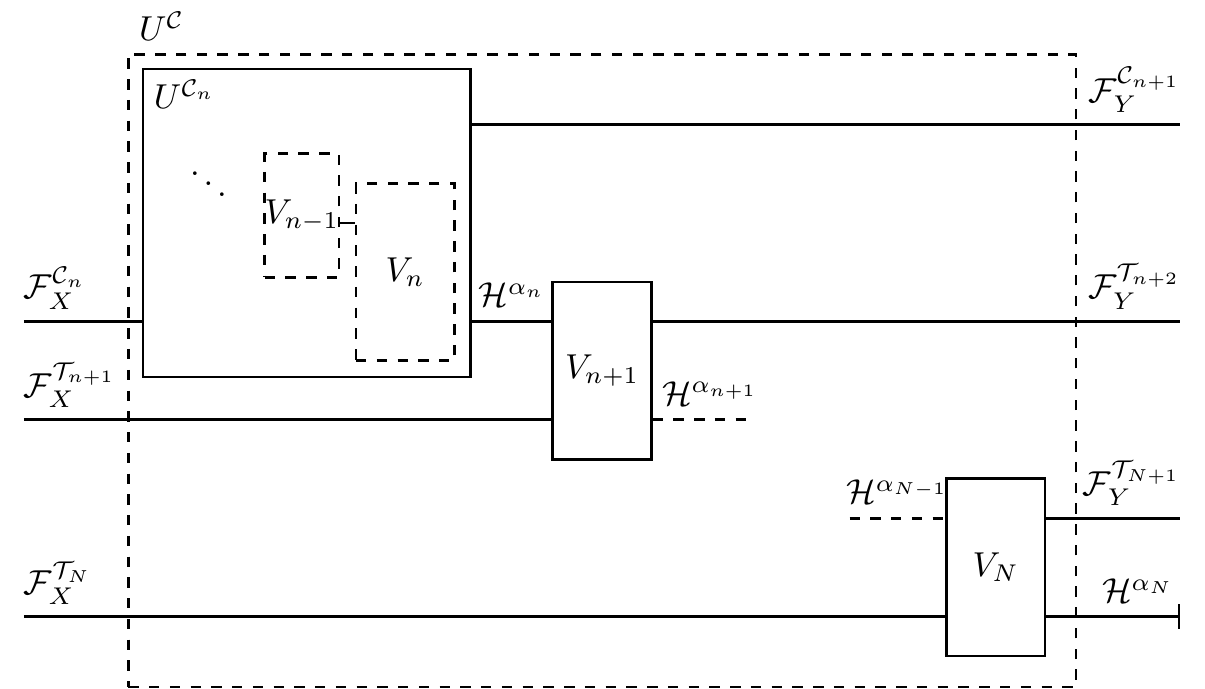}
    \caption{(Fig. 10 \protect \cite{Portmann_2017} \textcopyright 2017 IEEE, with small adaptations) Decomposition of a map into a sequence of isometries $V_n$. A causal box can thus be understood via its action during each time/position slice.}
    \label{fig:sequencerep}
\end{figure}

Note now that a causal box restricted to some cut is again a causal box. Thus, we can find a Stinespring representation of the restriction and plug it into \cref{eq:stinespring}. By doing this recursively as depicted in \cref{fig:sequencerep}, we obtain the sequence representation.

\begin{defi}[Sequence representation \cite{Portmann_2017}]
Let $\C{C}_1 \subseteq ... \subseteq \C{C}_N \subseteq \C{C}_{N+1} = \C{C}$ for $N \in \mathbb{N}$ be a sequence of cuts such that $\bigcap_{n=1}^{N+1} \C{C}_n = \emptyset$ and we define $\C{T}_n := \C{C}_n \backslash \C{C}_{n-1}$ for $n \geq 2$ and $\C{T}_1 = \C{C}_1$. A sequence representation of a causal box is given by such a sequence of cuts and a set of isometries 

\begin{equation}\label{eq:seqisom}
    \{V_n: \C{H}^{\alpha_{n-1}} \otimes \C{F}^{\C{T}_n}_X \rightarrow \C{F}^{\C{T}_{n+1}}_Y \otimes \C{H}^{\alpha_n}\}_{n=1}^N
\end{equation}

such that for all $n \geq 2$

\begin{equation}\label{eq:seqrep}
    U^{\C{C}} = (\mathbb{1}^{\C{C}_{N}}_Y \otimes V_N) ... (\mathbb{1}^{\C{C}_{n+1}}_Y \otimes V_{n+1} \otimes \mathbb{1}^{\C{C}_{N-1} \backslash \C{C}_{n+1}}_X)(U^{\C{C}_n} \otimes \mathbb{1}^{\C{C}_{N-1} \backslash \C{C}_n}_X).
\end{equation}

\end{defi}

The sequence representation can be viewed as a causal unraveling of the causal box. Looking at the set $\C{T}$ as a time parameter, each isometry $V_n$ tells us what the causal box does during some time span.

It is possible that the sequence of cuts in the above definition is not finite. If that is the case, our definition is slightly awkward as there is actually no ``first" isometry $V_1$. There is always another isometry that acts on even earlier inputs.\footnote{In the original definition of the sequence representation in \cite{Portmann_2017} the order of the isometries is reversed, that is $V_1$ denotes the last isometry. We choose the opposite convention here because we later wish to establish a connection between the isometry $V_n$ of the sequence representation and the isometry $V_n$ from the QC-QC framework.} However, we will be mostly concerned with finite and discrete sets $\C{T}$ in which case the sequence of cuts is always finite. 

A very useful fact is that not only does every causal box have a sequence representation, where the set of cuts is given recursively by $\C{C}_{n-1} = \chi(\C{C}_n)$, but also any set of isometries as in \cref{eq:seqisom} is the sequence representation of a causal box \cite{Portmann_2017}. This means if we want to find a causal box description of some process, we can define the action of the process for each time step and as long as this yields isometries, we can be sure that the final result is a causal box.


\subsection{Composition of causal boxes}\label{sec:composition}

We previously mentioned that causal boxes can be composed. In this section, we will discuss how to do this.

Composition is essentially achieved by taking the output wire of one causal box and connecting it to the input wire of another causal box. Subwires can also be used instead of considering all inputs and outputs. The set $\C{T}$ and the causality function $\chi$ then ensure that this is well defined, even if we make loops by also connecting the output wire of the second causal box to the input wire of the first causal box. Note that this picture of composition is also what we intuitively view the action of the link product to be. We will prove this intuition in \cref{sec:statespace}.


Composition of causal boxes can be defined as a two-step process. The first step is called parallel composition \cite{Portmann_2017}. It can be viewed simply as the union of the inputs/outputs of two causal boxes. This is equivalent to just taking the tensor product. The second step is called loop composition. It is an operation on a single CP map $\Phi$ and involves looping an output (sub-)wire of $\Phi$ back to one of its input (sub-)wires. 

We will once again only consider the finite-dimensional case. Loop composition is then given by the following definition.

\begin{defi}[Loop composition \cite{Portmann_2017, }]\label{def:comp}
Consider a CP map $\C{M}: \mathcal{L}(\mathcal{H}^{AB})\rightarrow\mathcal{L}(\mathcal{H}^{CD})$ with input systems $A$ and $B$ and output systems $C$ and $D$ with $\C{H}^B \cong \C{H}^C$. Let $\{\ket{k}^C\}_k$ be any orthonormal basis of $\mathcal{H}^C$, and denote with $\{\ket{k}^B\}_k$ the corresponding basis of $\mathcal{H}^B$ i.e., for all $k$, $\ket{k}^C\cong\ket{k}^B$. The new system resulting from looping the output system $C$ to the input system $B$, $\C{M}^{C\hookrightarrow B}$ is given as 
\begin{equation}
    \label{eq:loopfinite}
    \C{M}^{C\hookrightarrow B} (\ket{\psi}^A \bra{\phi}^A) = \sum_{k,l} \bra{k}^C\C{M}(\ket{\psi}^A\ket{k}^B\bra{l}^B\bra{\phi}^A)\ket{l}^C.
\end{equation}
\end{defi}

The sequential composition of two CP maps $\C{M}_A: \C{L}(\C{H}^A) \rightarrow \C{L}(\C{H}^C)$ and $\C{M}_B: \C{L}(\C{H}^B) \rightarrow \C{L}(\C{H}^D)$ is then given by $(\C{M}_A \otimes \C{M}_B)^{C\hookrightarrow B}$.

\begin{remark}[Basis dependence of the loop composition] Note that the composition is basis-dependent in the same sense that the Choi isomorphism and the link product are basis dependent. The bases we use for composition should thus be the same bases that we use to calculate Choi matrices and link products to obtain consistent results.
\end{remark}

\subsection{Process boxes}\label{sec:pb}

While there is some similarity between the process matrix and the causal box framework, they are nonetheless quite different. Process matrices do not assume a fixed background spacetime while causal boxes do so via the set $\C{T}$. Furthermore, the process matrix framework distinguishes between local labs and the environment, the latter being described by the process matrix itself, while causal boxes only have systems which can be composed to form other systems. The frameworks also treat messages rather differently. While the process matrix framework considers only a single message that is passed between the various local labs, the causal box framework allows for an arbitrary number of messages, including no message which is explicitly modeled by the vacuum state. Additionally, these messages can be looped back from the output of the causal box to the input at some later time. 

The process box framework \cite{Vilasini_2020} attempts to bridge this gap. The framework is obtained by imposing constraints on causal boxes such that they act like process matrices, i.e. acting as supermaps on other maps.

In the process box framework, we take the local labs from the process matrix framework and view them as causal boxes. In the process matrix framework, agents only act once, receiving a single message and then sending a single message. We can simulate this with causal boxes by restricting the input and output spaces to the 0- and 1-message spaces (or alternatively, starting with the process matrix picture and adding a vacuum state and position labels to the input and output spaces in \cref{def:locallabs}) and demanding that any output must come after some input. 


Furthermore, let us consider the set of position labels $\C{T}$, which we take to be discrete and finite as we are interested in processes that take place in some finite spatial region over a finite amount of time. It might be that an agent $A_i$ cannot send a non-vacuum state for some $t \in \C{T}$. Similarly, there might be some other $t' \in \C{T}$, for which that agent never receives a message. We can thus designate the set $\C{T}^I_i$ as the set of potential input positions and the set $\C{T}^O_i$ as the set of potential output positions for each agent $A_i$. By taking the union over all agents, we can also define $\C{T}^I = \bigcup_i \C{T}^I_i$ and $\C{T}^O = \bigcup_i \C{T}^O_i$. These are the sets of positions for which some agent could receive or respectively send a non-vacuum state. 

This now gives us the state spaces of the input wires $A^I_i$ and output wires $A^O_i$ of the agents as well as an additional wire $A^x_i$ carrying the measurement outcome.\footnote{For now we consider the full quantum instrument for the local operations of the agents, using the result wire $A^x_i$ to write it as a CPTP map. In \cref{sec:statespace}, we will discuss what the operation corresponding to a specific outcome looks like. Afterwards, we will mostly consider these, just like we did for the process matrix framework.}

\begin{gather}
\begin{aligned}
    \C{F}_{A^I_i} = \bigotimes_{t \in \C{T}^I_i} \C{F}^{t}_{A^I_i} = \bigotimes_{t \in \C{T}^I_i} (\underbrace{\ket{\Omega} \oplus \C{H}^{A^I_i}}_{\eqqcolon \C{H}^{\bar{A}^I_i}}) \otimes \ket{t} \\
    \C{F}_{A^O_i} = \bigotimes_{t \in \C{T}^O_i} \C{F}^{t}_{A^O_i} = \bigotimes_{t \in \C{T}^O_i} (\underbrace{\ket{\Omega} \oplus \C{H}^{A^O_i}}_{\eqqcolon \C{H}^{\bar{A}^O_i}}) \otimes \ket{t} \\
    \C{F}_{A^x_i} = \bigotimes_{t \in \C{T}^O_i} \C{F}^{t}_{A^x_i} = \bigotimes_{t \in \C{T}^O_i} (\underbrace{\ket{\Omega} \oplus \C{H}^{A^x_i}}_{\eqqcolon \C{H}^{\bar{A}^x_i}}) \otimes \ket{t}
\end{aligned}
\end{gather}

We can further formalize the above assumptions by introducing the following three conditions \cite{Vilasini_2020}.

\paragraph{Wire-space restriction (WSR):} The input and output wire spaces are associated with the restricted Fock spaces that allow only for zero or one message with positions in $\C{T}^{I/O}_i$. These are

\begin{gather}
\begin{aligned}
    \C{F}_{A^I_i}^{0/1} = \bigoplus_{n=0}^1 (\C{H}^{A^I_i} \otimes l^2(\C{T}^I_i))^{\otimes n} = \C{H}^{\bar{A}^I_i} \otimes l^2(\C{T}^I_i) \\
    \C{F}_{A^O_i}^{0/1} = \bigoplus_{n=0}^1 (\C{H}^{A^O_i} \otimes l^2(\C{T}^O_i))^{\otimes n} = \C{H}^{\bar{A}^O_i} \otimes l^2(\C{T}^O_i) \\
    \C{F}_{A^x_i}^{0/1} = \bigoplus_{n=0}^1 (\C{H}^{A^x_i} \otimes l^2(\C{T}^O_i))^{\otimes n} = \C{H}^{\bar{A}^x_i} \otimes l^2(\C{T}^O_i)
\end{aligned}
\end{gather}

Note that strictly speaking the second equality in each line does not hold as $(\C{H}^{A^I_i} \otimes l^2(\C{T}^I_i))^{\otimes 0} = \ket{\Omega}$ is one-dimensional while $\ket{\Omega} \otimes l^2(\C{T}^I_i)$ is $|\C{T}^I_i|$-dimensional and similarly for the output and measurement outcome spaces. However, we identify $\ket{\Omega, t} \cong \ket{\Omega}$ for all $t \in \C{T}$.

Using the wire isomorphism \cref{eq:wireiso}, we can identify the above spaces with the spaces spanned by the vectors $\{\ket{\psi, t_i} \otimes \ket{\Omega, \cancel{t_i}}: \ket{\psi} \in \C{H}^{\bar{A}^{I/O}_i}, t_i \in \C{T}^{I/O}_i\}$ with $\ket{\Omega, \cancel{t_i}}$ indicating vacuum states at all positions in $\C{T}^{I/O}_i$ except for $t_i$ and similarly for the result space $\C{F}_{A^x_i}^{0/1}$. While this notation is in principle ambiguous as to which agent's input or output positions this refers to, context will usually make this clear, especially if we add superscripts indicating the wire. For example, if we write $\ket{\psi, t_i}^{\bar{A}^I_i} \otimes \ket{\Omega, \cancel{t_i}}^{\bar{A}^I_i}$ it is clear that $\cancel{t_i} = \C{T}^I_i \backslash t_i$. This description makes the entanglement with the vacuum explicit. It is now also clear why we identify $\ket{\Omega, t} \cong \ket{\Omega}$ by plugging in $\Omega$ for $\psi$ in the above.

\paragraph{Local order (LO):} For each agent $A_i$ with non-trivial input and for each $t \in \C{T}^I_i$ the corresponding local operation $\C{M}^{a_i}_{\bar{A}_i}$, where $a_i$ denotes a measurement setting, acts on vacuum states as $\C{M}_{\bar{A}_i} \ket{\Omega, t}^{\bar{A}^I_i} = \ket{\Omega, \C{O}_i (t)}^{\bar{A}^O_i} \otimes \ket{\Omega, \C{O}_i (t)}^{\bar{A}^x_i}$ where $\C{O}_i: \C{T}^I_i \rightarrow \C{T}^O_i$ is an invertible function with $\C{O}_i(t) > t$ for all $t \in \C{T}^I_i$.

\bigskip

The LO assumption formalizes the notion that a non-vacuum output must be preceded by a non-vacuum input. In particular, LO implies that there is always an agent with trivial input which we can identify with the global past or there exists for all $t' \in \C{T}^O$ a $t \in \C{T}^I$ such that $t < t'$ (intuitively speaking, the process box is always the first to act if there are no agents with trivial input). The existence of the bijective maps $\C{O}_i$ also implies that $|\C{T}^I_i| = |\C{T}^O_i|$. 



Finally, as the process matrix framework does not contain any time dependency, one might expect that the agents in the process box framework should also act independently of the set $\C{T}$. However, this assumption can be slightly relaxed. It is enough to demand that the output at some $t' \in \C{T}^O_i$ of an agent $A_i$ only depends on the input at a single $t \in \C{T}^I_i$. We formalize this with the following constraint:

\paragraph{Operation-space restriction (OSR):} For each agent $A_i$, the corresponding local operation $\C{M}^{a_i}_{\bar{A}_i}$ is of the form $\C{M}^{a_i}_{\bar{A}_i}: \bigotimes_{t \in \C{T}^I_i} \C{M}_{\bar{A}_i}^{t, a_i}$ with $\C{M}_{\bar{A}_i}^{t, a_i}: \C{L}(\C{F}^t_{A^I_i}) \rightarrow \C{L}(\C{F}^{\C{O}_i(t)}_{A^O_i}) \otimes \C{L}(\C{F}^{\C{O}_i(t)}_{A^x_i}).$

\bigskip

Note that all these constraints are imposed on the local agents and not the process box. The process box is in principle free to violate them. For example, it could send non-vacuum states for some position $t$ even if it has not received a non-vacuum state at an earlier position $t'$ (as we noted earlier, this must be the case if there are no agents with trivial input). The process box only needs to have a well-defined composition with any set of local operations that respect WSR, LO and OSR. This essentially means that the process box only sends at most one message during the entire run to each agent.

A single additional constraint is imposed on the process box itself. It states that there are no ``passive" agents which receive only vacuum states during each $t \in \C{T}^I_i$ and send only vacuum states during each $t' \in \C{T}^O_i$. This is achieved by requiring that

\begin{equation}
    \tr_{\bar{A}^I_i} \circ \hat{\Phi} \circ \tr_{\bar{A}^O_i} \neq \hat{\Phi}
\end{equation}

for all $i$ where $\hat{\Phi}$ is the process box and $\tr_{\bar{A}^{I/O}_i}$ is the trace over the input/output space of agent $i$. Note that this constraint still allows for an agent to receive no messages in some branch of the superposition. 


Similar to QC-QCs, process boxes cannot violate causal inequalities \cite{Vilasini_2020}. In fact, this could be taken as a first hint that these two frameworks are closely related as we will show later in the results section.

\subsubsection{The effective Choi representation}\label{sec:effective}

In general, the Choi representation $\Phi$ of a process box $\hat{\Phi}$ takes the form

\begin{equation}
    \Phi = \C{I}^{\C{F}^O} \otimes \hat{\Phi} \dket{\mathbb{1}} \dbra{\mathbb{1}}^{\C{F}^O \C{F}^O} \in \C{L}(\C{F}^O \otimes \C{F}^I)
\end{equation}

where $\C{F}^{I/O} = \bigotimes_{i=1}^N \C{F}_{A^{I/O}_i}$. However, this space is much more general than what can actually be achieved given the constraints we put on the local agents. In particular, it contains states that correspond to agents receiving or sending multiple messages at different times. We can therefore define the effective Choi representation as

\begin{equation}
    \Phi^{\text{eff}} = \C{P}^{\text{eff}} (\C{I}^{\C{F}^O} \otimes \hat{\Phi} \dket{\mathbb{1}} \dbra{\mathbb{1}}^{\C{F}^O \C{F}^O})
\end{equation}

where $\C{P}^{\text{eff}}$ is the projector on

\begin{equation}\label{eq:Heff}
    \C{H}^{\text{eff}} = \text{Span}[\{\bigotimes_{i=1}^N \ket{j_i, t_i}^{\bar{A}^I_i} \otimes \ket{\Omega, \C{T}^I_i \backslash t_i}^{\bar{A}^I_i} \otimes \ket{k_i, \C{O}_i(t_i)}^{\bar{A}^O_i} \otimes \ket{\Omega, \C{T}^O_i \backslash \C{O}_i(t_i)}^{\bar{A}^O_i}\}_{\substack{j_i, k_i, t_i: \\ j_i = \Omega \implies k_i = \Omega}}].
\end{equation}

This is essentially the space one obtains by imposing WSR, LO and OSR on $\C{F}^I \otimes \C{F}^O$. The consequence is that if the process box is composed with local agents fulfilling the three assumptions, it does not matter whether we use the normal or the effective Choi representation \cite{Vilasini_2020}. They are operationally indistinguishable which means that the probabilities are the same or more generally, the output state is the same after composition

\begin{equation}\label{eq:pbeffcomp}
    (M_{\bar{A}_1} \otimes ... \otimes M_{\bar{A}_N}) \cdot \Phi = (M_{\bar{A}_1} \otimes ... \otimes M_{\bar{A}_N}) \cdot \Phi^{\text{eff}}.
\end{equation}

Here, $\cdot$ denotes complete composition, i.e. composition of the input of one with the output of the other and vice versa.

\subsubsection{Additional assumptions}

In the QC-QC framework, agents are assumed to be fully active, that is they receive and send a state in all branches of the superposition. Meanwhile, in the process box framework as described above it is possible for an agent to receive only the vacuum for all positions, as long as this does not happen with probability 1. Additionally, agents can send vacuum states on their output wires during all positions even if they receive a non-vacuum state at some point. 

In order to remove this disconnect between the QC-QC and the process box framework, we introduce two additional constraints, not made in \cite{Vilasini_2020}, one on the agents and one on the process box itself.



\paragraph{Active agents (AA):} For each agent $A_i$ and for each $t \in \C{T}^I_i$, it holds for the corresponding local operation $\C{M}^{a_i}_{\bar{A}_i}$ that $\C{M}^{a_i}_{\bar{A}_i} \ket{\psi, t}^{A^I_i} = \ket{\phi, \C{O}_i(t)}^{A^O_i} \otimes \ket{x, \C{O}_i}^{A^x_i}$ (or a coherent or incoherent sum of such terms). 

\bigskip

This constraint can be viewed as the converse of LO and implies that when an agent receives a non-vacuum input, they produce a non-vacuum output and a non-vacuum measurement result. Note that we use $A^{I/O/x}_i$ instead of $\bar{A}^{I/O/x}_i$ to indicate that these states are elements of $\C{H}^{A^{I/O/x}_i} \otimes \C{T}^{I/O/x}_i$, that is one-message states. 

The second constraint consists of assuming that there are no passive agents stronger. Instead of assuming that the probability that an agent acts is non-zero, we assume that it is one.

Note that this constraint is related to AA. Both together ensure that for every agent the probability that they only ever receive or send vacuum states is zero. We will thus generally refer to both of them together as the constraint of fully active agents (FAA).

Our results in the following section should thus be understood with FAA in mind. An alternative to using FAA would be adding explicit vacuum states to the QC-QC framework which would allow QC-QCs to have passive agents or agents that send nothing even though they received a vacuum state. This was done for the quantum switch in \cite{Paunkovi__2020}.




\newpage

\part{Results}\label{sec:results}

\section{Mapping QC-QCs to process boxes}\label{sec:qcqctopb}

\subsection{The state spaces of QC-QCs and process boxes}\label{sec:statespace}

Our overall goal will be to map QC-QCs to process boxes. If a QC-QC and a process box are related via this mapping, then they should in some sense describe the same thing. This means we need to establish a notion of equivalence between a QC-QC and a process box. The most reasonable way is to do this in an operational sense, that is if the local agents apply the same local operations, the QC-QC and the process box should predict the same probabilities or, more generally, their actions as supermaps should be the same. However, the QC-QC framework and the process box framework define their agents and state spaces in slightly different ways. Clarifying how these agents and state spaces relate to each other is thus a necessary step in formulating such an equivalence relation.

In the QC-QC picture, we had agents $A_k$ with input and output spaces $\C{H}^{A^I_k}$ and $\C{H}^{A^O_k}$. The QC-QC then models temporal superpositions by identifying these with the generic spaces $\C{H}^{\tilde{A}^I_n}$ and $\C{H}^{\tilde{A}^O_n}$ which correspond to a specific time step $t_n$ with the label $n$ being determined by a classical or quantum control system. 

We can then view these generic spaces as corresponding to the state space in the process box picture at a specific time. For simplicity, we will assume $\C{T}$ to be totally ordered and associate even times with inputs, $\C{T}^I = \{2, 4,..., 2N+2\}$, and odd times with outputs $\C{T}^O = \{1, 3,..., 2N+1\}$ and take $\C{O}_i(t) = t+1$ with $t=1$ only being a valid output time for a single agent which we identify with the global past and similarly for $t=2N+2$ and a global future agent. Later, when we attempt to construct a mapping of process boxes to QC-QCs, we will give some arguments that suggest that this is without loss of generality. However, for the mapping of QC-QCs to process boxes making restrictions to which process boxes we consider is a priori not a problem (as long as we do not exclude all process boxes a specific QC-QC could be mapped to).

\begin{figure}[t!]
    \centering
    \begin{subfigure}{1\textwidth}
    \centering
    \includegraphics[width=\textwidth]{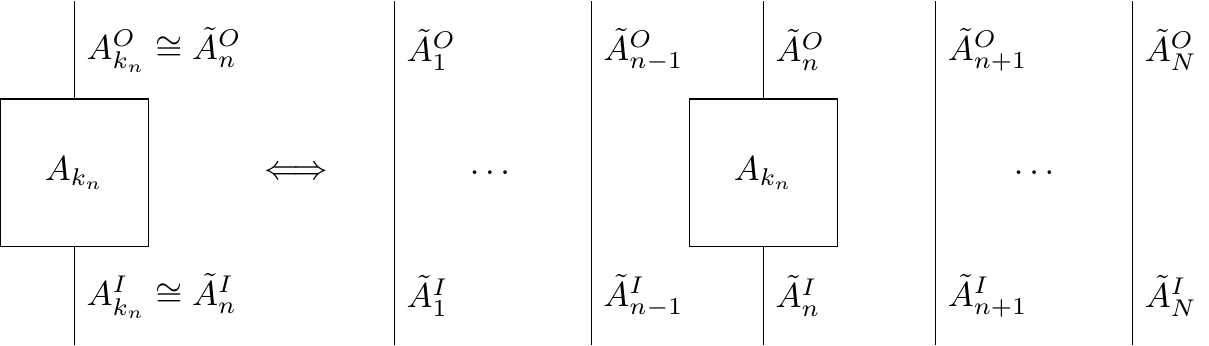}
    \caption{In the QC-QC framework the agent is initially associated with a single input and output wire in between which they apply their operation (left). The output spaces are then identified with generic spaces, for which there is one for each time step, yielding the picture to the right. In a given branch of the superposition, the agent acts during one time step and is inactive for all others.}
    \label{fig:qcqcstatespaces}
    \end{subfigure}
    \begin{subfigure}{1\textwidth}
    \centering
    \includegraphics[width=\textwidth]{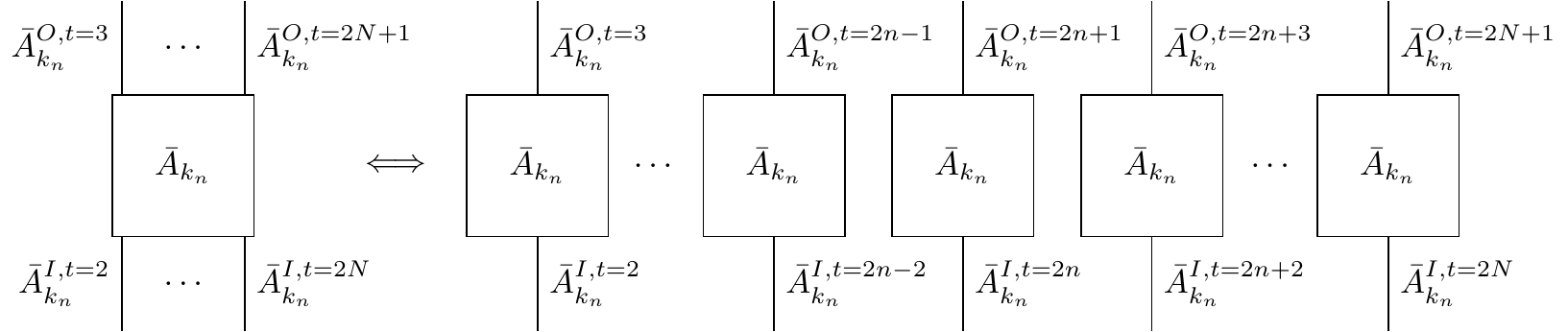}
    \caption{In the process box framework the agent can be viewed as having a wire for each time step (left) and acts during all time steps which yields the picture to the right. But as they receive only a single non-vacuum state under WSR all but one of the operations to the right are simply the vacuum projector. }
    \label{fig:pbstatespaces}
    \end{subfigure}
    \caption{Not acting during a time step and acting trivially on the vacuum can be understood as the same action. The overall action in the two frameworks is therefore effectively equivalent as long as the agents in either framework receive exactly one message, which is always the case in the QC-QC framework and requires the imposing of WSR and FAA (cf. \cref{sec:pb}) in the process box framework.}
\end{figure}


\Cref{fig:qcqcstatespaces,fig:pbstatespaces} then illustrate the correspondence. The input space $\C{H}^{A^I_k}$ is isomorphic to all $N$ generic input spaces $\C{H}^{\tilde{A}^I_n}$, but in a given branch of the temporal superposition, we only identify it with one of them. This is because of the assumption that each agent only acts once during a run of the experiment modeled by the QC-QC. On the other hand, the agent $A_k$ in the process box picture can be viewed as having $N$ input wires, one for each $t \in \C{T}^I$. However, once again in a given branch of the temporal superposition only one of these wires will contain a non-vacuum state. This time the reason is WSR. We thus find the correspondences

\begin{gather}
\begin{aligned}
    \C{H}^{\tilde{A}^I_n} &\leftrightarrow \C{H}^{\bar{A}^I_k} \otimes \ket{t=2n} \\
    \C{H}^{\tilde{A}^O_n} &\leftrightarrow \C{H}^{\bar{A}^O_k} \otimes \ket{t=2n+1}
\end{aligned}
\end{gather}
where the second line relating the output spaces of the two frameworks is obtained via analogous reasoning to that which yielded the first line.

On the level of states, we can then make the following identification

\begin{gather}
\begin{aligned}
    \ket{\psi}^{\tilde{A}^I_n} \otimes \ket{\C{K}_{n-1}, k_n} &\leftrightarrow \ket{\psi, t=2n}^{\bar{A}^I_{k_n}} \ket{\C{K}_{n-1}, k_n}^{C_n} \otimes \ket{\Omega, \C{T}^I \backslash 2n}^{\bar{A}^I_{k_n}} \ket{\Omega}^{C_n} \\
    \ket{\psi}^{\tilde{A}^O_{n}} \otimes \ket{\C{K}_{n-1}, k_n} &\leftrightarrow \ket{\psi, t=2n+1}^{\bar{A}^O_{k_n}} \ket{\C{K}_{n-1}, k_n}^{C_n} \otimes \ket{\Omega, \C{T}^O \backslash 2n+1}^{\bar{A}^O_{k_n}} \ket{\Omega}^{C_n}
\end{aligned}
\end{gather}
where we now also added the control system which allows us to make the previous correspondence one-to-one.

Given these correspondences, let us now consider how the agents apply their local operations in each framework. In the QC-QC framework, we can view the agent $A_{k_n}$ as applying their local operation during the time step $t_n$, or, to use the same time stamps as for the process box, during $t=2n$ and outputting the result during $t=2n+1$. During all other time steps, they are essentially inactive.

In the process box framework, the agents act during all input times which can be seen from the form of their local operation $\C{M}_{\bar{A}_i} = \bigotimes_{t \in \C{T}^I} \C{M}_{\bar{A}_i}^t$ where we dropped the setting label $a$ to avoid clutter. They could therefore, in principle, obtain some measurement result for each time step. However, due to WSR the agent only receives a non-vacuum state once, although this state may arrive in a temporal superposition of different time steps. The LO constraint then tells us that the agent acts trivially on vacuum states, obtaining the outcome $\Omega$ for all other time steps. Therefore, we can write the effective action of the agent $A_i$ who receives a non-vacuum state at $t \in \C{T}^I_i$ and obtains the outcome $x$ as

\begin{equation}\label{eq:kraustx}
    \bar{A}^{t, x}_i \otimes \proj{\Omega, \C{O}_i(\cancel{t})}{\Omega, \cancel{t}} 
\end{equation}

where $\bar{A}^{t, x}_i: \C{F}_{A^I_i}^{t} \rightarrow \C{F}_{A^O_i}^{\C{O}_i(t)} \otimes \C{F}_{A^x_i}^{\C{O}_i(t)}$ is the Kraus operator of $\C{M}_{\bar{A}_i}^t$ associated with the outcome $x$ while the vacuum projector $\proj{\Omega, \C{O}_i(\cancel{t})}{\Omega, \cancel{t}}$ acts between $\C{F}_{A^I_i}^{t'}$ and $\C{F}_{A^O_i}^{\C{O}_i(t')} \otimes \C{F}_{A^x_i}^{\C{O}_i(t')}$ for all $t' \neq t$. Note that due to FAA we can assume that $x \neq \Omega$ as the agent cannot obtain the outcome $\Omega$ for all positions.

Acting trivially on the vacuum could alternatively be achieved by simply not acting at all during these time steps. We can thus view \cref{eq:kraustx} as the equivalent action in the process box framework to an agent applying their measurement during the corresponding time step in the QC-QC framework.

Further, we are not interested in the $t$ for which the outcome $x$ was obtained. The probability to obtain outcome $x$ is then simply the probability to obtain that outcome for any $t$. Additionally, we do not wish to collapse the temporal superposition, which is why the operation associated with obtaining $x$ (at any $t$) is the sum over all $t \in \C{T}^I_i$ of \cref{eq:kraustx}
\begin{equation}\label{eq:notime}
    \bar{A}^x_i \coloneqq \sum_{t \in \C{T}^I_i} \bar{A}^{t, x}_i \otimes \proj{\Omega, \C{O}_i(\cancel{t})}{\Omega, \cancel{t}}.
\end{equation}


The operator $\bar{A}^x_i$ can then be understood as the Kraus operator associated to agent $i$ obtaining the outcome $x$. Note that this definition only makes sense if we impose WSR as otherwise the agent could obtain multiple outcomes at different times.

Given FAA, we can assume that the state on the input wire of the agent has no overlap with the state that has the vacuum for all input positions, $\ket{\Omega, \C{T}^I_i}^{\bar{A}^I_i}$. We can then get rid of the sum in \cref{eq:notime} by rewriting it as

\begin{equation}\label{eq:local_pb}
    \bar{A}^x_i = \bigotimes_{t \in \C{T}^I_i} (\bar{A}^{t, x}_i + \proj{\Omega, \C{O}_i(t)}{\Omega, t}).
\end{equation}

We can see that these two equations are equivalent on the WSR restricted space by direct calculation,


\begin{gather}
\begin{aligned}
    (\sum_{t \in \C{T}^I_i}& \bar{A}^{t, x}_i \otimes \proj{\Omega, \C{O}_i(\cancel{t})}{\Omega, \cancel{t}}) \ket{\psi, t_i}^{\bar{A}^I_i} \ket{\Omega, \cancel{t_i}}^{\bar{A}^I_i}  \\
    =& \bar{A}^{t_i, x}_i \ket{\psi, t_i}^{A^I_i} \otimes \proj{\Omega, \C{O}_i(\cancel{t_i}))}{\Omega, \cancel{t_i}} \ket{\Omega, \cancel{t_i}}^{\bar{A}^I_i} \\
    =& (\bar{A}^{t_i, x}_i + \proj{\Omega, \C{O}_i(t_i)}{\Omega, t_i}) \ket{\psi, t_i}^{\bar{A}^I_i} (\bar{A}^{\cancel{t_i}, x}_i + \proj{\Omega, \C{O}_i(\cancel{t_i})}{\Omega, \cancel{t_i}}) \ket{\Omega, \cancel{t_i}}^{\bar{A}^I_i} \\
    =& \bigotimes_{t\in \C{T}^I_i}  (\bar{A}^{t, x}_i + \proj{\Omega, \C{O}_i(t)}{\Omega, t}) \ket{\psi, t_i}^{\bar{A}^I_i} \ket{\Omega, \cancel{t_i}}^{\bar{A}^I_i}
\end{aligned}
\end{gather}
where we used that $\braket{\psi|\Omega} = 0$ and that $\bar{A}^{t, x}_i \ket{\Omega, t} = 0$ for all $t \in \C{T}^I_i$. The latter is due to LO which imposes $\bar{A}^{t, \Omega} \ket{\Omega, t} = \ket{\Omega, \C{O}_i(t)}$. The probability to obtain the outcome $\Omega$ when receiving the vacuum is thus 1, which means all other measurement outcomes have probability 0. This is equivalent to saying the Kraus operators associated with all other measurement outcomes must annihilate the vacuum. 


In the QC-QC framework, we assume the action of the agents to be independent of time. Let us therefore do the same for the action in the process box framework, i.e. we assume that $\bar{A}^{t, x}_i$ decomposes into a part acting on $\C{H}^{\bar{A_i}}$ which is independent of $t$ and a projector acting on the time stamp $\proj{\C{O}_i(t)}{t}$. 

Additionally, let us drop the superscript denoting the measurement outcome $x$ as it is enough to consider each outcome separately as we already did for the process matrix framework.\footnote{Note that we defined the local agents in the process box framework as causal boxes, which we defined as CPTP maps. However, defining causal boxes as CP maps, so-called subnormalized causal boxes, is also possible \cite{Portmann_2017}} We then also no longer need the wire for the measurement outcome $A^x_i$. 

We then define equivalence between a local operation in the QC-QC framework and one in the process box framework as follows.

\begin{defi}[Equivalence of local operations]\label{def:local_equivalence}
Let $A_i$ be a local agent in the QC-QC framework who applies a local operation $A_i: \C{H}^{A^I_i} \rightarrow \C{H}^{A^O_i}$. We define the equivalent agent $A_i$ in the process box framework with input/output space $\C{F}^{0/1}_{A_i^{I/O}} = \C{H}^{\bar{A}_i^{I/O}} \otimes l^2(\C{T}_i^{I/O})$. We then call the local operation $A_i$ in the QC-QC framework operationally equivalent to the local operation $\bar{A}_i$ in the process box framework if $\bar{A}_i = \bigotimes_{t \in \C{T}^I_i} (A_i \otimes \proj{\C{O}_i(t)}{t} + \proj{\Omega, \C{O}_i(t)}{\Omega, t})$ for some invertible $\C{O}_i: \C{T}^I_i \rightarrow \C{T}^O_i$ with $\C{O}_i(t) > t$. 
\end{defi}




Understanding now what it means for an operation in the QC-QC and the process box framework to be the ``same", we can also formulate what it means for a QC-QC and a process box to be equivalent. However, this will be easier to do if we first show that sequential composition in the causal box framework is the same as the link product as we already mentioned in \cref{sec:composition}. 

\begin{restatable}[Composition is equivalent to the link product]{lemma}{compislink}
\label{lemma:compislink}
Let $\C{M}_A: \C{L}(\C{H}^A) \rightarrow \C{L}(\C{H}^C)$ and $\C{M}_B: \C{L}(\C{H}^B) \rightarrow \C{L}(\C{H}^D)$ be two CP maps and denote their parallel composition as $\C{M} = \C{M}_A \otimes \C{M}_B$. Then, the Choi matrix of their sequential composition $\C{M}^{C \hookrightarrow B} = (\C{M}_A \otimes \C{M}_B)^{C \hookrightarrow B}$ can be expressed with the link product as

\begin{equation}
    M^{C \hookrightarrow B} = M_A * M_B
\end{equation}
if one takes $\C{H}^B = \C{H}^C$ for the purposes of the link product.
\end{restatable}

This allows us to write the composition of a process box with local agents as

\begin{equation}
    (M_{\bar{A}_1} \otimes ... \otimes M_{\bar{A}_N}) * \Phi^{\text{eff}}
\end{equation}

or in terms of the effective Choi vector

\begin{equation}
    (\bar{A}_1 \otimes ... \otimes \bar{A}_N) * \dket{\bar{V}}.
\end{equation}

Finally, we introduce the definition of operational equivalence between a QC-QC and a process box.

\begin{defi}[Equivalence of QC-QCs and process boxes]\label{def:pbqceq}
Let $\ket{w_{\C{N}, F}}$ be the process vector of an $N$-partite QC-QC and $\dket{\bar{V}}$ the effective Choi vector of an $N$-partite process box whose set of positions is totally ordered. We say that the QC-QC and the process box are operationally equivalent if for any set of local operations of the QC-QC $\{A_1,...,A_N\}$ and any set of local operations of the process box $\{\bar{A}_1,...,\bar{A}_N\}$ such that $A_n$ and $\bar{A}_n$ are operationally equivalent for all $n \in \C{N}$, it holds that

\begin{equation}
    (\dket{A_1} \otimes .... \otimes \dket{A_N}) * \ket{w_{\C{N}, F}} \cong (\dket{\bar{A}_1} \otimes ... \otimes \dket{\bar{A}_N}) * \dket{\bar{V}}.
\end{equation}

where $\cong$ should be understood as a unitary isomorphism between the global past spaces as well as between the global future spaces of the QC-QC and the process box.

\end{defi}

The isomorphisms that relate the global past/future spaces will usually be rather simple, for example dropping/adding the time stamp $\C{H}^{P} \cong \C{H}^{P} \otimes \ket{t=1}$ and $\C{H}^F \cong \C{H}^{F} \otimes \ket{t=2N+2}$. We exclude the vacuum from the output/input space of the global past/future which is without loss of generality due to FAA.

For simplicity, we just consider pure processes in the above definition. This is not a problem as we have seen that all QC-QCs and all causal boxes are purifiable. Therefore, when we speak of the global future from now on this should be understood as including any purifying ancillaries $\alpha_F$. 

\subsection{Dynamical switch as a process box}\label{sec:switchtopb}

Our goal for this section is to find a process box description of the dynamical switch which we discussed in \cref{sec:dynamicalswitch}. This concrete example will help us understand how QC-QCs can be mapped to process boxes in general.

We will consider two different constructions. Our first construction will consider the components of the experimental setup depicted in \cref{fig:new_QCQC}. We will show that these components are isometries for the types of inputs that are possible in the process box framework. As compositions and tensor products of isometries are isometries again, we obtain a sequence representation this way. For the second approach we will construct a sequence representation using the internal operators $V^{\rightarrow k_{n+1}}_{\C{K}_{n-1}, k_n}$ of the QC-QC. 

\subsubsection{Component-wise construction}\label{sec:component}

The COPY gate in \cref{fig:new_QCQC} implements the operation $V_{\text{Copy}} = \sum_{i=0,1} \ket{i}^{t} \ket{i}^{\alpha} \bra{i}^t$. This is an isometry as $V_{\text{Copy}}^\dagger V_{\text{Copy}} = \sum_{i,j=0,1} \ket{i}^t \braket{i|j}^t \braket{i|j}^\alpha \bra{j}^t = \sum_{i=0,1} \ket{i} \bra{i}^t = \mathbb{1}^t$. Similarly, for the CNOT gates, which apply $V_{\text{CNOT}} = \sum_{ij} \ket{i}^t \ket{i \oplus j}^\alpha \bra{i}^t \bra{j}^\alpha$, we find $V_{\text{CNot}}^\dagger V_{\text{CNot}} = \sum_{i,j, k, l=0,1} \ket{i}^t \ket{j}^\alpha \braket{i|k} \braket{i \oplus j | k \oplus l} \bra{k}^t \bra{l}^\alpha = \sum_{i,j = 0,1} \ket{i}^t \ket{j}^\alpha \bra{i}^t \bra{j}^\alpha = \mathbb{1}^{t \alpha}$ showing that these gates are isometries as well.

For a causal box representation, we also need to consider the possibility of multiple messages as well as no messages arriving at one of these gates. Are the COPY and CNOT gates still isometries in this case? The gates should act on each message individually, i.e. we apply $V_{\text{Copy}}$ and $V_{\text{CNot}}$ to each tensor factor, $V_{\text{Copy}} \bigodot_i \ket{\psi^i}^t = \bigodot_i V_{\text{Copy}} \ket{\psi^i}^t$ and $V_{\text{CNot}} \bigodot_i \ket{\psi^i}^t \ket{\phi^i}^\alpha = \bigodot_i V_{\text{CNot}} \ket{\psi^i}^t \ket{\phi^i}^\alpha$. In other words, the gates for $n$ messages are simply an $n$-fold tensor product of the gates for a single message. Thus, they remain isometries even for multiple messages. If no messages arrive, the gates simply do nothing, i.e. we posit that the gates leave the vacuum invariant. As the vacuum is orthogonal to all other states, this extension is also compatible with the gates being isometries.

Next, we need to show that beam splitters are isometries. A beam splitter consists of two input wires $I^0$ and $I^1$ and two output wires $O^0$ and $O^1$, each of which can carry 2-dimensional messages. The action of the beam splitter is that it reflects one type of polarization (say photons in state $\ket{0}$) and transmits those of the orthogonal polarization ($\ket{1}$). We can thus say that a polarizing beam splitter sends a photon in state $\ket{i}$ arriving on wire $I^j$ to the output wire $O^{i \oplus j}$. If there is only one photon, it is clear that the beam splitter acts as an isometry. However, we cannot view the action of the beam splitter on $n$ photons as a simple $n$-fold tensor product due to there being two input and output spaces, which have to be symmetrized separately.

Instead, we construct orthogonal bases of input and output Fock spaces and use them to generalize the beam splitter. A basis for the inputs is given by the vectors 
\begin{equation}
\ket{m, n, k, l}^I = (\underbrace{\ket{0}^{I^0} \odot ... \odot \ket{0}^{I^0}}_{m \text{ times}} \odot \underbrace{\ket{1}^{I^0} \odot ... \odot \ket{1}^{I^0}}_{n \text{ times}}) \otimes (\underbrace{\ket{0}^{I^1} \odot ... \odot \ket{0}^{I^1}}_{k \text{ times}} \odot \underbrace{\ket{1}^{I^1} \odot ... \odot \ket{1}^{I^1}}_{l \text{ times}})
\end{equation}
with $m, n, k, l \in \mathbb{N}$ (if $m=n=0$ or $k=l=0$, we take the state on the respective wire to be the vacuum). In other words, the state $\ket{m, n, k, l}$ is the one with $m$ photons with polarization $\ket{0}$ and $n$ photons with polarization $\ket{1}$ on the $0$ wire and  $k$ photons with polarization $\ket{0}$ and $l$ photons with polarization $\ket{1}$ on the $1$ wire. We also define a basis for the output state $\ket{m, n, k, l}^O$ completely analogously. 

What does the beam splitter map $\ket{m,n,k,l}$ to? The $\ket{0}$ photons are mapped to the output wire with the same label as the respective input wire while the $\ket{1}$ photons are sent to the output wire with the opposite label. Thus, we define the action of the beam splitter as $\ket{m,n,k,l}^I \mapsto \ket{m,l,k,n}^O$. The norm of $\ket{m,n,k,l}^{I/O}$ is $\sqrt{m! n! k! l!}$ which is invariant under exchange of any two of $m, n, k, l$.\footnote{The norm is the product of the norms of the state on each wire. We use \cref{eq:fock_inner} to find the norm of $\underbrace{\ket{0}^{I^0} \odot ... \odot \ket{0}^{I^0}}_{m \text{ times}} \odot \underbrace{\ket{1}^{I^0} \odot ... \odot \ket{1}^{I^0}}_{n \text{ times}}$. We have a sum over permutations in $S^{n+m}$ and as $\braket{0|1}=0$ only permutations that map $\{1,...,n\}$ and $\{n+1,...,n+m\}$ to themselves contribute and they contribute 1 to the norm squared. These permutations can be viewed as the product of $S^n$ with $S^m$ and as such there are exactly $n!m!$ such permutations. Repeating the same argument for the other input wire and also the output wires gives us the norm as in the main text.} The beam splitter thus conserves the norm of basis states. Additionally, different basis states remain orthogonal as we map an orthogonal basis of the input space injectively to an orthogonal basis of the output space. Thus, the beam splitter conserves the inner product between basis states which means it is an isometry (see \cref{sec:isometries} for further information on isometries).


Finally, we need to add the time stamps. We can do so after composing the gates and beam splitters into the various isometries that take the outputs of the local agents to their inputs. If $\bar{V}$ is such an isometry we can replace it with $\bar{V} \otimes \C{O}$ where $\C{O}$ simply increases the time by 1, $\C{O}\ket{t} = \ket{t+1}$. The operator $\C{O}$ is an isometry and therefore $\bar{V} \otimes \C{O}$ is an isometry as well. 

In conclusion, arbitrary compositions and tensor products of COPY and CNOT gates and polarizing beam splitters are isometries. Any sequence representation made up of these components is thus a causal box. 

However, it is not immediately clear that such causal boxes are also process boxes. In fact, the WSR requirement seems to be violated as the beam splitter sends a $\ket{0}$ photon on one input wire and a $\ket{1}$ photon on the other input wire to the same output wire. Note however, that all the components we discussed keep the number of messages constant. Thus, WSR is fulfilled as long as each component receives at most a single message. Due to LO and inductive reasoning this is the case if at the first possible input time only one of the local agents receives a non-vacuum state (superpositions being allowed). This is the case for the dynamical switch and is in fact the case for any process that can be described as a QC-QC. 

\subsubsection{Construction using the QC-QC operators}\label{sec:constructionqcqcoperators}

We will now present another approach that utilizes the operators from the QC-QC picture. First of all, we define the set $\C{T} = \{2,3,4,5,6,7,8\}$ as the time steps during which the local and internal operations can act. We do not include $t=1$ as the global past is trivial and so we do not have to model it as an agent. We split the set $\C{T}$ into the set of input times $\C{T}^I = \{2, 4, 6, 8\}$ and the set of output times $\C{T}^O = \{3, 5, 7\}$ where $t=8$ is an input time only for the global future. Next, considering the experimental setup in \cref{fig:new_QCQC}, we note that the target, the ancillary and the control are all part of the same overall system as they are all encoded in different degrees of freedom of the photon. Importantly, this means that the ancillary and control are both sent to local agents, but they only act trivially on it. The input/output spaces of the agents will thus be $\C{F}^{0/1}_{\bar{A}^{I/O}_i \bar{\alpha} \bar{C}}$, but we only allow local operations that leave $\C{H}^{\bar{\alpha}}$ and $\C{H}^{\bar{C}}$ invariant. This is essentially equivalent to routing these two systems past the local operations and directly to the next internal operation, which is not really different from sending it internally (cf. \cref{fig:routingpast}). We can thus still consider agents whose input and output Hilbert spaces do not include the ancillary and control Hilbert spaces and contract over these spaces when calculating the Choi vector.

\begin{figure}
\centering
\includegraphics[width=\textwidth]{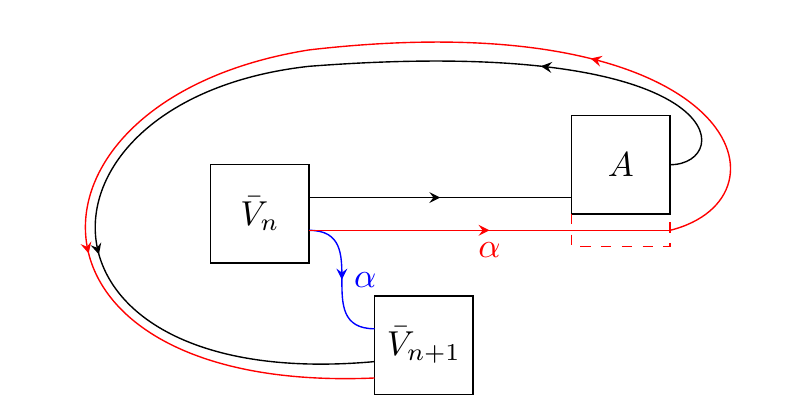}
\caption{Diagrammatic representation of how a process box can send an ancillary $\alpha$ from one isometry of the sequence representation to the next. The isometry $\bar{V}_n$ can either send it internally to $\bar{V}_{n+1}$ (blue) or send it to the local agent $A$ along with the target system (black line) who then sends it to $\bar{V}_{n+1}$ (red). If $A$ acts trivially on $\alpha$ (indicated by the red, dashed extension of $A$), then the two options are essentially equivalent.}
\label{fig:routingpast}
\end{figure}

By again considering the experimental setup from \cref{fig:new_QCQC}, we see that if we intervene in the circuit and place a photon directly after one of the local operations but before the internal operations, the control automatically corresponds to that local operation. That is the state on the output wires is given by $\ket{\psi, t=2}^{A^O_{k_1}} \ket{\emptyset, k_1}^{C_1}, \ket{\psi, t=4}^{A^O_{k_2}} \ket{\{k_1\}, k_2}^{C_2}$ or $\ket{\psi, t=6}^{A^O_{k_3}} \ket{\{k_1, k_2\}, k_3}^{C_3}$\footnote{We did not put the line over the agent label as we specifically consider non-vacuum states. The other two output wires contain vacuum states.} depending on when exactly we intervened. This is a consequence of the fact that the control is encoded in the path degree of freedom (in the case of $C_2$, $k_1$ is encoded in the polarization, but $k_2$ is encoded in the path). We thus do not have to wonder how the internal operations act on states such as $\ket{\psi}^{A^O_1} \ket{\emptyset, 2}^{C_1}$ as they can never occur.

Another question we could ask is what happens if we place photons after multiple local operations. For example, we could place a photon on each of the output wires of agent 1 and agent 2 directly after their first application. The corresponding state would be $\ket{\psi, t=2}^{A^O_1} \ket{\emptyset, 1}^{C_1} \ket{\phi, t=2}^{A^O_2} \ket{\emptyset, 2}^{C_1} \ket{\Omega}^{A^O_3} \ket{\Omega}^{C_1}$. As argued previously, beam splitters, COPY gates and CNOT gates all act on each photon they receive independently from any other photon that may be in the circuit (besides the need to symmetrize). Thus, the photon coming from agent 1 can be sent to agent 2 or 3, while the photon coming from agent 2 can be sent to agent 1 or 3 depending dynamically on the state of each photon. In general, we expect a superposition of all these possibilities

\begin{gather}
\begin{aligned}\label{eq:2photons}
    \bar{V}_2 &((\alpha_1 \ket{0, t=3}^{A^O_1} + \beta_1 \ket{1, t=3}^{A^O_1}) \ket{\emptyset, 1}^{C_1} (\alpha_2 \ket{0, t=3}^{A^O_2} + \beta_2 \ket{1, t=3}^{A^O_2}) \ket{\emptyset, 2}^{C_1} \ket{\Omega, t=3}^{\bar{A}^O_3} \ket{\Omega}^{\bar{C}_1}) \\
    &= \alpha_1 \alpha_2 \ket{\Omega, t=4}^{\bar{A}^I_1 \bar{C}_2} \ket{0, t=4}^{A^I_2} \ket{\{1\}, 2}^{C_2} \ket{0, t=4}^{A^I_3} \ket{\{2\}, 3} \\
    &+ \alpha_1 \beta_2 \ket{1, t=4}^{A^I_1} \ket{\{2\}, 1}^{C_2} \ket{0, t=4}^{A^I_2} \ket{\{1\}, 2}^{C_2} \ket{\Omega, t=4}^{\bar{A}^I_3 \bar{C}_2} \\
    &+ \beta_1 \beta_2 \ket{1, t=4}^{A^I_1} \ket{\{2\}, 1}^{C_2} \ket{\Omega, t=4}^{\bar{A}^I_2 \bar{C}_2} \ket{1, t=4}^{A^I_3} \ket{\{1\}, 3}  \\
    &+ \beta_2 \alpha_1 \ket{\Omega, t=4}^{\bar{A}^I_1 \bar{C}_2} \ket{\Omega, t=4}^{\bar{A}^I_2 \bar{C}_2} ((\ket{1, t=4}^{A^I_3} \ket{\{1\}, 3}^{C_3}) \odot (\ket{0, t=4}^{A^I_3} \ket{\{2\}, 3}^{C_3}))
\end{aligned}
\end{gather}
where we placed a line over $V_2$ to differentiate it from the internal operation $V_2$ in the QC-QC picture.

The last term corresponds to both photons being sent to agent 3 and therefore has to be symmetrized, indicated by the symmetric tensor product $\odot$, to be a valid element of the Fock space. Inspecting \cref{eq:2photons}, we see that it is achieved by applying the internal operations to each non-vacuum tensor factor while adding vacuum states and symmetrizing where needed. 

To be able to write down the action of the isometries with these ideas in mind, we introduce the operation $symm$ which symmetrizes tensor products in the same Hilbert space and leaves tensor products over different Hilbert spaces invariant, $symm(\C{H}^X \otimes \C{H}^X) = \C{H}^X \odot \C{H}^X$ and $symm(\C{H}^X \otimes \C{H}^Y) = \C{H}^X \otimes \C{H}^Y$ for $X \neq Y$. More generally, for $N$ Hilbert spaces $\C{H}^{X_n}$ such that $X_n \neq X_m$ if $n \neq m$, we define it as

\begin{equation}
    symm(\bigotimes_{n=1}^N (\ket{\psi^{n, 1}}^{X_n} \otimes \ket{\psi^{n,2}}^{X_n} \otimes ...)) \coloneqq \bigotimes_{n=1}^N (\ket{\psi^{n, 1}}^{X_n} \odot \ket{\psi^{n,2}}^{X_n} \odot ...).
\end{equation}

\begin{example}[Symmetrization]
Consider three systems, $A, B, C$ and the tensor product of states on these systems $\ket{0}^A \ket{1}^A \ket{0}^B \ket{1}^C \ket{1}^C$. The action of $symm$ on this state is

\begin{equation}
    symm(\ket{0}^A \ket{1}^A \ket{0}^B \ket{1}^C \ket{1}^C) = \ket{0}^A \odot \ket{1}^A \otimes \ket{0}^B \otimes \ket{1}^C \odot \ket{1}^C.
\end{equation}

Note that $symm$ gives the same state when applying it to $\ket{1}^A \ket{0}^A \ket{0}^B \ket{1}^C \ket{1}^C$ because $\ket{0}^A \odot \ket{1}^A = \ket{1}^A \odot \ket{0}^A$. In fact, because the order of the tensor factors belonging to different systems does not matter, $\ket{0}^A \ket{1}^B = \ket{1}^B \ket{0}^A$, we can freely exchange any two tensor factors and still obtain the same state after applying $symm$. This is a particularly useful feature. In fact, it is the reason we introduce $symm$ because it allows us to ensure states are properly symmetrized without forcing us to group all messages belonging to the same wire, like we would have to if we used the symmetric tensor product $\odot$.
\end{example}

Additionally, we adopt the convention that if a wire does not show up explicitly, it contains the vacuum. We can then define the action of the four isometries via their actions on states with only a single photon


\begin{gather}
\begin{aligned}\label{eq:novel_pb}
    &\bar{V}_1 = \frac{1}{\sqrt{3}} \sum_{k_1} \ket{\psi, t=2}^{A^I_{k_1}} \ket{\emptyset, k_1}^{C_1} \ket{\Omega, t=2}^{\bar{A^I}_{k_1 +1} \bar{C}_1} \ket{\Omega, t=2}^{\bar{A^I}_{k_1 +2} \bar{C}_1} \\
    &\bar{V}_2(\bigotimes_{k_1} \ket{\psi^{k_1}, t=3}^{A^O_{k_1}} \ket{\emptyset, k_1}) = symm(\bigotimes_{k_1} V_2 (\ket{\psi^{k_1}, t=4}^{A^O_{k_1}} \ket{\emptyset, k_1})) \\
    &\bar{V}_3(\bigotimes_{k_2} \ket{\psi^{k_2}, t=5}^{A^O_{k_2}} \ket{\{k_1^{k_2}\}, k_2}) = symm(\bigotimes_{k_2} V_3 \ket{\psi^{k_2}, t=6}^{A^O_{k_2}} \ket{\{k_1^{k_2}\}, k_2}) \\
    &\bar{V}_4(\bigotimes_{k_3} \ket{\psi^{k_3}, t=7}^{A^O_{k_3}} \ket{\alpha_3^{k_3}}^{\alpha_3} \ket{\{k_1^{k_3}, k_2^{k_3}\}, k_3}) = symm(\bigotimes_{k_3} V_4 \ket{\psi^{k_3}, t=8}^{A^O_{k_3}} \ket{\alpha_3^{k_3}}^{\alpha_3} \ket{\{k_1^{k_3}, k_2^{k_3}\}, k_3}).
\end{aligned}
\end{gather}

It should be noted that $symm$ should not be viewed as an operation that is actually carried out by the process box. It is simply how we denote the fact that messages that are sent from different agents to the same agent end up in the symmetric Fock space of that agent's input space. 

In future equations, we will not explicitly write the state on any ancillary wires $\alpha_n$ to make the equations more compact. As such any state on an input or output wire should be understood as being accompanied by some state on the appropriate ancillary wire if it exists, i.e. $\ket{\psi}^{A^I_{k_n}}$ should be understood as some possibly bipartite state on the wires $A^I_{k_n}$ and $\alpha_n$ and analogously for a state on an output wire.

The action of the isometries on single photons is well defined in the QC-QC picture and we simply have to add time stamps and vacuum states to obtain the process box picture

\begin{gather}
\begin{aligned}\label{eq:novel_single}
    \bar{V}_{n+1}&(\ket{\psi, t=2n+1}^{A^O_{k_n}} \ket{\C{K}_{n-1}, k_n} \ket{\Omega, t=2n+1}^{\bar{A}^O_{k_n+1} \bar{C}_n} \ket{\Omega, t=2n+1}^{\bar{A}^O_{k_n+2} \bar{C}_n}) \\
    &= V_{n+1}(\ket{\psi}^{A^O_{k_n}} \ket{\C{K}_{n-1}, k_n}) \ket{t=2n+2} \\
    &= \sum_{k_{n+1}} V^{\rightarrow k_{n+1}}_{\C{K}_{n-1}, k_n}(\ket{\psi}^{A^O_{k_n}}) \ket{t=2n+2} \ket{\C{K}_{n-1} \cup k_n, k_{n+1}} \ket{\Omega, t=2n+2}^{\bar{A}^I_{k_{n+1}+1} \bar{A}^I_{k_{n+1}+2}  \bar{C}_{n+1}}.
\end{aligned}
\end{gather}

Additionally, if all wires contain the vacuum, $\bar{V}_{n+1}$ leaves the state invariant besides incrementing the time by 1

\begin{equation}\label{eq:isovac}
    \bar{V}_{n+1}(\bigotimes_{i=1}^3 \ket{\Omega,t=2n+1}^{\bar{A}^O_i \bar{C}_n}) = \bigotimes_{i=1}^3 \ket{\Omega,t=2n+2}^{\bar{A}^I_i \bar{C}_{n+1}}.
\end{equation}

We now wish to show that the operators we defined are indeed isometries.

\begin{restatable}[Sequence representation of the dynamical switch]{prop}{switchiso}
\label{prop:3switchiso}
The operators $\bar{V}_1, \bar{V}_2, \bar{V}_3, \bar{V}_4$ as defined by \cref{eq:novel_pb} and \cref{eq:novel_single} are isometries.
\end{restatable}

The main idea of the proof, which can be found in \cref{sec:proofs}, is to write the inner product of states in \cref{eq:novel_pb} in such a way that it is the product of the inner products of the individual messages (and not the product of the inner products of the state on each wire, which contains cross terms between the individual messages). Then, one can use that $\bar{V}_1, \bar{V}_2, \bar{V}_3, \bar{V}_4$ are isometries on one-message states which follows straightforwardly from the fact that $V_1, V_2, V_3, V_4$ are isometries in the QC-QC picture and \cref{eq:novel_single}.

\Cref{prop:3switchiso} gives us a sequence representation which means that the dynamical switch can be described by a causal box. We need to check that this causal box is also a process box. Note that $\bar{V}_1$ outputs a single-message state and $\bar{V}_2, \bar{V}_3, \bar{V}_4$ conserve the number of messages. Thus, the only way WSR could be violated is if that message is sent to the same local agent multiple times. However, the control system $\ket{\C{K}_{n-1}, k_n}$ of a message that was sent to agent $k$ at some earlier time must fulfill either $k \in \C{K}_{n-1}$ or $k = k_n$ which means the message cannot be sent to $k$ by $\bar{V}_{n+1}$. This means that the causal box given by the sequence representation $\bar{V}_1, \bar{V}_2, \bar{V}_3, \bar{V}_4$ is a process box when composed with agents obeying WSR, LO and OSR.

This process box is then indeed operationally equivalent to the dynamical switch according to our notion of equivalence given in \cref{def:pbqceq}.

\begin{restatable}[Equivalence of the QC-QC and process box description of the dynamical switch]{prop}{switchequivalence} \label{prop:switchequivalence}
The QC-QC description of the dynamical switch is operationally equivalent to its process box description. 
\end{restatable}

\subsection{Mapping general QC-QCs to process boxes}\label{sec:generalqcqctopb}

In this section, we wish to generalize our mapping from the QC-QC formulation of the dynamical switch to its process box formulation to arbitrary QC-QCs. This would then show that all QC-QCs can be viewed as process boxes.

There are two reasonable ways to go about this and they differ in how we model the control system as well as how we deal with multi-message states. The first way is a generalization of what we did for the dynamical switch. In this case, the control system is a part of the target, and the process box acts independently on each message of a multi-message state.

The second way considers the control system to be part of an internal memory. The control system then tells the process box on which wire it should expect an incoming message in the next time step based on which wires it sent messages on in the previous time step. If this expectation does not match what actually arrives on the wires, the process box aborts. This happens in particular when multiple messages arrive at the same time.

\subsubsection{Control system as part of the target system}\label{sec:controlintarget}


When we constructed a process box formulation of the dynamical switch, we actually did not use many features that are unique to this specific QC-QC. So, for the most part, we just have to generalize from the 3-partite to the $N$-partite case. We introduce the set $\C{T} = \{1,...,2N+2\}$ which splits into $\C{T}^O = \{1,3,...,2N+1\}$ and $\C{T}^I = \{2,4,...,2N+2\}$. Once again, $t=1$ is a valid output time only for the global past and $t=2N+2$ is a valid input time only for the global future. Given the set of operators of a QC-QC

\begin{gather}
\begin{aligned}\label{eq:Kraus}
    V^{\rightarrow k_1}_{\emptyset, \emptyset}&: \C{H}^{P} \rightarrow \C{H}^{A^I_{k_1} \alpha_1} \\
    V^{\rightarrow k_{n+1}}_{\C{K}_{n-1}, k_n}&: \C{H}^{A^O_{k_n} \alpha_n} \rightarrow \C{H}^{A^I_{k_{n+1}} \alpha_{n+1}} \\
    V^{\rightarrow k_N}_{\C{N}\backslash k_N, k_N}&: \C{H}^{A^O_{k_N} \alpha_N} \rightarrow \C{H}^{F \alpha_F},
\end{aligned}
\end{gather}
the basic idea is again that we use these to define the single-message action of the isometries of the sequence representation by adding time stamps and the vacuum state and then declare that the isometries essentially act independently on each individual non-vacuum state in a multi-message state. Up to this point, we have


\begin{gather}
    \bar{V}_1: \C{H}^P \otimes \ket{t=1} \rightarrow \bigotimes_{k_1} \C{F}(\C{H}^{A^I_{k_1} \alpha_1 C} \otimes \ket{t=2}) \nonumber \\
    \bar{V}_{n+1}: \bigotimes_{k_n} \C{F}(\C{H}^{A^O_{k_n} \alpha_n C} \otimes \ket{t=2n+1}) \rightarrow \bigotimes_{k_{n+1}} \C{F}(\C{H}^{A^I_{k_{n+1}} \alpha_{n+1} C} \otimes \ket{t=2n+2}) \\
    \bar{V}_{N+1}: \bigotimes_{k_N} \C{F}(\C{H}^{A^O_{k_N} \alpha_N C} \otimes \ket{t=2N+1}) \rightarrow \bigotimes_{k_{n+1}} \C{F}(\C{H}^{F \alpha_F} \otimes \ket{t=2N+2}). \nonumber
\end{gather}

The action on states on the output wire is then

\begin{gather}
\begin{aligned}\label{eq:pb_ext}
    \bar{V}_1 \ket{\psi, t=1}^P &= V_1 \ket{\psi}^P \ket{t=2} = \sum_{k_1} V^{\rightarrow k_1}_{\emptyset, \emptyset} \ket{\psi}^P \ket{t=2} \ket{\emptyset, k_1} \\
    \bar{V}_{n+1}(\bigotimes_{k_n} \ket{\psi^{k_n}, t=2n+1}^{A^O_{k_1}} \ket{\C{K}^{k_n}_{k_{n-1}}, k_n}) &= symm(\bigotimes_{k_n} \bar{V}_{n+1} (\ket{\psi^{k_n}, t=2n+1}^{A^O_{k_n}} \ket{\C{K}^{k_n}_{n-1}, k_n})) \\
    &= symm(\bigotimes_{k_n} \sum_{k^{k_n}_{n+1}} V^{\rightarrow k^{k_n}_{n+1}}_{\C{K}^{k_n}_{n-1}, k_n} \ket{\psi^{k_n}}^{A^O_{k_n}} \ket{t=2n+2} \ket{\C{K}^{k_n}_{n}, k_{n+1}^{k_n}}) \\
    \bar{V}_{N+1}(\bigotimes_{k_N} \ket{\psi^{k_N}, t=2N+1}^{A^O_{k_N}} \ket{\C{N} \backslash k_N, k_N})
    &= symm(\bigotimes_{k_N} \bar{V}_{N+1} \ket{\psi^{k_N}, t=2N+1}^{A^O_{k_N}} \ket{\C{N} \backslash k_N, k_N} \\ &= \bigodot_{k_n} V^{\rightarrow F}_{\C{K}_{\C{N} \backslash k_N}, k_N} \ket{\psi^{k_N}} \ket{t=2N+2}
\end{aligned}
\end{gather}
where $\C{K}^{k_n}_n \coloneqq \C{K}^{k_n}_{n-1} \cup k_n$. For the 0-message case, we use \cref{eq:isovac} as the definition, replacing the tensor product over 3 agents with one over $N$ agents.

Note that we use the symmetric tensor product directly instead of using the more complicated $symm$ in the last line, because all the states are sent to the same agent, the global future, in the last time step.

We now wish to prove that these maps actually define isometries.

\begin{restatable}[Sequence representation for QC-QCs]{prop}{symmiso}
\label{prop:symmiso}
The maps $\bar{V}_1,...,\bar{V}_{N+1}$ as defined by \cref{eq:pb_ext} are isometries.
\end{restatable}

In order to prove this, we will need the following lemma.

\begin{restatable}[Isometry condition for internal operations]{lemma}{krausiso}
\label{lemma:krausiso}
For $\C{K}_{n-1}, \C{L}_{n-1} \subsetneq \C{N}$ and $k_n \in \C{N} \backslash \C{K}_{n-1}, l_n \in \C{N} \backslash \C{L}_{n-1}$ with $\C{K}_{n-1} \cup k_n = \C{L}_{n-1} \cup l_n$, it holds that

\begin{equation}
\sum_{k_{n+1}} V^{\rightarrow k_{n+1} \dagger}_{\C{K}_{n-1}, k_n} V^{\rightarrow k_{n+1}}_{\C{L}_{n-1}, l_n} = \mathbb{1}^{A^O_{k_{n}} \alpha_n} \delta_{\C{K}_{n-1}, \C{L}_{n-1}} \delta_{k_n, l_n}.
\end{equation}
\end{restatable}

This is a direct consequence of the operators $V_n$ being isometries in the QC-QC picture. The lemma is necessary to get rid of certain cross terms in the inner product that appear as a result of having to symmetrize when there are multiple messages on a single wire.

Given \cref{prop:symmiso}, what we constructed is therefore a causal box and it is also a process box by the same argument we gave for the dynamical switch.

Finally, the process box we constructed is also equivalent to the original QC-QC.

\begin{restatable}[Process box description of QC-QCs]{prop}{symmequivalence} \label{prop:symmequivalence}
The QC-QC described by the process vector $\ket{w_{\C{N},F}} = \dket{V_1} * ... * \dket{V_N}$ is equivalent to the process box whose sequence representation is given by the isometries $\bar{V}_1,..., \bar{V}_{N+1}$ from \cref{eq:pb_ext}.
\end{restatable}

The proof is completely analogous to the proof of \cref{prop:switchequivalence}. Generalizing from the 3-partite to the $N$-partite case does not change anything about the proof. Additionally, while we considered a trivial global past in the case of the dynamical switch, we can reduce to that case even if the global past is non-trivial. There is an obvious isomorphism between the global past in the QC-QC and in the process box picture that consists of appending/dropping the time stamp $\ket{t=1}$. We can then compose the Choi/process vector with an arbitrary element of the respective global past and then proceed as if the global past were trivial.

\subsubsection{Control system as an internal memory}\label{sec:internalcontrol}

In the previous part as well as in our example of the dynamical switch, we explicitly constructed extensions to allow for multiple messages being sent to the process box at the same time. However, in the proofs (see \cref{sec:proofs}) of the various propositions we also see that when the process box is actually composed with local agents, there is at any time at most one non-vacuum state on the input wires of the process box. One could therefore take the position that what happens with multiple messages does not matter. Our goal for this section is therefore to construct a process box that is only defined on states that can actually occur after composition and that aborts if it receives some other state.

An additional advantage, which this formulation will have, is that we can turn the control and ancillary systems into an internal wire and do not have to view it as part of the system. The isometries are then

\begin{gather}\label{eq:abortiso}
\begin{aligned}
    \bar{V}_1&: \C{F}(\C{H}^{P} \otimes \ket{t=1}) \rightarrow \bigotimes_{k_1} \C{F}(\C{H}^{A^I_{k_1}} \otimes \ket{t=2}) \otimes \C{H}^{\bar{\alpha}_1} \otimes \C{H}^{\bar{C}_1} \\
    \bar{V}_{n+1}&: \bigotimes_{k_n} \C{F}(\C{H}^{A^O_{k_n}} \otimes \ket{t=2n+1}) \otimes \C{H}^{\bar{\alpha}_{n}} \otimes \C{H}^{\bar{C}_n}  \rightarrow \bigotimes_{k_{n+1}} \C{F}(\C{H}^{A^I_{k_{n+1}}} \otimes \ket{t=2n+2}) \otimes \C{H}^{\bar{\alpha}_{n+1}} \otimes \C{H}^{\bar{C}_{n+1}} \\
    \bar{V}_{N+1}&: \bigotimes_{k_N} \C{F}(\C{H}^{A^O_{k_N}} \otimes \ket{t=2N+1}) \otimes \C{H}^{\bar{\alpha}_{N}} \otimes \C{H}^{\bar{C}_N} \rightarrow \C{F}(\C{H}^{F} \otimes \ket{t=2N+2}) \otimes \C{H}^{\bar{\alpha}_F}.
\end{aligned}
\end{gather}

It is then possible that the process box receives invalid inputs even if only one message is sent. That is the control system has the form $\ket{\C{K}_n, k_{n+1}}$, but there is an input on a wire that does not correspond to $k_{n+1}$. This will actually be very common if we do not compose the system. Imagine a process box that during the first time step simply forwards the target system in a superposition to every lab. The control system is then $\sum_{k_1 \in \C{N}} \ket{\emptyset, k_1}$. If we then input a non-vacuum state $\ket{\psi}$ on a single wire (say wire 1, for example) in the next time step, the isometry $V_2$ will receive $\sum_{k_1 \in \C{N}} \ket{\psi, t=3}^{\bar{A}^O_1} \otimes \ket{\emptyset, k_1}$. Only the term corresponding to $k_1 = 1$ is valid. How should we deal with the other terms? A reasonable solution is to say the process box aborts when it encounters an unexpected input, i.e. it outputs the state $\ket{abort}$. More precisely, before applying the isometry $V_{n+1}$, the circuit applies a projective measurement $\{\C{P}^{n}_{accept}, 1 - \C{P}^n_{accept}\}$, with outcomes $\{accept, abort\}$, where $\C{P}^n_{accept}$ projects on the subspace

\begin{equation}\label{eq:accept}
    \C{H}^n_{accept} = \text{Span}[\{\ket{\psi, t=2n+1}^{A^O_{k_n}} \otimes \bigotimes_{i \in \C{N}\backslash k_n} \ket{\Omega, t=2n+1}^{\bar{A}^O_i} \otimes \ket{\alpha}^{\alpha_n} \ket{\C{K}_{n-1}, k_n}\}_* \cup \ket{\Omega}^{\bar{A}^{O}_\C{N} \bar{\alpha}_n \bar{C}_n}].
\end{equation}


The asterisk ($*$) is a placeholder symbol for the conditions on this set, namely $\C{K}_{n-1} \subseteq \C{N}, k_n \in \C{N} \backslash \C{K}_{n-1}, \ket{\psi} \in \C{H}^{A^O_{k_n}}, \ket{\alpha} \in \C{H}^{\alpha_n}$.

The accepted states are thus exactly those that have a single non-vacuum state on a single wire with a control bit that corresponds to that wire.

If the outcome is $abort$, the circuit outputs the state $\ket{abort}$. If the outcome is $accept$, the circuit applies the isometry $\bar{V}_{n+1}$. We define their action on states in $\C{H}^n_{accept}$ to be

\begin{gather}
\begin{aligned}\label{eq:abortdef}
    \bar{V}_{n+1}&\ket{\psi, t=2n+1}^{A^O_{k_n}} \otimes \bigotimes_{i \in \C{N}\backslash k_n} \ket{\Omega, t=2n+1}^{\bar{A}^O_i} \otimes \ket{\alpha}^{\alpha_n} \otimes \ket{\C{K}_{n-1}, k_n} \\
    &= \sum_{k_{n+1}} V^{\rightarrow k_{n+1}}_{\C{K}_{n-1}, k_n} (\ket{\psi}^{A^O_{k_n}} \otimes \ket{\alpha}^{\alpha_n}) \otimes \ket{t=2n+2} \otimes \bigotimes_{i \in \C{N}\backslash k_{n+1}} \ket{\Omega, t=2n+2}^{\bar{A}^I_i} \otimes \ket{\C{K}_{n-1} \cup k_n, k_{n+1}}
\end{aligned}
\end{gather}
if $\psi \neq \Omega$ and analogously for the edge cases $\bar{V}_1$ and $\bar{V}_{N+1}$. This is essentially the action that these isometries have in the QC-QC picture once we add appropriate time stamps and vacuum states. If $\psi = \Omega$, then we define the action similar to \cref{eq:isovac}

\begin{equation}
    \bar{V}_{n+1}(\bigotimes_{i=1}^N \ket{\Omega,t=2n+1}^{\bar{A}^O_i} \ket{\Omega}^{\bar{\alpha}_n} \ket{\Omega}^{\bar{C}_n}) = \bigotimes_{i=1}^N \ket{\Omega,t=2n+2}^{\bar{A}^I_i} \ket{\Omega}^{\bar{\alpha}_{n+1}} \ket{\Omega}^{\bar{C}_{n+1}}.
\end{equation}

From the above equations one can see that the image of $V_{n+1}$ lies in 

\begin{equation}
    \C{H}^{n+1}_{accepted} = \text{Span}[\{\ket{\psi, t=2n+1}^{A^I_{k_{n+1}}} \otimes \bigotimes_{i \in \C{N}\backslash k_{n+1}} \ket{\Omega, t=2n+2}^{\bar{A}^I_i} \otimes \ket{\alpha}^{\alpha_{n+1}} \otimes \ket{\C{K}_{n}, k_{n+1}}\}_* \cup \ket{\Omega}^{\bar{A}^{I}_\C{N} \bar{\alpha}_{n+1} \bar{C}_{n+1}}].
\end{equation}


This time the asterisk should be interpreted as the conditions $\C{K}_{n} \subseteq \C{N}, k_{n+1} \in \C{N} \backslash \C{K}_{n}, \ket{\psi} \in \C{H}^{A^I_{k_n}}, \ket{\alpha} \in \C{H}^{\alpha_{n+1}}$.

Note that the measurement is perfectly non-disturbing, and the outcome is always $accept$ on states obtained from composition with local operations. The reason for this is that if the internal (local) operations receive a state in $\C{H}^n_{accept}$ ($\C{H}^{n+1}_{accepted}$), they output a state in $\C{H}^{n+1}_{accepted}$ ($\C{H}^{n+1}_{accept}$). This is formalized in the proposition below.

\begin{restatable}[No aborts]{lemma}{accept} \label{lemma:accept}
If the process box given by the sequence representation $\bar{V}_1,...,\bar{V}_{N+1}$ from \cref{eq:abortdef} is composed with local operations obeying WSR, LO and OSR, then the projective measurements yield the outcome $accept$ with unit probability.
\end{restatable}

We note that the isometries are thus only isometries on valid states. However, if we compose the process box with local operations, there are only ever valid states on the wires. 

Finally, we can show that this process box indeed corresponds to the QC-QC.


\begin{restatable}[Process box description of QC-QCs: Abort extension]{prop}{abortequivalence} \label{prop:abortequivalence}
The QC-QC described by the process vector $\ket{w_{\C{N},F}} = \dket{V_1} * ... * \dket{V_N}$ is equivalent to the process box whose sequence representation is given by the isometries $\bar{V}_1,..., \bar{V}_{N+1}$ from \cref{eq:abortdef}.

\end{restatable}

\subsection{Resolving the composability problem}\label{sec:resolvingcomp}

In an experimental setting, systems are usually composable \cite{Procopio_2015,Rubino_2017}. We can take the output of one experiment and use it as the input of another or intervene at some point to add another input. Meanwhile, process matrices in general are not composable \cite{Gu_rin_2019}. This can be demonstrated by an example, taken from \cite{Gu_rin_2019} and depicted in \cref{fig:pmcomp}. We consider two processes, $W$ with input wires $A^O, B^O$ and output wires $A^I, B^I$ and $W'$ with input wires $A'^O, B'^O$ and output wires $A'^I, B'^I$. Both of these processes are fixed order processes with $W$ sending a state $\rho$ first to $A^I$ and then sending the input from $A^O$ to $B^I$ and $W'$ essentially doing the reverse, sending $\rho$ to $B'^I$ and sending the output from $B'^O$ to $A'^I$.

\begin{figure}
    \centering
    \includegraphics[width=0.5\textwidth]{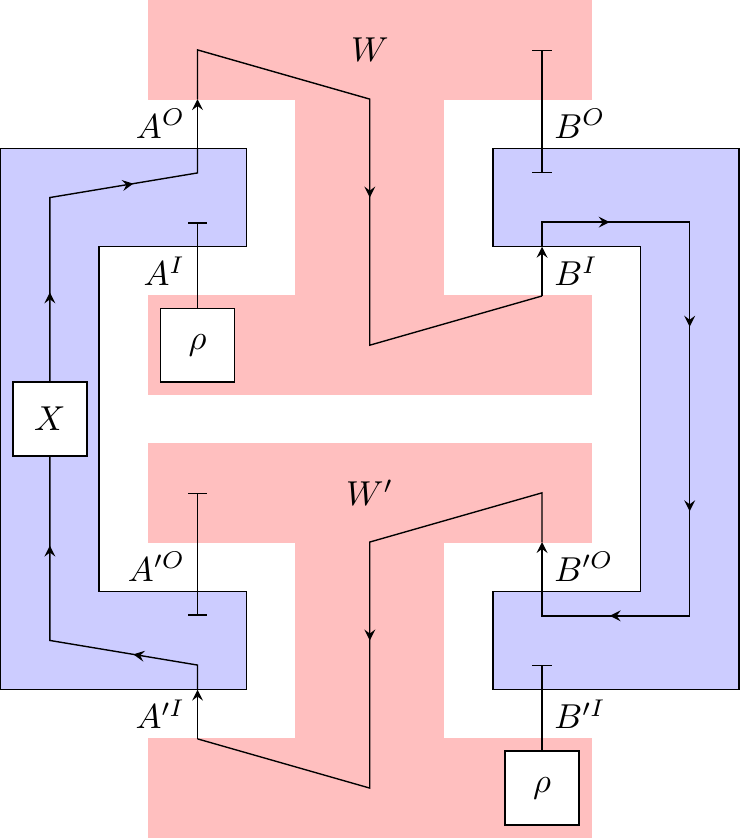}
    \caption{Composition of two process matrices $W$ and $W'$. $W$ sends a state to Alice and then sends her output to Bob, while $W'$ does the same but with the order of Alice and Bob reversed. The composition then includes a loop and if Alice applies the Pauli-$X$ unitary while Bob applies the identity the generalized Born rule yields a probability of 0, even though it should have yielded 1 if this composition were well-defined.}
    \label{fig:pmcomp}
\end{figure}

We then consider agents Alice and Bob for the composition of $W$ and $W'$. Alice's input wire is the combined system $A^I A'^I$ while her output is $A^O A'^O$. Similarly, Bob's input is $B^I B'^I$ and his output is $B^O B'^O$.

Alice then applies the Pauli-$X$ unitary to the state she receives on the wire $A'^I$ and sends the result along $A^O$, denoted in \cref{fig:pmcomp} with an $X$. Bob, meanwhile, simply sends the state he receives on wire $B'^I$ back on wire $B^O$. As clearly seen from the figure, this choice of local operations leads to a loop. The probability according to the generalized Born rule is not 1 as should be the case for a deterministic operation but 0. The reason is that the Pauli-$X$ unitary makes it so the value of the qubit on $A^O$ is bit-flipped compared to the qubit on $A'^O$, but, on the other hand, the two qubits are also supposed to be the same if one follows the path from $A^O$ over Bob to $A'^I$. The composition of $W$ and $W'$ is therefore not a valid process matrix.

Our result that QC-QCs are process boxes now resolves this problem of composability. Let us consider $W$ and $W'$ as process boxes. In this case, each would send a state $\rho$ at the earliest time, which is $t=2$ with our convention (as the global past is trivial), to one of the two agents, $W$ to Alice, $W'$ to Bob. The agents then discard this system. The loop, meanwhile, is never actually realized as there is no initial input to it. Even if we assumed that some non-vacuum state, let us say the state $\ket{0}$, arrived somehow on Alice's wire $A'^I$, this would not be an issue. The state would be associated with some time stamp $t_1$, $\ket{0, t_1}$ and Alice would send her output at $t_2 = \C{O}_A(t_1) > t_1$, $\ket{1, t_2}$. The process box $W$ then sends this state from Alice to Bob, i.e. Bob receives the state $\ket{1, t_3}$ with $t_3 \geq t_2$. Bob then applies his identity channel outputting $\ket{1, t_4}$ with $t_4 = \C{O}_B(t_3) > t_3$. The process box $W'$ then sends this state back to Alice and the whole cycle repeats. 

We thus see that the composition of process boxes may also lead to cycles, but these are well-defined. Thanks to adding time stamps, they simply correspond to sending a system back and forth between agents indefinitely. Such a description would not be possible without the use of Fock spaces (or a similar space that models sending of multiple messages), justifying the effort we put into working out the technicalities of these spaces and extensions of QC-QCs to Fock spaces. However, note that if such loops show up, the resulting composition is not a process box as the same agent receives a non-vacuum state multiple times, violating WSR. It is thus exactly these assumptions, WSR, LO and OSR, which we explicitly stated for the process box framework, but which are to some degree implicit in the process matrix framework that led to process matrices and process boxes not being closed under composition. Specifically, we can say that in the process matrix framework for each local agent a single time (or a superposition of times) is chosen, and the local agent applies their operations only at that time (cf. \cref{fig:qcqcstatespaces}). This works as long as WSR is satisfied but becomes a problem when there are multiple messages or even when a single message is sent to the same agent multiple times. This is then the source of the composability issues outlined in this section and \cite{Gu_rin_2019}. On the other hand, the process box treats all time steps equally a priori and only because we impose WSR on the statespace can we say that the action of the agents is effectively equivalent to the one in the QC-QC picture (cf. \cref{fig:pbstatespaces}). We can thus recover composability in the process matrix framework by recognizing that arbitrary compositions of process boxes can take us outside the WSR restricted statespace and by, if necessary, defining an appropriate extension. The boxes are still well-defined, composable physical systems.

In conclusion, we know now how to define composition, even if not for the full process matrix framework, then at least for QC-QCs. If $W$ and $W'$ are QC-QCs, we convert them to process boxes with an appropriate extension. We detailed two such extensions, but others are possible. Which extension is the appropriate one will ultimately depend on the experimental setup under consideration. The two process boxes operationally equivalent to $W$ and $W'$ are then composed according to the composition rules for causal boxes. This is always possible as causal boxes are closed under arbitrary composition \cite{Portmann_2017}, but the resulting system is not necessarily a process box as this subset of causal boxes does not have this feature \cite{Vilasini_2020}.











\section{Characterizing process boxes}\label{sec:characterizing}

\subsection{Process boxes are unitarily extensible}\label{sec:unitary}

Unitarity is a useful property in quantum information theory due to its connection to the reversibility of quantum mechanics \cite{Chiribella_2010}. Thanks to the Stinespring dilation \cite{Stinespring_1955}, any completely positive map can be regarded as a unitary operator by considering additional ancillary systems. We say they are purifiable or unitarily extensible. It is then natural to ask whether process matrices and causal boxes are also unitarily extensible. 

This question becomes even more interesting as it is argued in \cite{Ara_jo_2017} that unitary extensibility is a necessary condition for physical processes, noting that not all processes have this property. Therefore, we would expect that process boxes admit a unitary extension as we believe them to be physically implementable by virtue of being special cases of causal boxes, which model physical scenarios \cite{Portmann_2017, Vilasini_2020}.

Let us first define what it means for these objects to be unitary. We consider a process to be unitary if the supermap it is associated with yields unitary maps whenever it acts on agents applying unitary local operations, even if these agents have arbitrary ancillaries.


\begin{defi}[Unitary process matrices \cite{Ara_jo_2017}]
Let $W \in \C{L}(P A^{IO}_{\C{N}} F)$ be a process matrix. We call $W$ unitary if for all ancillaries $A'^{I/O}_n$, $n=1,...,N$ and all local operations $\C{M}_{A_n}: \C{L}(\C{H}^{A^I_n A'^I_n}) \rightarrow \C{L}(\C{H}^{A^O_n A'^O_n})$ that are unitary and CPTP, the induced map

\begin{equation}
    (M_{A_1} \otimes ... \otimes M_{A_N}) * W
\end{equation}

is unitary and CPTP.
\end{defi}

Based on this, we postulate an analogous definition for unitarity of process boxes.

\begin{defi}[Unitary process boxes]
Let $\Phi \in \C{L}(\C{H}^{\textup{eff}})$ be the effective Choi matrix of a process box. We call $\Phi$ unitary if for all ancillaries $A'^{I/O}_n$, $n=1,...,N$ and all local operations $\C{M}_{\bar{A}_n}: \C{L}(\C{F}_{A^I_n A'^I_n}) \rightarrow \C{L}(\C{F}_{A^O_n A'^O_n})$ that are unitary and CPTP, the induced map

\begin{equation}
    (M_{\bar{A}_1} \otimes ... \otimes M_{\bar{A}_N}) * \Phi
\end{equation}

is unitary and CPTP.

\end{defi}

\begin{figure}
    \centering
    \includegraphics[width=\textwidth]{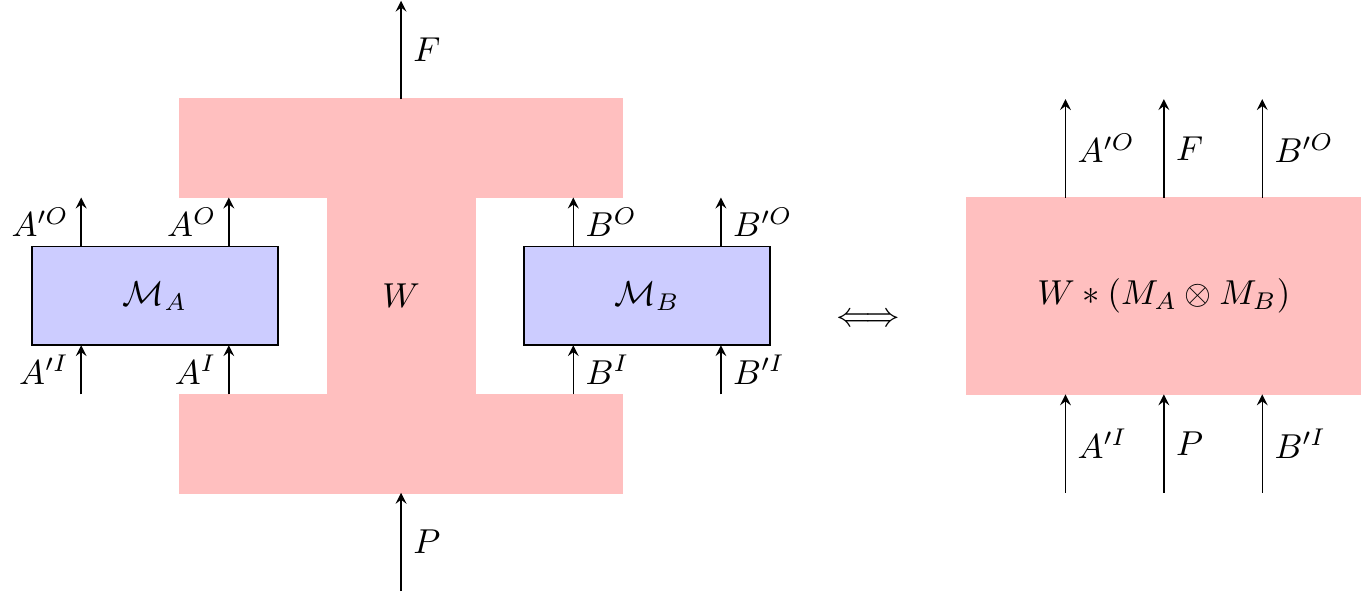}
    \caption{Composing the process with local operations with arbitrary ancillaries (primed spaces) leads to a map from (in the bipartite case depicted here) $P A'^I B'^I$ to $F A'^O B'^O$. If this map is unitary for all unitary local operations, then we say that the process is unitary.}
    \label{fig:pmwithancillary}
\end{figure}

Note that the ancillary $A'^I_n$ can be viewed as a generalization of the measurement setting $a$ while the ancillary $A'^O_n$ can be viewed as a generalization of the measurement outcome $x$ (cf. \cref{def:locallabs}). In particular, we can use these ancillaries to write a quantum instrument as a CPTP map as we did when introducing the agents in the process box picture.

There is an alternate and equivalent definition. We can view process matrices and process boxes as channels from their inputs to their outputs (or equivalently from the outputs to the inputs of the local agents). A process matrix or causal box is then unitary if that channel is unitary.

\begin{lemma}[Unitarity condition \cite{Ara_jo_2017}]
A process matrix $W$ is unitary iff $W = \dket{U} \dbra{U}$ for some unitary $U$. A causal box is unitary iff its Choi representation $\Phi$ can be written as $\Phi = \dket{\bar{U}} \dbra{\bar{U}}$ for some unitary $\bar{U}$.
\end{lemma}



The equivalence of these two definitions is proven in the case of process matrices in \cite{Ara_jo_2017}. Their proof also works for process boxes. From this we can define unitary extensibility.

\begin{defi}[Unitary extensibility of process matrices \cite{Ara_jo_2017}]\label{def:pm_purification}
A process with process matrix $W$ is unitarily extensible if there exist Hilbert spaces $\C{H}^{P'}$ and $\C{H}^{F'}$ and a unitary  process with process matrix $W'$ such that $W = W' * (\ket{0}\bra{0}^{P'} \otimes \mathbb{1}^{F'})$. 
\end{defi}

The Hilbert spaces $\C{H}^{P'}$ and $\C{H}^{F'}$ can be viewed as additional global past and future spaces.

Using these definitions, Araújo et al. \cite{Ara_jo_2017} then show that there are process matrices that are not unitarily extensible. Interestingly, while it is generally possible to find some unitary dilation, that unitary dilation is not always a valid process itself because of the additional assumptions imposed on process matrices to guarantee normalized and positive probabilities. 

When translating \cref{def:pm_purification} to process boxes, we would like $P'$ and $F'$ to actually be a global past and a global future. That is to say, we want the time stamp of the output of $P'$ to be earlier or equal to all time stamps in $\C{T}$ and the time stamp of the input to $F'$ to be later or equal to all time stamps in $\C{T}$.

\begin{defi}[Unitary extensibility of process boxes]
\label{def:pb_purification}
A process box described by an effective Choi matrix $\Phi$ with set of positions $\C{T}$ is unitarily extensible if there exist Hilbert spaces $\C{H}^{P'}$ and $\C{H}^{F'}$ and a process box with effective Choi matrix $\Phi'$ such that $\Phi = \Phi' * (\ket{0, t_i}\bra{0, t_i}^{P'} \otimes \mathbb{1}^{F'} \otimes \proj{t_f}{t_f})$ such that $t_i \leq t \leq t_f$ for all $t \in \C{T}$. 
\end{defi}


This definition then allows us to show what we already suspected, that process boxes are unitarily extensible.

\begin{restatable}{prop}{PBunitary} \label{prop:PBunitary}
Process boxes are unitarily extensible.
\end{restatable}

The proposition is essentially a consequence of the existence of the sequence representation. 

\begin{remark}[Unitary extensions of general causal boxes]
We could also consider general causal boxes and ask if they admit a unitary extension. On the one hand, the existence of the sequence representation makes this statement somewhat trivial as we can simply unitarize the isometries. However, we cannot necessarily assign the global past and future singular position labels as this might then lead to an infinite number of outputs during these positions.
\end{remark}

Finally, for a QC-QC and process box that are operationally equivalent, their unitary extensions can be chosen such that they are operationally equivalent as well.

\begin{restatable}[Operational equivalence of unitary extensions]{coro}{unitaryequivalence}
\label{coro:unitaryequivalence}
Let $\dket{\bar{V}}$ be the effective Choi vector of a process box and $\ket{w_{\C{N}, F}}$ the process vector of a QC-QC such that they are operationally equivalent. If $\dket{U}$ is a unitary extension of $\ket{w_{\C{N}, F}}$, then there exists an effective Choi vector $\dket{\bar{U}}$ that is operationally equivalent to $\dket{U}$ and is a unitary extension of $\dket{\bar{V}}$.
\end{restatable}

\subsection{Simplifying the set of positions}\label{sec:simplifying}

The process box framework, like the more general causal box framework, assumes the existence of some background spacetime which is described by the set $\C{T}$. Meanwhile, QC-QCs strictly speaking do not assume the existence of such a set but are intuitively interpreted as applying a single operation during each time step of a totally ordered time set. When mapping QC-QCs to process boxes, we simply picked a set $\C{T}$ to match this intuition. However, if we try to go in the other direction, it is less clear how we should deal with $\C{T}$. How does the QC-QC we construct from a process box depend on the process box's set of positions $\C{T}$? 

In order to answer this question, we need to first understand the role that $\C{T}$ plays in the process box framework more thoroughly. As we discussed in \cref{sec:fock}, it orders the messages and thus allows us to define a notion of causality. The causality condition ultimately depends on the order relations between the position labels and not on the labels themselves. Simply relabeling the position labels while not changing any of the order relations should therefore not change the process box in an operational sense. This is also what we expect from special relativity. The labels correspond to spacetime coordinates, which are frame dependent, while the order relations correspond to the frame-independent light cone structure. We then expect physics to be the same in all inertial reference frames, which means that the operational predictions should not depend on the labels. Different labels could also refer to the same reference frame but be due to agents using different units to measure the same spacetime distances. This should also not affect any final operational results, beyond needing to convert them between the different units.

Therefore, we see that there are different process boxes that are operationally equivalent. We can thus define a notion of operational equivalence between process boxes just as we did in \cref{sec:statespace} between a process box and a QC-QC.

\begin{defi}[Equivalence of process boxes]\label{def:pb_eq}
Let $\hat{\Phi}$ and $\hat{\Phi}'$ be two $N$-partite process boxes. We say that $\hat{\Phi}$ and $\hat{\Phi}'$ are operationally equivalent if there exists for each agent $A_i$ of $\hat{\Phi}$ a unique agent $A'_i$ of $\hat{\Phi}'$ and a map that relates the possible local operations of the agents

\begin{equation}
    \C{M}_{\bar{A}_i} \mapsto \C{M}'_{\bar{A}'_i}
\end{equation}

such that 

\begin{equation}
    \bigotimes_{i=1}^N M_{\bar{A}_i} * \Phi \cong \bigotimes_{i=1}^N M'_{\bar{A}'_i} * \Phi' 
\end{equation}

where $\cong$ should be understood as a unitary isomorphism between the global past spaces as well as between the global future spaces of the process boxes.

\end{defi}

\begin{remark}[Relating process box local operations]
We note that the mapping between local operations does not have to be an isomorphism. For example, consider adding an input time $0^I$ and an output time $0^O$ to a process box such that $0^I < 0^O$ and $\C{O}(0^I) = 0^O$ and both are in the past of all other time stamps. The new process box just sends the vacuum during $0^I$ and as such what the agents do during $0^I$ does not actually matter because they effectively always apply the vacuum projector. We can thus identify any operation $\C{M}_{A_i}^{0^I} \otimes \C{M}_{A_i}$ with $\C{M}_{A_i}$ in the original process box without the new time steps. 
\end{remark}

The above definition then essentially defines equivalence classes of process boxes that all produce the same measurement statistics, or more generally have the same (up to unitary isomorphisms) behavior as supermaps. We can then formalize the previous idea about relabeling not mattering in the following lemma.

\begin{restatable}[Invariance of process boxes under relabeling]{lemma}{relabeling}
\label{lemma:relabeling}
Let $\hat{\Phi}$ be a process box with positions $\C{T}$. Let $\C{T}'$ be another set positions and $\C{R}: \C{T} \rightarrow \C{T}'$ an isomorphic map that respects the order relations of $\C{T}$, i.e. for any $t_1, t_2 \in \C{T}$ with $t_1 < t_2$, it holds that $\C{R}(t_1) < \C{R}(t_2)$. Define then the process box $\Phi'$ via the relabeling of position labels with the help of $\C{R}$ from $\Phi$, that is we define

\begin{equation}
    \hat{\Phi}' \coloneqq \C{R} \circ \hat{\Phi} \circ \C{R}^{-1}
\end{equation}

where $\C{R}(\ket{\psi, t}) = \ket{\psi, \C{R}(t)}$ and analogously for multi-message states.

Then $\hat{\Phi}$ and $\hat{\Phi}'$ are operationally equivalent.

\end{restatable}

However, it is not just relabeling that produces operationally equivalent process boxes. For example, we could model a process where each agent receives a state $\ket{\psi}$ and there is no way for the agents to influence each other using several different sets $\C{T}$: the process box could send $\ket{\psi}$ to all agents at the same time, it could send at different spatially separated spacetime points, or it could send the state $\ket{\psi}$ to one agent after another with the outputs of the agents being sent straight to the global future. 

In particular when considering the last implementation, we see that there is a lot of redundancy in the process box framework. The structure of $\C{T}$ has an impact on which messages can influence each other but so do the internal operations of the process box. Furthermore, the position label can also be used to send information in addition to the message itself. A simple example is a setup where all Hilbert spaces are trivial and $\C{T} \subseteq \mathbb{N}$. In this case, the parity of $t \in \C{T}$ could be used to encode a bit. We could thus also describe such a process with a process box with half as many time steps that uses qubits. 





Another example of redundancy in the process box framework is the possibility of time-dependent action. As a simple example, consider a process using qubits as the messages and $\C{T}^I_i = \{1,...,M\}$ for all agents. An agent could then decide to measure in the computational basis if they receive a message at an even time $t$ while they measure in the plus/minus basis if the message's time stamp is odd. However, since the agent has no influence over when they receive a message, the same could also be achieved by a time-independent measurement if the process box sends an additional bit which tells the agent what measurement to apply. More generally, if $|\C{T}^I_i| = M$, then an agent could have $M$ different operations based on the time stamp of the message. This could be alternatively realized by a time-independent action if the process box sends an additional $M$-dimensional system. 

The question is now can we restrict $\C{T}$ to some less general form while still being able to describe all process boxes in the sense of \cref{def:pb_eq}.

We make the following conjecture.

\begin{restatable}[Simplifying the set of positions]{conj}{simplifying}
\label{conjecture:simplifying}

Any process box $\Phi$ with set of positions $\C{T}$, input positions $\C{T}^I_i$ and output positions $\C{T}^O_i$ and functions $\C{O}_i: \C{T}^I_i \rightarrow \C{T}^O_i$ is equivalent to a process box $\Phi'$ whose set of positions has the following properties:

\begin{enumerate}
    \item $\C{T}' = \{1,2,...,M\}$
    \item $\C{T}'^I_i = \C{T}'^I = \{2,4,...\}$ and $\C{T}'^O_i = \C{T}'^O = \{3,5,...\}$ for agents with non-trivial input and $\C{T}'^O_i = \{1\}$ for agents with trivial input
    \item $\C{O}(t) = \C{O}'_i(t) = t+1$ for all $t \in \C{T}'^I$
\end{enumerate}

\end{restatable}

We do not give a full proof of this statement but give some additional justification in \cref{sec:justification}, which should contain the main ideas for a proof in a future work.

\begin{remark}[Physical significance of redundancy]
Note that the above conjecture does not imply that process boxes can only describe settings as described in the conjecture but that any setting that can be described by process boxes can be identified with the setting in the conjecture. This is somewhat reminiscent of the quantum switch, which can either be implemented gravitationally \cite{Zych_2019}, using a superposition of spacetimes (complicated spacetime, simple internal operations), or optically \cite{Taddei_2021}, with a control bit (simple spacetime, complicated internal operation). 

Furthermore, the redundancy of the process box framework we discussed in this section and used to formulate the conjecture, is something we would expect physically. A given abstract information processing protocol may be physically implementable in different ways. The process box framework would then give different descriptions to the different implementations as can be seen from some of the examples in this section. Note also that if we then switch to a different context, the details that were previously redundant may become relevant and the different implementations may then behave differently from each other. For instance, the different causal box extensions of a QC-QC are operationally equivalent when considering composition with agents obeying WSR, LO and OSR, but in other scenarios, for example when considering composition of processes, the choice of extension matters.
\end{remark}

\section{Mapping process boxes to QC-QCs}\label{sec:pbtoqcqc}
\subsection{Role of control}\label{sec:role}

If a mapping from process boxes to QC-QCs exists, this would mean that process boxes can only model controlled superpositions of orders. However, a priori it is not clear why this would be the case. On the other hand, in \cite{Costa_2022} it was shown that pure superpositions between causal orders cannot exist (at least for certain classes of process matrices). This means that the control bit in the quantum switch is an indispensable part of the process. 

We wish to understand the nature of the control in more depth. What does it mean to have a controlled superposition of orders? For example, in the case of the quantum switch the control bit determines the causal order of Alice and Bob and by measuring it an agent in the global future can learn who acted first. If we generalize the quantum switch to the $N$-switch with $N$ agents, we find that the control must record the full causal order for the process to be a valid QC-QC \cite{Wechs_2021}. 

It is also possible that the control is an implicit part of the target system. One can always rewrite a tensor product of the target and control Hilbert spaces as a single Hilbert space that is isomorphic to the tensor product. In a certain sense, when the process has dynamical control, this is also the case. We can see this in the dynamical switch. Who acted last is explicitly encoded in $\alpha_F^{(2)}$ (cf. \cref{eq:switchkraus}). However, who acted first and who acted second is not explicitly encoded anywhere. Nevertheless, an agent in the global future can still obtain this information. Note that just before the last agent acts the system is either in the state $a \ket{0}^{A^I_{k_3}} \ket{0}^{\alpha_3}  + b \ket{1}^{A^I_{k_3}} \ket{1}^{\alpha_3}$ or in the state $c \ket{0}^{A^I_{k_3}} \ket{1}^{\alpha_3}  + d \ket{1}^{A^I_{k_3}} \ket{0}^{\alpha_3}$ depending on the order of the other two agents and for some $a, b, c, d \in \mathbb{C}$. These two states are orthogonal and if the agents apply unitaries they remain orthogonal when they arrive in the global future. This means that there exists some measurement that allows the agent in the global future to perfectly distinguish between the two causal orders (although, the agent needs to know the unitaries that the other agents applied to know the measurement. Additionally, if the agents actually apply a measurement, the agent in the global future would also need access to the purifying degrees of freedom that make these measurements unitary). 

It is tempting to conclude that for any QC-QC the agent in the global future can somehow tell the exact causal order, either by reading it off of some explicit control system in the global future or by some measurement on the target system. However, this is not the case, as demonstrated by the following example. 

\begin{example}[Double quantum switch]
Consider the process depicted in \cref{fig:doubleQS}, which we will call the double quantum switch because it is essentially a quantum switch involving agents 1 and 2 followed by another quantum switch involving agents 3 and 4. 

\begin{figure}
\centering
\includegraphics[width=\textwidth]{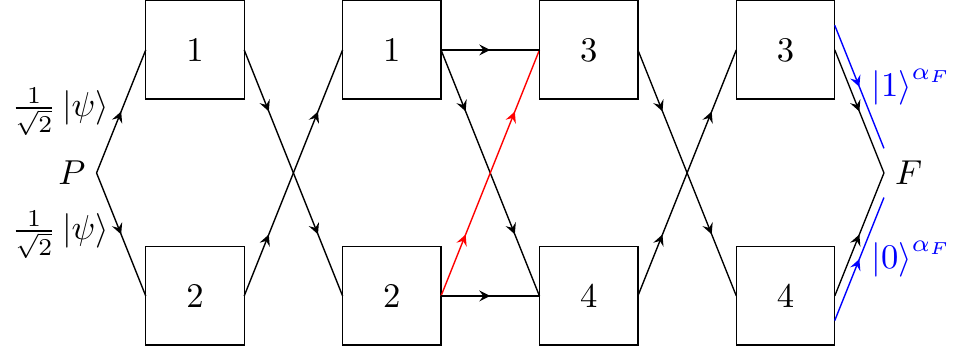}
\caption{The double quantum switch. Essentially two quantum switches in series, one occurring between agents 1 and 2 and the other occurring between agents 3 and 4. However, only the second switch requires an explicit control system (indicated by the blue lines). As a result, an agent in the global future cannot tell the full causal order simply by some measurement on the state arriving in the global future. The black lines indicate identities (with some appropriate normalization) while the red line which indicates negative identity, i.e. a phase flip, which is necessary to have a well-defined QC-QC according to \cref{def:qcqcchar}.}
\label{fig:doubleQS}
\end{figure}

A system is prepared in the global past and then sent to agents 1 and 2 in superposition. Agent 1 then sends their system to agent 2 and vice versa. Next, agent 1 sends their system to agents 3 and 4 in superposition. Agent 2 does essentially the same, however, the system they send to agent 3 has its phase flipped, i.e. the operator is $V^{\rightarrow 3}_{\{1\}, 2} = -\frac{1}{\sqrt{2}} \mathbb{1}$. There is then another swap operation between agents 3 and 4 before the system is sent to the global future along with a control bit that encodes whether agent 3 or 4 applied the last operation. One can now check that this process fulfills the QC-QC conditions. The phase flip turns out to be crucial because it causes unwanted cross terms to cancel. All in all, the internal operations are given by

\begin{gather}
\begin{aligned}
    V^{\rightarrow k_1}_{\emptyset, \emptyset} &= \frac{1}{\sqrt{2}} \ket{\psi}^{A^I_{k_1}} \\
    V^{\rightarrow k_2}_{\emptyset, k_1} &= \mathbb{1}^{A^O_{k_1} \rightarrow A^I_{k_1}} \\
    V^{\rightarrow k_3}_{\{k_1\}, k_2} &= \begin{cases} -\frac{1}{\sqrt{2}} \mathbb{1}^{A^O_{k_2} \rightarrow A^I_{k_3}}, &k_2=2 \text{ and } k_3=3 \\ \frac{1}{\sqrt{2}} \mathbb{1}^{A^O_{k_2} \rightarrow A^I_{k_3}}, &\text{else} \end{cases} \\
    V^{\rightarrow k_4}_{\{k_1, k_2\}, k_3} &= \mathbb{1}^{A^O_{k_3} \rightarrow A^I_{k_4}} \\
    V^{\rightarrow F}_{\{k_1, k_2, k_3\}, k_4} &= \mathbb{1}^{A^O_{k_4} \rightarrow F} \otimes \ket{k_4 \text{ mod } 2}^{\alpha_F}
\end{aligned}
\end{gather}
where we assume that $k_1, k_2 \in \{1,2\}$ and $k_3, k_4 \in \{3, 4\}$ while all other internal operations are assumed to vanish.

While the control bit $\alpha_F$ allows the agent in the global future to know whether agent 3 acted before 4 or vice versa, they have in general no way of knowing whether agent 1 or agent 2 acted first. This can be easily seen when assuming that all agents act trivially on the system. There are four possible causal orders $(1234), (2134), (1243), (2143)$. Given an input state $\ket{\psi}$, the final states for these causal orders are respectively $-\frac{1}{2} \ket{\psi} \ket{0}, \frac{1}{2} \ket{\psi} \ket{0}, \frac{1}{2} \ket{\psi} \ket{1}, \frac{1}{2} \ket{\psi} \ket{1}$. We see that there is destructive interference between the two causal orders where agent 3 acts before 4. The states resulting from the other two causal orders are exactly the same and thus there is no way to distinguish them.
\end{example}

Given the above example, if one wishes to determine the causal order, the agents need to cooperate. If all of them (or in the particular case of the double quantum switch, even just agents 1 and 2) measure the control system of the QC-QC $\ket{\C{K}_{n-1}, k_n}$, they can reconstruct the causal order after the process. Alternatively, if the control system is copied at each step to some external system (using an $n$-dimensional COPY gate), an agent in the global future could also measure these systems to reconstruct the causal order without the cooperation. This second option would also not collapse the superposition between the different orders until the process is finished, unlike measuring the control while the process is still ongoing.

We see that it is difficult to define what it means for a process to be controlled in the absence of an explicit control system. We would expect that any QC-QC is controlled, but the only hint to this is the control system $\ket{\C{K}_{n-1}, k_n}$ which does not show up explicitly in the process vector as it is contracted over. A particular implementation of a QC-QC might therefore not make use of this control system.

However, we can give a sufficient condition. We noticed that control is often but not always connected to orthogonality. We can thus say that if two process or Choi vectors are orthogonal, a superposition of them is controlled. The idea is that if the vectors are orthogonal, we could add control systems to the process similar to how control $\ket{\C{K}_{n-1}, k_n}$ is added when embedding the internal and local operations of the agents in the QC-QC framework.

\subsection{Spacetime position labels as control in process boxes}\label{sec:timestamps}

Let us consider $N$ agents. For fixed orders, we can represent the order relations between the agents as directed acyclic graphs (DAG) over the agents, that is an agent $A_i$ can signal an agent $A_j$ if the node representing $A_j$ in the DAG can be reached from the node representing $A_i$ by following directed edges. Process boxes with fixed orders can then in turn be represented by such DAGs. Note that a specific DAG $\C{G}$ can be implemented by different process boxes as we can relabel the positions as long as the order relations are still respected due to \cref{lemma:relabeling}. Consider now two DAGs $\C{G}$ and $\C{G}'$ over $N$ agents such that there exists at least one pair of agents $A_i$ and $A_j$ for whom the two graphs specify a different order relation. Then the effective Choi vector of any process box implementing $\C{G}$ and that of any process box implementing $\C{G}'$ are orthogonal. This is a consequence of the relativistic causality condition of causal boxes. In each of the two process boxes, the agents $A_i$ and $A_j$ are associated with a position label when they receive a non-vacuum state. Let us call this position label $t^i_{\C{G}}$ for $A^i$ in the implementation of $\C{G}$ and analogously for $t^j_{\C{G}}, t^i_{\C{G}'}$ and $t^j_{\C{G}'}$. If the Choi vectors of the two process boxes are not orthogonal, then $t^i_{\C{G}} = t^i_{\C{G}'}$ and $t^j_{\C{G}} = t^j_{\C{G}'}$. Then, due to the different order relationships of $A_i$ and $A_j$ in the two DAGs, we have w.l.o.g. $t^i_{\C{G}} < t^j_{\C{G}}$ and $t^j_{\C{G}'} < t^i_{\C{G}'}$. But if we assume that $t^j_{\C{G}} = t^j_{\C{G}'}$, we find that $t^i_{\C{G}} < t^j_{\C{G}} = t^j_{\C{G}'} < t^i_{\C{G}'}$ and therefore $t^i_{\C{G}} \neq t^i_{\C{G}'}$. Thus, the Choi vectors must be orthogonal. 

In general, we could have superpositions of different position labels for which $A_i$ and $A_j$ receive non-vacuum state. In this case, we can apply the above argument in each term of the superposition to find that the overall Choi vectors are still orthogonal.

We thus find that the superposition of two process boxes that implement different DAGs is controlled. More generally, we can say that a superposition of any number of process boxes implementing different DAGs is controlled.

While this does not show that process boxes in general are controlled, it gives us the physical intuition that control in the process box framework comes from the position labels. In the next two subsections, we will formalize this intuition and show that when assuming \cref{conjecture:simplifying}, one can construct an operationally equivalent QC-QC for every process box.

\subsection{An example of a process with dynamical control of parallel operations}\label{sec:dynamicalparallel}

A remaining problem with finding a QC-QC formulation for every process box, even after all the steps we have taken so far to bring the framework closer to what we intuitively interpret QC-QCs to be, is that the process box allows for multiple agents to receive a message at the same time. For example, Alice could send a message to both Bob and Charlie. This can in principle be described in the QC-QC framework with the help of QC-PARs, quantum circuits with operations used in parallel \cite{Wechs_2021}. These are circuits where all agents act in parallel instead of sequentially. Parallel operation is, however, equivalent to sequential operation if the output of the agents is directly sent to the global future via internal wires such that different agents cannot influence each other. QC-PARs are thus a subclass of QC-FOs. 

In general, the possible classical orders can be viewed as DAGs as discussed in the previous section. Each of these graphs can be written as a QC-PAR (more precisely, we need to extend the idea of QC-PARs a bit so that we can have both parallel and sequential operations in a single circuit). Thus, any controlled superposition of such graphs can be written as a controlled superposition of QC-PARs.

\begin{figure}
\centering
\includegraphics[width=\textwidth]{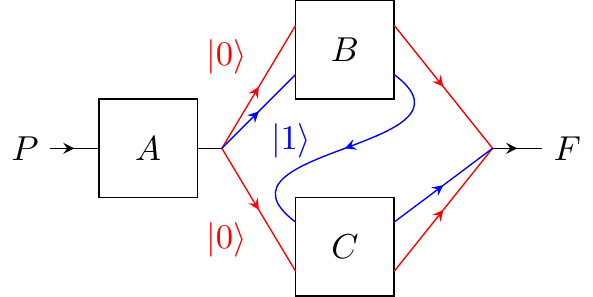}
\caption{Alice receives a state from the global past. Her output is then sent to both Bob and Charlie or just to Bob depending dynamically on whether the state of her output is $\ket{0}$ (red path) or $\ket{1}$ (blue path). modeling such a process with dynamical control of the number of messages as a process box is straightforward, but modeling it as a QC-QC illustrates some of the difficulties that we have to overcome to find a mapping from process boxes to QC-QCs.}
\label{fig:dynamicalparallel}
\end{figure}

However, we could also have dynamical control of order or more specifically of the number of message recipients. Consider the situation depicted in \cref{fig:dynamicalparallel} with Alice, Bob and Charlie. At the initial time $t=2$, Alice receives a qubit. After she acts on it at $t=3$, the process box sends the $\ket{0}$ component to both Bob and Charlie and the $\ket{1}$ component to just Bob at $t=4$. In the former case, the process box then sends both Bob's and Charlie's state at $t=5$ to the global future where it arrives at $t=8$ and in the latter case it sends Bob's state to Charlie and afterwards outputs Charlie's state at $t=7$ to the global future, also at $t=8$.

Using the sequence representation, we can write this process as

\begin{gather}
\begin{aligned}
    \bar{V}_1 &= \ket{\psi}^{A^I} \ket{\Omega}^{B^I} \ket{\Omega}^{C^I} \\
    \bar{V}_2 &= \ket{\Omega}^{A^I} \ket{0}^{B^I} \ket{0}^{C^I} \ket{0}^{\alpha} \bra{0}^{A^O} \bra{\Omega}^{B^O} \bra{\Omega}^{C^O} + \ket{\Omega}^{A^I} \ket{1}^{B^I} \ket{\Omega}^{C^I} \ket{1}^{\alpha} \bra{1}^{A^O} \bra{\Omega}^{B^O} \bra{\Omega}^{C^O} \\
    \bar{V}_3 &= \ket{\Omega}^{A^O} ( \mathbb{1}^{B^O C^O \rightarrow \alpha_3} \ket{0}^\alpha \bra{0}^{\alpha} \bra{\Omega}^{A^O} + \mathbb{1}^{B^O \rightarrow C_I} \ket{\Omega}^{\alpha_3} \ket{1}^{\alpha} \bra{1}^{\alpha} \bra{\Omega}^{A^O} \bra{\Omega}^{C^O}) \\
    \bar{V}_4 &= \mathbb{1}^{\alpha_3 \rightarrow F^{(0)}} \bra{0}^{\alpha} \bra{\Omega}^{A^O} \bra{\Omega}^{B^O} \bra{\Omega}^{C^O} + \mathbb{1}^{C^O \rightarrow F^{(1)}} \bra{1}^{\alpha} \bra{\Omega}^{A^O} \bra{\Omega}^{B^O} 
\end{aligned}
\end{gather}
where superscripts $A^I, B^I, C^I$ designate states in the input spaces of Alice, Bob, Charlie while $A^O, B^O, C^O$ designate states in the output spaces and we introduced two 2-dimensional ancillary Hilbert spaces $\C{H}^{\alpha}$ and $\C{H}^{\alpha_3}$. The Hilbert space of the global future is $\C{H}^F = \C{H}^{F^{(0)}} \oplus \C{H}^{F^{(1)}}$ with $\C{H}^{F^{(0)}} = \mathbb{C}^2 \otimes \mathbb{C}^2$ and $\C{H}^{F^{(1)}} = \mathbb{C}^2$. The ancillaries and the partition of the global future are necessary to obtain isometries. We dropped the time stamps as every isometry $\bar{V}_n$ acts on states at time $t=2n-1$ and outputs states at time $t=2n$. Note that the isometries as defined here are only isometries on the types of states that can actually be produced by local operations. One could extend the above definitions to be isometries on all states, however, this is not necessary when our goal is to find a QC-QC description as all possible extensions would yield the same QC-QC.

We now convert the above process box into a QC-QC by giving the internal operators

\begin{gather}\label{eq:dynpar_kraus}
\begin{aligned}
    V^{\rightarrow A}_{\emptyset, \emptyset} &= \ket{\psi}^{A^I} \\
    V^{\rightarrow B}_{\emptyset, A} &= \ket{0}^{\alpha_2} \ket{0}^{B^I} \bra{0}^{A^O} + \ket{2}^{\alpha_2} \ket{1}^{B^I} \bra{1}^{A^O} \\
    V^{\rightarrow C}_{\{B\}, A} &= \mathbb{1}^{B^O \rightarrow \alpha_3} \ket{0}^{C^I} \bra{0}^{\alpha_2} + \mathbb{1}^{B^O \rightarrow C_I} \ket{2}^{\alpha_3} \bra{2}^{\alpha_2} \\
    V^{\rightarrow F}_{\{A, B\}, C} &= \mathbb{1}^{C^O \alpha_3' \rightarrow F^{(0)}} + \mathbb{1}^{C^O \rightarrow F^{(1)}} \bra{2}^{\alpha_3}
\end{aligned}
\end{gather}
where $\C{H}^{\alpha_3} = \C{H}^{\alpha_3'} \oplus \ket{2}$. All other operators vanish. It can be checked with the characterization condition given in \cref{eq:qcqccond} that this defines a QC-QC.

We can now compare the Choi vector of the process box and the process vector of the QC-QC. We first calculate the Choi vector

\begin{gather}\label{eq:dynpar}
\begin{aligned}
    \dket{\bar{V}_4 \bar{V}_3 \bar{V}_2 \bar{V}_1} =& \dket{\bar{V}_4} * \dket{\bar{V}_3} * \dket{\bar{V}_2} * \dket{\bar{V}_1} \\
    =& (\ket{\psi}^{A^I_2} \ket{\Omega}^{B^I_2} \ket{\Omega}^{C^I_2}) \\
    &(\ket{0}^{A^O_3} \ket{\Omega}^{B^O_3} \ket{\Omega}^{C^O_3} \ket{\Omega}^{A^I_4} \ket{0}^{B^I_4} \ket{0}^{C^I_4} \ket{\Omega}^{A^O_5} \\
    &\dket{\mathbb{1}}^{B^O_5 C^O_5 F^{(0)}_8} \ket{\Omega}^{A^I_6} \ket{\Omega}^{B^I_6} \ket{\Omega}^{C^I_6} \ket{\Omega}^{A^O_7} \ket{\Omega}^{B^O_7} \ket{\Omega}^{C^O_7} \\
    &+\ket{1}^{A^O_3} \ket{\Omega}^{B^O_3} \ket{\Omega}^{C^O_3} \ket{\Omega}^{A^I_4} \ket{1}^{B^I_4} \ket{\Omega}^{C^I_4} \ket{\Omega}^{A^O_5} \ket{\Omega}^{C^O_5}\\
    &\dket{\mathbb{1}}^{B^O_5 C^I_6} \ket{\Omega}^{A^I_6} \ket{\Omega}^{B^I_6} \ket{\Omega}^{A^O_7} \ket{\Omega}^{B^O_7} \dket{\mathbb{1}}^{C^O_7 F^{(1)}_8}).
\end{aligned}
\end{gather}

Here, we moved the time stamps so that the equation is more compact, i.e. $\ket{\psi}^{A^I_t} = \ket{\psi, t}^{A^I}$.

The process vector of the QC-QC is then

\begin{gather}
\begin{aligned}
    \ket{w_{\{A, B, C\}, F}} &= \ket{w_{(A, B, C, F)}} \\
    &= \dket{V^{\rightarrow A}_{\emptyset, \emptyset}} * \dket{V^{\rightarrow B}_{\emptyset, A}} * \dket{V^{\rightarrow C}_{\{B\}, A}} * \dket{V^{\rightarrow F}_{\{A, B\}, C}} \\
    &= \ket{\psi}^{A^I} (\ket{0}^{A^O} \ket{0}^{B^I} \ket{0}^{C^I} \dket{\mathbb{1}}^{B^O C^O F^{(0)}} + \ket{1}^{A^O} \ket{1}^{B^I} \dket{\mathbb{1}}^{B^O C_I} \dket{\mathbb{1}}^{C^O F^{(1)}}).
\end{aligned}
\end{gather}

We see that by dropping vacuum states and time labels the effective Choi vector of the process box is the same as the process vector of the QC-QC. This is also what happens when we compose the Choi vector of the process box with local operations as we considered them in \cref{sec:statespace} in \cref{eq:local_pb} and \cref{def:local_equivalence}. If we only contract over the vacuum states and the time stamps, we obtain

\begin{gather}
\begin{aligned}
    (\dket{\bar{A}} \otimes \dket{\bar{B}} \otimes \dket{\bar{C}})*\dket{\bar{V}} =& \bigotimes_{i=1}^3 ((\dket{A} + \ket{\Omega} \ket{\Omega}) \otimes (\dket{B} + \ket{\Omega} \ket{\Omega}) \otimes \\
    &(\dket{C} + \ket{\Omega} \ket{\Omega}) \otimes \ket{t=2i+1} \ket{t=2i}) * \dket{\bar{V}} \\
    =& (A \otimes B \otimes C) * \dket{w_{\{A, B, C\}, F}}.
\end{aligned}
\end{gather}

The QC-QC and the process box are thus operationally equivalent. We see that a specific process box where whether operations are parallel or sequential is dynamically controlled can be written as a QC-QC. On the one hand, this may be expected as QC-QCs are capable of describing parallel operation and dynamical control when they appear separately. On the other hand, it remains unclear how to achieve this in the general case. When we were trying to find a process box description of every QC-QC, we were able to define $\bar{V}_n$ from just $V_n$. This was possible because we were able to interpret $V_n$ as taking inputs at $t=2n-1$ and producing outputs at $t=2n$. This is not possible for this example. To obtain the QC-QC operation $V_3$ here, which is essentially just $V^{\rightarrow C}_{\{B\}, A}$, we had to look at both $\bar{V}_2$ and $\bar{V}_3$. On the other hand, we also had to ``stretch out" the action of $\bar{V}_2$ over $V_2$ and $V_3$. This was a consequence of the dynamical control of the number of messages sent.

These problems can be much more severe for general process boxes. We can have dynamical control of how many messages are sent to the agents during each time step. If this is the case, the number of messages sent during a later time step may depend on the number of messages sent during an earlier time step. This then makes it more difficult to define the QC-QC operation $V_n$. Informally speaking, the more of this type of dynamical control the process box has, the more isometries of the sequence representation we need to look at to define $V_n$ and also the more internal operations a single isometry gets ``stretched out" over. For a general process box, the internal operation $V_n$ could depend on any number of isometries $\bar{V}_m$. The difficulty is then finding a general scheme to construct $V_n$ from the set of $\bar{V}_m$.

Additionally, there could be time steps where it is possible that no agent receives anything. This is a problem because QC-QCs assume that during each step an agent receives something. However, we cannot just ignore such time steps because whether or not a message is sent could again be dynamically controlled.

In the end, we also need to check that the QC-QC we constructed is actually equivalent to the initial process box which becomes more difficult, the less related $\bar{V}_n$ and $V_n$ are.



\subsection{Mapping process boxes to QC-QCs}\label{sec:pbtoqcqcsub}

Our goal for this subsection will be to prove that general process boxes, with the additional constraint of FAA, can be mapped to QC-QCs under the assumption that \cref{conjecture:simplifying} holds.

\begin{restatable}[QC-QC description of process boxes]{prop}{pbtoqcqc}
\label{prop:pbtoqcqc}
For every process box of the form of \cref{conjecture:simplifying} there exists a QC-QC such that they are operationally equivalent.
\end{restatable}

However, after outlining the complications that appear when trying to do this, let us consider a slightly different approach to what we did for the example in the last section. 

In principle, the isometry $\bar{V}_n$ could send zero, one or multiple messages to the agents with all these possibilities potentially being dynamically controlled as we discussed already. However, the internal operation $V_n$ sends exactly one message to the agents. In the example of the previous section, we essentially ``stretched out" the operation of $\bar{V}_2$ over $V_2$ and $V_3$ while also ``compressing" parts of the action of $\bar{V}_2$ and $\bar{V}_3$ into $V_3$ such that each of the QC-QC internal operations only sends a single message. As we already argued, this is difficult to do in the general case where we do not have an explicit expression for $\bar{V}_n$. However, notice that we actually have two problems here. The first problem is that $\bar{V}_n$ could send multiple messages. The second is that it could send no messages.

We can thus try to solve the first problem by ``stretching out" $\bar{V}_n$ in the process box framework. This is easier to do as the process box has no problem with sending nothing and we thus do not have to worry about the second problem right away.

\begin{restatable}[``Stretching" the internal operations]{lemma}{stretching}
\label{lemma:stretching}
Any process box of the form of \cref{conjecture:simplifying} is operationally equivalent to a process box which is again of the form of \cref{conjecture:simplifying} with the additional property that it sends and receives at most one message during each time step. 
\end{restatable}

The basic idea is that we can add additional time steps over which the messages that $\bar{V}_n$ sends are spread out. These time steps do not matter as they simply get contracted over when taking the composition with local agents, which as we argued previously apply time-independent local operations.

As we only have to consider 0- and 1-message states from now on, in the rest of the section we will use the wire isomorphism \cref{eq:wireiso} to write $\ket{\psi, t}^{\bar{A}^{I/O}_n}$ instead of $\ket{\Omega, t}^{\bar{A}^{I/O}_1} ... \ket{\psi, t}^{\bar{A}^{I/O}_n} ... \ket{\Omega, t}^{\bar{A}^{I/O}_N}$ to simplify our notation. Additionally, if $\psi = \Omega$, we will simply write $\ket{\Omega, t}^{I/O}$.

Let us now outline the proof idea of \cref{prop:pbtoqcqc}, the details of which can be found once again in \cref{sec:proofs}. We wish to add internal wires that carry the control system from the QC-QC to the process box. This will help us to define the internal operator $V_{n+1}$ of the QC-QC as it is supposed to act on states where the control length is $n$.\footnote{When we say controls with length $n$, we mean states in $\C{H}^{C_n}$ because they have the form $\ket{\C{K}_{n-1}, k_n}$ and $|\C{K}_{n-1} \cup k_n| = n$.} We take the position that the control is already implicitly contained in the process box as discussed in \cref{sec:role,sec:timestamps}. Making the control explicit can then be viewed as defining some function $C$ that appends the correct control to each state, $C \ket{\psi} = \ket{\psi} \ket{i}^C$. $C$ obviously should not make states orthogonal if they were not orthogonal before (however, just because two states are orthogonal does not mean that their explicit controls must be orthogonal). In other words, $C$ must be an isometry. 

\begin{example}[Adding the control]
We consider two examples where the control is rather explicit already to see how this works: The function $C$ could duplicate basis states $C \ket{i} = \ket{i} \ket{i}^C$ or it could use the lab as the control $C \ket{\psi}^A = \ket{\psi}^A \ket{A}^C$. In both cases, we do not add any additional information.
\end{example}




\begin{figure}
    \centering
    \includegraphics[width=\textwidth]{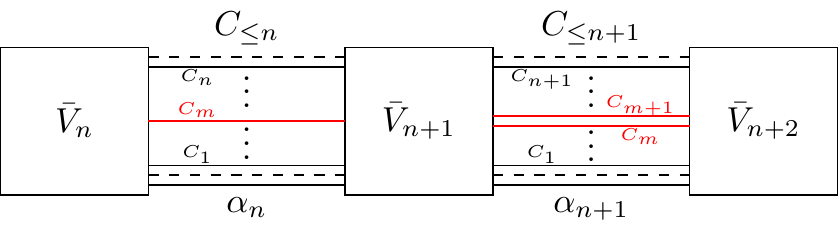}
    \caption{The internal wires of the sequence representation after adding internal wires that correspond to the control in the QC-QC picture. The isometry $\bar{V}_{n+1}$ has internal wires $\alpha_n$ and $\alpha_{n+1}$ that were already there before adding the control. The wires $C_i, i=1,...,n+1$ are new internal wires and carry controls of length $i$. On both the incoming and outgoing side, they are consolidated into a single wire $C_{\leq n}$ and $C_{\leq n+1}$, respectively, as indicated by the dashed lines. The red wires indicate that a state whose control is on $C_m$ is mapped to a state whose control is either on $C_m$ or $C_{m+1}$ by $\bar{V}_{n+1}$.}
    \label{fig:addcontrol}
\end{figure}

The idea is therefore to start out with an arbitrary sequence representation of the process box and then adding additional internal wires between the isometries that carry the control. This then still yields the same process box. We formalize this idea in \cref{lemma:addcontrol} in \cref{sec:proofs}. The resulting isometries are depicted in \cref{fig:addcontrol}. Importantly, due to the fact that previous isometries may have sent no messages to the agents, a control arriving at the isometry $\bar{V}_{n+1}$ need not have length $n$. 

What is therefore the effective input space of $\bar{V}_{n+1}$? This isometry can receive either 0 or 1 messages from local agents, it can receive a state via an internal wire $\alpha_n$ and a control of length $m \leq n$ via an internal wire which we will denote with $C_m$ as we did for QC-QCs.\footnote{In the interest of simplifying notation we will use $C_m$ for all isometries $\bar{V}_{n+1}$ with $m \leq n$. It should be understood that the wires are different wires, but the Hilbert spaces can be taken to be the same.} On the other hand, the output can consist of 0 or 1 messages, a state on an internal wire $\alpha_{n+1}$ and a control $C_m$ or $C_{m+1}$ depending on the number of messages sent. We can thus write


\begin{gather}\label{eq:isomdirectsum}
\begin{aligned}
    \bar{V}_{n+1}: &\bigoplus_{m=0}^{n} \C{H}^{\bar{A}^O}_{t=2n+1} \otimes \C{H}^{\alpha_n} \otimes \C{H}^{C_m} \rightarrow \bigoplus_{m=0}^{n} \C{H}^{\bar{A}^I}_{t=2n+2} \otimes \C{H}^{\alpha_{n+1}} \otimes \C{H}^{C_{m+1}}
\end{aligned}
\end{gather}
where

\begin{equation}
\C{H}^{\bar{A}^{I/O}}_t \coloneqq (\bigoplus_{k=1}^N  \C{H}^{A^{I/O}_k} \oplus \ket{\Omega}^{I/O}) \otimes \ket{t}
\end{equation}
is the Hilbert space with zero or one messages on the input/output wires at time $t$.

In order to define the isometries from the QC-QC framework, we need to look at different $\bar{V}_{n+1}$. The control system that we added will be a great help to us in this as the action of $V_{n+1}$ will essentially be the action of the $\bar{V}_{m+1}$ on states where the control has length $n$. The main complication is the fact that the isometries in the process box picture can send no messages while the QC-QC isometries cannot do that. This requires looking at different $\bar{V}_{m+1}$ again, but the steps we have taken up until now allow us to do this in a systematic fashion.

Following these ideas, we can show \cref{prop:pbtoqcqc} and thus find a mapping from process boxes to QC-QCs.

\newpage

\part{Conclusion}\label{sec:conclusion}

In summary, we constructed a mapping from QC-QCs to process boxes such that a QC-QC and the process box it is mapped to are operationally equivalent. This required a careful analysis of the underlying state spaces of both frameworks as they model temporal superpositions in different ways. The process box framework uses time stamps and vacuum states such that agents act during all time steps but only receive a non-vacuum state during one. Meanwhile, the QC-QC framework, like the general process matrix framework, chooses for each agent a single special time step during which the agent acts. Otherwise, the agents are inactive. This is similar to other frameworks that use time-delocalization to model indefinite causal order \cite{Oreshkov_2019, https://doi.org/10.48550/arxiv.2201.11832}. In the course of this analysis, we defined an identification of the local operations that can be applied in each framework and argued that a QC-QC and a process box are operationally equivalent if the measurement statistics are the same under this identification, or, more generally, the behavior of the QC-QC and the process box as supermaps is isomorphic. Our mapping also resolves the composability issue of process matrices \cite{Gu_rin_2019} in the case of QC-QCs as it allows us to understand composition of QC-QCs as the composition of causal boxes, which is well-defined, even if the composition does not necessarily yield a process box. This is a consequence of process boxes being, by construction, a subset of causal boxes \cite{Vilasini_2020}, which in turn are closed under composition \cite{Portmann_2017}. Composability is therefore recovered by constructing an extension for each QC-QC, which requires embedding the state space of the QC-QC into a Fock space. We have also seen that different extensions are possible when mapping a QC-QC to a process box and that the chosen extension has physical relevance for composition. The extension must thus be defined with the experimental setup in mind.

We can also ask whether composability is a necessary condition for physicality of processes, now that we know that this condition would not exclude QC-QCs, which we expect to be physical \cite{Wechs_2021}. Future research could therefore go in the following directions: firstly, one could develop new frameworks that generalize process matrices, loosening the assumptions that prevent composition by incorporating features from the causal box framework. For example, the process matrix framework can be generalized such that it can model multi-round processes by including a local order relation for each agent \cite{Hoffreumon_2021}, although this on its own without a global order relation (i.e. background spacetime) seems insufficient to resolve the composability issue. Considering more general frameworks, like the post-selected CTC framework \cite{Lloyd_2011} or the multi-time states framework \cite{Aharonov1964}, that include process matrices as a linear subset \cite{Ara_jo_2017_CTC, Silva_2017} would be another possible approach that does not rely on a background spacetime to recover composability. Investigating composition via such more general frameworks is somewhat similar to what we did in this work. Composition of QC-QCs, which we take to be the composition of process boxes, is only well-defined if we consider it as composition in the more general causal box framework. Secondly, one could ask whether QC-QCs are the largest class of processes that admit a physical interpretation and also the largest class that admit a causal box description (which we expect to describe anything we can do over a fixed, acyclic spacetime). As we have shown that QC-QCs and process boxes are operationally equivalent, we can pose this as the question of whether we can drop any of the constraints in the process box framework. FAA is one obvious candidate, but OSR is also a somewhat unnatural assumption, that is largely imposed to simplify the framework. Thirdly, one could attempt to find no-go results that show that no extensions to causal boxes can exist for some or even all non-QC-QC process matrices.

An interesting feature that came out while developing the extensions of the QC-QCs to process boxes was that symmetrization and the use of Fock spaces was necessary for the extension defined in \cref{sec:controlintarget}. The initial motivation to use Fock spaces for causal boxes was to indicate the fact that all the ordering information is contained in the position labels \cite{Portmann_2017}. However, the extension of the dynamical switch in \cref{sec:switchtopb} and its generalization to arbitrary QC-QCs in \cref{sec:controlintarget} were physically motivated by the idea of viewing a QC-QC as a circuit in which photons propagate independently from each other. Symmetrization then turns out to be necessary to obtain isometries for the internal operations. This may be viewed as a consequence of viewing the target systems as photons, which are bosons. The use of Fock spaces may thus allow us to use causal boxes to describe processes in a quantum field theory (QFT) setting. We could also ask if we could realize causal boxes with antisymmetric Fock spaces in order to describe fermions. Using process matrices in QFT has previously been discussed in \cite{Faleiro2020}, but one issue they raise is that some QFT processes would require a variable number of operations (e.g. radiative beta decay has a variable number of outgoing photons). This might be less of a problem for the causal box framework as it does not assume that each agent acts exactly once.

We also showed that process boxes admit a unitary extension by adding a single additional wire each in the global past and the global future of all positions. This is a somewhat expected result as the authors of \cite{Ara_jo_2017} argue that physical processes admit unitary extensions due to the connection of unitarity and reversibility of quantum mechanics \cite{Chiribella_2010}. Our result therefore gives support to their argument. 

In our attempt to find a mapping from process boxes to QC-QCs, we argued that process boxes can only realize controlled superpositions of order due to the orthogonality of the position labels. This can be viewed as a consequence of the principle of relativistic causality. This result thus suggests that QC-QCs are singled out by relativistic causality in a fixed, acyclic spacetime (up to the assumptions we imposed on process boxes and thus also QC-QCs). We also identified the potentially very general structure of the set of positions $\C{T}$ as a major hurdle for constructing a mapping from process boxes to QC-QCs. We showed that only the order relations, and not the position labels themselves, matter and conjectured that each process box can be reduced to a simpler, but operationally equivalent process box, where only the time information is relevant with a few more restriction. Based on this conjecture, we constructed a mapping from process boxes to QC-QCs. Should our conjecture not hold, our proof may still serve as a template for a future construction that is not based on the conjecture.

We note that our mapping from process boxes to QC-QCs only holds if we impose the constraint of fully active agents. Instead of imposing this constraint on process boxes, one could instead extend the QC-QC framework such that it allows for passive agents. This would likely include introducing an explicit vacuum state into the framework. 

Finally, let us briefly discuss what our results suggest for the broader open question of processes in indefinite spacetime. A key difference between fixed and indefinite spacetime processes is that in the latter whether states are localized at a single spacetime point or delocalized over several spacetime points is no longer an agent-independent notion \cite{Zych_2019,Paunkovi__2020}. Each agent's perspective would then correspond to a different partially ordered set. For example, the gravitational quantum switch \cite{Zych_2019} appears to Alice like a process box where she acts at a fixed time while Bob acts either before or after her, while Bob would describe the process as a process box where he acts at a fixed time and Alice is delocalized. It could then be that once the perspective of a single agent is fixed, a process in an indefinite spacetime can always be modeled by a process box. The spacetime could also depend dynamically on the action of the agents, for example because their messages are or interact with massive systems. The question is then if that can be understood as just a different flavor of dynamical control of causal order, which both QC-QCs and process boxes are capable of. We also saw that a single QC-QC can have multiple operationally equivalent process boxes, depending on the exact implementation. This hints at the possibility of more general process box and causal box frameworks that can model indefinite spacetimes directly. All in all, this suggests that our results and techniques should still be relevant in a broad class of indefinite spacetime structures, but whether the impossibility of violating causal inequalities generalizes to physical implementations in indefinite spacetimes remains to be seen.


\newpage

\printbibliography

\appendix

\part{Appendix}

\section{Isometries}\label{sec:isometries}

In the results section, isometries showed up in several propositions. In order to make it easier to follow the proofs of these propositions in \cref{sec:proofs}, we quickly review isometries here.

\begin{defi}[Isometries]
Let $W^I$ and $W^O$ be inner product spaces. Let $V: W^I \rightarrow W^O$ be a linear map. We call $V$ an isometry if for all $v, w \in W^I$, it holds that $v^\dagger V^\dagger V w = v^\dagger w$, i.e. $V$ conserves the inner product.
\end{defi}

It can be difficult to check that a map $V$ conserves the inner product. Luckily, there are three equivalent conditions which we can use instead.

\begin{lemma}[Characterization of isometries]
\label{lemma:isochar}
Let $V: W^I \rightarrow W^O$ be a linear map. The following four conditions are equivalent.

\begin{enumerate}
    \item $V$ is an isometry.
    \item $V$ conserves the norm, $w^\dagger V^\dagger V w = w^\dagger w, \forall w \in W^I$.
    \item If $w_1,...,w_N$ is a basis of $W^I$, then $w_i^\dagger V^\dagger V w_j = w_i^\dagger w_j, \forall i,j = 1,...,N$.
    \item If $W^I = \bigoplus_{i=1}^M W^I_i$ where $W^I_1,...,W^I_M$ are orthogonal subspaces that are mapped to orthogonal subspaces by $V$, then the restriction of $V$ to each $W^I_i$, $V|_{W^I_i}$ is an isometry.
\end{enumerate}

\end{lemma}

The fourth condition may seem a bit strange at first sight. It essentially states that if a map maps orthogonal subspaces to orthogonal subspaces, it is enough to check that it is an isometry on each of these subspaces separately. Among other things, this condition allows us to consider each space with a fixed number of messages separately in the proofs of \cref{sec:qcqctopb}.

The second condition is frequently used as the definition of an isometry in place of the definition we used. Slightly weaker versions of the third and fourth condition, stating that $V$ is an isometry iff it maps orthonormal bases to orthonormal bases or iff there exists an orthonormal basis that is mapped to an orthonormal basis, are also often included in standard textbooks (cf. \cite{axler2014linear}).

\begin{proof}

``$1 \implies 2"$: The norm (squared) is an inner product, thus if all inner products are conserved, then in particular all norms are conserved.

``$2 \implies 1$": Let $v, w \in W^I$. Then $|v+w|^2 = |V(v+w)|^2 = |Vv|^2 + |Vw|^2 + v^\dagger V^\dagger V w + w^\dagger V^\dagger V v = |v|^2 + |w|^2 + 2 \text{Re}(v^\dagger V^\dagger V w)$ and also $|v+w|^2 = |v|^2 + |w|^2 + 2 \text{Re}(v^\dagger V^\dagger V w)$ where $\text{Re}$ gives the real part of a complex number. Thus, the real part of the inner product $v^\dagger w$ is conserved and repeating the previous calculation with $v-w$ in place of $v+w$ shows that the imaginary part is also conserved. 

``$1 \implies 3$'': If $V$ conserves the inner product between all vectors then in particular it conserves the inner product between basis vectors as these are a subset of all vectors.

``$3 \implies 1$'': Condition 1 is equivalent to condition 2 so we can consider the norm of a vector $w \in W^I$. We can write $w = \sum_{i=1}^N \lambda_i w_i$ for some $\lambda_i \in \mathbb{C}$ as the vectors $w_i$ form a basis. Then $w^\dagger V^\dagger V w = \sum_{ij} \overline{\lambda}_i \lambda_j w^{\dagger}_i V^\dagger V w_j = \sum_{ij} \overline{\lambda}_i \lambda_j w^{\dagger}_i w_j = w^\dagger w$ where we used condition 3 in the second equality.

``$1 \implies 4$'': If $V$ is an isometry on the whole space $W^I$, it is in particular an isometry on any subspace.

``$4 \implies 1$'': We consider again the norm of an arbitrary vector $w \in W^I$. As $W = \bigoplus_{i=1}^M W^I_i$, we can write $w = \sum_{i=1}^M w^I_i$ for vectors $w^I_i \in W^I_i$. Then, $w^\dagger V^\dagger V w = \sum_{ij} w^{I \dagger}_i V^\dagger V w^I_j = \sum_i w^{I \dagger}_i V^\dagger V w^I_i =  \sum_i w^{I \dagger}_i w^I_i = w^\dagger w$. In the second equality, we used that $V$ conserves the orthogonality of the subspaces $W^I_i$ which causes cross terms to vanish, $w^{I \dagger}_i V^\dagger V w^I_j = 0$ for $i \neq j$, and in the third equality, we used that $V$ is an isometry when restricted to each subspace $W^I_i$.
\end{proof}

Of course, we can also mix and match these conditions. For example, if we have split our input space into orthogonal subspaces as in condition 4, we can then use condition 2 to show that our map is an isometry when restricted to each of the subspaces.

\section{Proofs of results}\label{sec:proofs}

\subsection{Proofs for \texorpdfstring{\Cref{sec:qcqctopb}}{Section 5}}

\compislink*

\begin{proof}
Using the definition of composition \cref{def:comp}, we can write

\begin{gather}
\begin{aligned}
    \C{M}^{C\hookrightarrow B} (\ket{\psi}^A \bra{\phi}^A) &= \sum_{k,l} \bra{k}^C\C{M}(\ket{\psi}^A\ket{k}^B\bra{l}^B\bra{\phi}^A)\ket{l}^C \\
    &= \sum_{k,l} \bra{k}^C(\C{M}_A \otimes \C{M}_B)(\ket{\psi}^A\ket{k}^B\bra{l}^B\bra{\phi}^A)\ket{l}^C \\
    &= \sum_{k, l} \bra{k}^C (M_A \otimes M_B)*(\ket{\psi}^A\ket{k}^B\bra{l}^B\bra{\phi}^A)\ket{l}^C \\
    &= \sum_{k, l} \underbrace{\bra{k}^C M_A \ket{l}^C}_{(M_A)_{kl}} * (\ket{\psi}^A \bra{\phi}^A) \otimes \sum_{i, j} \underbrace{\bra{i}^B M_B \ket{j}^B}_{(M_B)_{ij}} \underbrace{\braket{i|k}^B\braket{l|j}^B}_{\delta_{ik} \delta_{lj}} \\
    &= \sum_{k,l} (M_A)_{kl} \otimes (M_B)_{kl} * (\ket{\psi}^A \bra{\phi}^A) \\
    &\stackrel{\text{\cref{def:linkmat}}}{=} (M_A*M_B) * (\ket{\psi}^A \bra{\phi}^A).
\end{aligned}
\end{gather}

In the third line, we used that $\C{M}_{A/B}(\rho) = M_{A/B}*\rho$ and in the fourth line we used commutativity of the link product and the second sum is the explicit form of the link product $M_B*(\ket{k}^B \bra{l}^B)$. Finally, in the last line, we used the definition of the link product, taking $\C{H}^B$ and $\C{H}^C$ to be the same Hilbert space. We can then also write $\C{M}^{C\hookrightarrow B} (\ket{\psi}^A \bra{\phi}^A) = M^{C\hookrightarrow B} * \ket{\psi}^A \bra{\phi}^A$. As these equations must hold for all $\ket{\psi}^A, \ket{\phi}^A \in \C{H}^A$, we can drop these states and obtain

\begin{equation}
    M^{C\hookrightarrow B} = M_A*M_B.
\end{equation}

This shows that sequential composition can be expressed as the link product, at least in the case of finite dimensional Hilbert spaces.

\end{proof}

\switchiso*

\begin{proof}

It suffices to show the statement separately for states with zero, one, two or three photons on the output wires of the agents. This is because the operators $\bar{V}_1, \bar{V}_2, \bar{V}_3, \bar{V}_4$ conserve the number of non-vacuum states and thus the spaces of zero, one, two and three messages remain orthogonal under application of the operators. We can thus use condition 4 of \cref{lemma:isochar}. The proposition follows immediately for the case of zero messages from how we defined the action of the internal operations on the vacuum, see \cref{eq:isovac}. For a single message, the proposition follows from the fact that $V_1,V_2,V_3,V_4$ are isometries in the QC-QC picture, or more formally from our definition of the action on single message states \cref{eq:novel_single} 

\begin{gather}
\begin{aligned}
    \bra{\psi, t=2n+1}^{A^O_{k_n}}& \bra{\C{K}_{n-1}, k_n} \bar{V}^{\dagger}_{n+1} \bar{V}_{n+1} \ket{\phi, t=2n+1}^{A^O_{l_n}} \ket{\C{L}_{n-1}, l_n} \\
    &= \bra{\psi}^{A^O_{k_n}} \bra{\C{K}_{n-1}, k_n} V^{\dagger}_{n+1} V_{n+1} \ket{\phi}^{A^O_{l_n}} \ket{\C{L}_{n-1}, l_n} \\
    &=\bra{\psi}^{A^O_{k_n}} \ket{\phi}^{A^O_{l_n}} \braket{\C{K}_{n-1}, k_n|\C{L}_{n-1}, l_n}
\end{aligned}
\end{gather}
where we did not explicitly write the vacuum state on the other wires.

Therefore, for the remainder of the proof, we consider only states with two or three messages.

Before we continue, note that 

\begin{equation}\label{eq:orthog}
\bra{\psi} V^{\rightarrow k_{n+1} \dagger}_{\C{K}_{n-1}, k_n} V^{\rightarrow k_{n+1}}_{\C{K}'_{n-1}, k'_n} \ket{\phi} = 0
\end{equation}

if $k_n \neq k'_n$ for any $\ket{\psi}, \ket{\phi}$ in the appropriate spaces. This can be directly seen from the explicit form of the operators in \cref{eq:switchkraus}.

Additionally, we can show that states that have differences in one or more controls also remain orthogonal under application of the operators. We consider two states $\bigotimes_{i} \ket{\psi^i, t=2n+1}^{A^O_i} \ket{\C{K}^i_{n-1}, i}$ and $\bigotimes_{i} \ket{\phi^i, t=2n+1}^{A^O_i} \ket{\C{L}^i_{n-1}, i}$ where the vacuum states are implicit according to our convention and such that their controls are not all the same. This means at least one control of one of the two must be different from all controls of the other one. W.l.o.g. we can assume $\ket{\C{K}^1_{n-1}, 1} \neq \ket{\C{L}^i_{n-1}, i}$ for all $i$ (this is because the process is symmetric under cyclical permutation of the agents). Now consider applying $\bar{V}_{n+1}$ to each of the two states and taking the inner product between them. Each term in the result must contain a factor 

\begin{equation}\label{eq:controlortho}
\bra{\psi^1} V^{\rightarrow k_{n+1} \dagger}_{\C{K}^1_{n-1}, 1} V^{\rightarrow k_{n+1}}_{\C{L}^i_{n-1}, i} \ket{\phi^i} \braket{\C{L}^i_{n-1} \cup i, k_{n+1}|\C{K}^1_{n-1} \cup 1, k_{n+1}}
\end{equation} 

However, if $i \neq 1$, then this factor vanishes because of \cref{eq:orthog}. If $i=1$, then $\C{K}^1_{n-1} \neq \C{L}^1_{n-1}$ and thus also $\C{K}^1_{n-1} \cup 1 \neq \C{L}^1_{n-1} \cup 1$ and the inner product between the two controls in \cref{eq:controlortho} vanishes. Thus, we have shown that if the controls of two states are orthogonal to each other, the states remain orthogonal under application of $\bar{V}_1, \bar{V}_2, \bar{V}_3, \bar{V}_4$. This means that the spaces $\text{Span}(\{\bigotimes_{i} \ket{\psi^i, t=2n+1}^{A^O_i} \ket{\C{K}^i_{n-1}, i}: \ket{\psi^i}^{A^O_i} \in \C{H}^{A^O_i}\})$ fulfill the requirements of condition 4 of \cref{lemma:isochar}. Therefore, it is enough to show that $\bar{V}_{n+1}$ is an isometry on each of these subspaces separately. We can do this by using condition 3 of \cref{lemma:isochar}, that is by showing that $\bar{V}_{n+1}$ conserves the inner product between all the elementary tensor products within each such subspace. Additionally, when considering states with two messages it is enough to show that $\bar{V}_{n+1}$ is an isometry on states where the messages are on wire $1$ and wire $2$. The fact that the dynamical switch is invariant under cyclical permutation of the agents then ensures that $\bar{V}_{n+1}$ is also an isometry on all other two-message states.


Applying $\bar{V}_{n+1}$ to a two-message state from the above subspace, we obtain

\begin{gather}\label{eq:Vapplied2}
\begin{aligned}
    \bar{V}_{n+1}&(\bigotimes_{i=1}^{2} \ket{\psi^i, t=2n+1}^{A^O_i} \ket{\C{K}^i_{n-1}, i}) \\
    &= symm(\bigotimes_{i=1}^{2} \sum_{k^i_{n+1} \not \in \C{K}^i_{n-1} \cup i} V^{\rightarrow k^i_{n+1}}_{\C{K}^i_{n-1}, i} \ket{\psi^i} \ket{t=2n+2} \ket{\C{K}^i_{n-1} \cup i, k^i_{n+1}} \\
    &= symm(\sum_{k_{n+1}^1} \sum_{k_{n+1}^2} \bigotimes_{k_{n+1}} (\bigotimes_{i: k^i_{n+1} = k_{n+1}} V^{\rightarrow k_{n+1}}_{\C{K}^i_{n-1}, i} \ket{\psi^i} \ket{t=2n+2} \ket{\C{K}^i_{n-1} \cup i, k_{n+1}}) \\
    &= \sum_{k_{n+1}^1} \sum_{k_{n+1}^2} \bigotimes_{k_{n+1}} (\bigodot_{i: k^i_{n+1} = k_{n+1}} V^{\rightarrow k_{n+1}}_{\C{K}^i_{n-1}, i} \ket{\psi^i} \ket{t=2n+2} \ket{\C{K}^i_{n-1} \cup i, k_{n+1}}).
\end{aligned}
\end{gather}

In the third line, we used that a (tensor) product over sums can always be written as a sum of all the possible ways one can take the (tensor) product with each factor being a term of one of the initial sums. This form also allowed us to order the tensor product by the agent spaces and get rid of the cumbersome operator $symm$. Note that in the sums in the third line $k^i_{n+1}$ only goes over $k^i_{n+1} \in \C{N} \backslash (\C{K}^i_{n-1} \cup i)$. For a three-message state, we obtain analogously

\begin{gather}\label{eq:Vapplied3}
\begin{aligned}
    \bar{V}_{n+1}&(\bigotimes_{i=1}^{3} \ket{\psi^i, t=2n+1}^{A^O_i} \ket{\C{K}^i_{n-1}, i}) \\
    &= \sum_{k_{n+1}^1} \sum_{k_{n+1}^2} \sum_{k_{n+1}^3} \bigotimes_{k_{n+1}} (\bigodot_{i: k^i_{n+1} = k_{n+1}} V^{\rightarrow k_{n+1}}_{\C{K}^i_{n-1}, i} \ket{\psi^i} \ket{t=2n+2} \ket{\C{K}^i_{n-1} \cup i, k_{n+1}}).
\end{aligned}
\end{gather}



If we now take the inner product, we obtain

\begin{gather}\label{eq:norm}
\begin{aligned}
    \sum_{k^1_{n+1}, k^2_{n+1}} \prod_{i=1}^{2} \bra{\psi^i} V^{\rightarrow k^i_{n+1} \dagger}_{\C{K}^i_{n-1}, i} V^{\rightarrow k^i_{n+1}}_{\C{K}^i_{n-1}, i} \ket{\phi^i} \\
    \sum_{k^1_{n+1}, k^2_{n+1}, k^3_{n+1}} \prod_{i=1}^{3} \bra{\psi^i} V^{\rightarrow k^i_{n+1} \dagger}_{\C{K}^i_{n-1}, i} V^{\rightarrow k^i_{n+1}}_{\C{K}^i_{n-1}, i} \ket{\phi^i}
\end{aligned}
\end{gather}
with the top line in case of a two-message states and the bottom line for three-message states.

Let us show how we obtain this form. Firstly, notice the lack of cross terms, by which we mean a term where the bra and the ket correspond to two different distributions of the photons. Let us show that such terms must always vanish. If there are cross terms, there must be a message which is initially with some agent $i$ and is then sent to some other agent $j$ in the bra but not in the ket. There are now two possibilities. The agent $j$ receives no message in the ket in which case the entire term vanishes because the inner product between the vacuum and a non-vacuum state vanishes. Alternatively, agent $j$ receives a message from some agent $k \neq i$. But then the term contains a factor of the form of \cref{eq:orthog} with $k_n=i$ and $k'_n = k$ and thus vanishes as well. Therefore, there are no cross terms in the norm.

Secondly, any term in \cref{eq:Vapplied2} where $k^1_{n+1} \neq k^2_{n+1}$ (i.e. the messages are sent to different agents) contributes $\prod_{i=1}^{2} \bra{\psi^i} V^{\rightarrow k^i_{n+1} \dagger}_{\C{K}^i_{n-1}, i} V^{\rightarrow k^i_{n+1}}_{\C{K}^i_{n-1}, i} \ket{\phi^i}$ to \cref{eq:norm}. Similarly, any term in \cref{eq:Vapplied3} where the messages are all sent to different agents contributes $\prod_{i=1}^{3} \bra{\psi^i} V^{\rightarrow k^i_{n+1} \dagger}_{\C{K}^i_{n-1}, i} V^{\rightarrow k^i_{n+1}}_{\C{K}^i_{n-1}, i} \ket{\phi^i}$.

Lastly, we need to ask what happens with terms from \cref{eq:Vapplied2} and \cref{eq:Vapplied3} where multiple messages are sent to the same agent in which case we have to symmetrize. We can assume that the number of messages that a single agent receives is two because an agent cannot receive a message from themselves, so in a 3-partite scenario, they can at most receive two messages. Additionally, due to the dynamical switch being invariant under cyclic permutation of the agents, we can assume that the messages are coming from agents 1 and 2 and are sent to agent 3. We then calculate explicitly the inner product of such a 2-message states on the output wire of agent 3 assuming the tensor factors have the same form as in \cref{eq:Vapplied2} and \cref{eq:Vapplied3} and using \cref{eq:fock_inner_n} which gives us the expression for the inner product of $n$-message states

\begin{gather}
\begin{aligned}
    (\bra{\psi^1} V^{\rightarrow 3 \dagger}_{\C{K}^1_{n-1}, 1}& \bra{t=2n+2} \bra{\C{K}^1_{n-1} \cup 1, 3} \odot \bra{\psi^2} V^{\rightarrow 3 \dagger}_{\C{K}^2_{n-1}, 2} \bra{t=2n+2} \bra{\C{K}^2_{n-1} \cup 2, 3}) \\
    (V^{\rightarrow 3}_{\C{K}^1_{n-1}, 1} \ket{\phi^1}& \ket{t=2n+2} \ket{\C{K}^1_{n-1} \cup 1, 3} \odot V^{\rightarrow k_{n+1}}_{\C{K}^2_{n-1}, 2} \ket{\phi^2} \ket{t=2n+2} \ket{\C{K}^2_{n-1} \cup 2, 3}) \\
    &= \sum_{\pi \in S^2} \prod_{i=1}^2 \bra{\psi^{i}} V^{\rightarrow 3 \dagger}_{\C{K}^{i}_{n-1}, i} V^{\rightarrow 3}_{\C{K}^{\pi(i)}_{n-1}, \pi(i)} \ket{\phi^{\pi(i)}} \braket{\C{K}^{i}_{n-1} \cup i, k_{n+1}|\C{K}^{\pi(i)}_{n-1} \cup \pi(i), k_{n+1}} \\
    &= \prod_{i=1}^2 \bra{\psi^i} V^{\rightarrow k_{n+1} \dagger}_{\C{K}^i_{n-1}, i} V^{\rightarrow k_{n+1}}_{\C{K}^i_{n-1}, i} \ket{\phi^i}.
\end{aligned}
\end{gather}

Note that $S^2$ consists of two elements, the identity and the transposition. However, due to \cref{eq:orthog} the term corresponding to the transposition vanishes because $\pi(i) \neq i$.

Putting all this together, we obtain \cref{eq:norm}. Exchanging the sum and the product again, we arrive at 

\begin{equation}
    \prod_{i=1}^{2,3} \sum_{k^i_{n+1} \not \in \C{K}^i_{n-1} \cup i} \bra{\psi^i} V^{\rightarrow k^i_{n+1} \dagger}_{\C{K}^i_{n-1}, i} V^{\rightarrow k^i_{n+1}}_{\C{K}^i_{n-1}, i} \ket{\phi^i}.
\end{equation}
where $2,3$ indicates that there can either be two or three factors depending on the number of messages. Now consider that $\sum_{k^i_{n+1} \not \in \C{K}^i_{n-1} \cup i} V^{\rightarrow k^i_{n+1} \dagger}_{\C{K}^i_{n-1}, i} V^{\rightarrow k^i_{n+1}}_{\C{K}^i_{n-1}, i}$ is the identity on $\C{H}^{A^O_i}$ which can be seen by inspection of \cref{eq:switchkraus}. We thus find the norm after applying the operator $\bar{V}_{n+1}$ to be

\begin{equation}
    \prod_{i=1}^{2,3} \braket{\psi^i|\phi^i}.
\end{equation}
which is the inner product between $\bigotimes_{i=1}^{2,3} \ket{\psi^i, t=2n+1}^{A^O_i} \ket{\C{K}^i_{n-1}, i}$ and $\bigotimes_{i=1}^{2,3} \ket{\phi^i, t=2n+1}^{A^O_i} \ket{\C{K}^i_{n-1}, i}$. Thus $\bar{V}_1, \bar{V}_2, \bar{V}_3, \bar{V}_4$ are isometries. 
\end{proof}

\switchequivalence*

\begin{proof}

We begin by showing that the effective Choi vector lies in the space

\begin{gather}
\begin{aligned}\label{eq:effHeff}
    \text{Span}[\{\bigotimes_{n=1}^3 \ket{i_{k_n}, t=2n}^{\bar{A}^I_{k_n}} \otimes \ket{\Omega, \cancel{2n}}^{\bar{A}^I_{k_n}} \otimes  \ket{j_{k_n}, t=2n+1}^{\bar{A}^O_{k_n}} \otimes \ket{\Omega, \cancel{2n+1}}^{\bar{A}^O_{k_n}}
    \otimes \ket{i_F, t=8}^{\bar{F}}\}_{\substack{i_{k_1} \neq \Omega; \\i_{k_n}, j_{k_n}: \\ i_{k_n} = \Omega \\ \implies j_{k_n} = \Omega}}.
\end{aligned}
\end{gather}

This space is obtained by imposing two additional constraints on the space $\C{H}^{\text{eff}}$ from \cref{eq:Heff}: the first constraint is that at any time there is a message on at most a single wire. The second constraint is that at $t=2$, there cannot be vacuum on all input wires. The second constraint follows immediately from how we defined $V_1$ as it always prepares a non-vacuum state on one of the wires which means that $i_{k_1} \neq \Omega$. For the first constraint, note that in particular the basis defined above is a subset of the basis that spans $\C{H}^{\text{eff}}$. We can therefore consider an expansion of the effective Choi vector in terms of the basis of $\C{H}^{\text{eff}}$ and show that only basis states of the basis from \cref{eq:effHeff} have non-vanishing coefficients. Assume there is a term with a non-vacuum state on two output wires at some time $t$. W.l.o.g. we can assume these wires to be wires 1 and 2. The properties of the effective Choi representation then demand that at $t-1$ there are non-vacuum states on input wires 1 and 2 and that for all other times $t' \neq t-1$ there is vacuum on the input wires 1 and 2. 

There are now two possibilities. The first is that both messages on the output wires are sent to the same agent. This can be either the third agent or the global future depending on $t$. However, the effective Choi vector does not contain terms where one agent receives multiple messages. The second possibility is that the state from output wire 2 is sent to the input wire 1 (or vice versa). However, this contradicts what we said earlier about the vacuum being on the input wires 1 and 2 for all times $t' \neq t-1$. In either case, the coefficient must vanish.

Assume now there is a term with a non-vacuum state on two input wires at some time $t$. We defined the isometries $\bar{V}_1, \bar{V}_2, \bar{V}_3, \bar{V}_4$ such that they conserve the number of non-vacuum states. Therefore, there must have been two non-vacuum states on the output wires at $t-1$. They could not have been on the same wire as this would not be a term of the effective Choi vector. But if they were on different output wires, then the coefficient must vanish as we already argued.

Let us now consider a basis state from \cref{eq:effHeff} and see what happens when we compose it with the pure local operations from \cref{def:local_equivalence}

\begin{gather}\label{eq:basiscomp}
\begin{aligned}
    (\dket{\bar{A}_1} &\otimes \dket{\bar{A}_2} \otimes \dket{\bar{A}_3}) * \\
    &\bigotimes_{n=1}^3 \ket{i_{k_n}, t=2n}^{\bar{A}^I_{k_n}} \ket{\Omega, \cancel{2n}}^{\bar{A}^I_{k_n}} \ket{j_{k_n}, t=2n+1}^{\bar{A}^O_{k_n}} \ket{\Omega, \cancel{2n+1}}^{\bar{A}^O_{k_n}} \ket{i_F, t=8}^{\bar{F}} \\
    =& (\bigotimes_{n=1}^3 \bigotimes_{t \in \C{T}^I} (\dket{A_n} + \ket{\Omega}^{\bar{A}^I_{k_n}} \ket{\Omega}^{\bar{A}^O_{k_n}}) \otimes \ket{t+1} \ket{t}) * \\
    &\bigotimes_{n=1}^3 \ket{i_{k_n}, t=2n}^{\bar{A}^I_{k_n}} \ket{\Omega, \cancel{2n}}^{\bar{A}^I_{k_n}}  \ket{j_{k_n}, t=2n+1}^{\bar{A}^O_{k_n}} \ket{\Omega, \cancel{2n+1}}^{\bar{A}^O_{k_n}}
    \ket{i_F, t=8}^{\bar{F}} \\
    =& \prod_{n=1}^3 (\dket{A_{k_n}} + \ket{\Omega}^{\bar{A}^I_{k_n}} \ket{\Omega}^{\bar{A}^O_{k_n}})\\
    &* (\ket{i_{k_n}}^{\bar{A}^I_{k_n}} \ket{j_{k_n}}^{\bar{A}^O_{k_n}}) ((\ket{\Omega}^{\bar{A}^I_{k_n}} \ket{\Omega}^{\bar{A}^O_{k_n}}) * (\ket{\Omega}^{\bar{A}^I_{k_n}} \ket{\Omega}^{\bar{A}^O_{k_n}}))^3 \ket{i_F, t=8} \\
    =& \begin{cases}
    \prod_{n=1}^3 \dket{A_{k_n}} * (\ket{i_{k_n}}^{A^I_{k_n}} \ket{j_{k_n}}^{A^O_{k_n}}) \ket{i_F, t=8}, &\text{ if } i_{k_n}, j_{k_n} \neq \Omega, \forall n \\ 
    0, &\text{ else} 
    \end{cases}
\end{aligned}
\end{gather}

In the first equality, we used the definition of equivalent local operations. In the second equality, we contracted the time stamps and also used that $\dket{A_{k_n}} * (\ket{\Omega}^{\bar{A}^I_{k_n}} \otimes \ket{\Omega}^{\bar{A}^O_{k_n}}) = 0$ for all $k_n$. Finally, we used that both the local and internal operations map non-vacuum states to non-vacuum states which comes from FAA and the definition of the internal operations, respectively. The contraction thus vanishes unless all $i_{k_n}, j_{k_n}$ are either the vacuum or the non-vacuum. However, as we stated earlier, $i_{k_1} \neq \Omega$ and so only the latter is actually possible.

Additionally, note that there is an obvious isomorphism between the global future state in the above equation and the global future in the QC-QC, namely $\ket{i_F, t=8} \cong \ket{i_F}$. We will thus drop the time stamp and the vacuum states of the global future from here on out.

The coefficient of the basis state in \cref{eq:basiscomp} can be found using the definition of the Choi vector. The Choi vector of some operator $T: \C{H}^X \rightarrow \C{H}^Y$ is $\mathbb{1}^X \otimes T \dket{\mathbb{1}}^{XX} = \sum_i \ket{i}^X T \ket{i}^X = \sum_{ij} \ket{i}^X \ket{j}^Y \bra{j}^Y T \ket{i}^X = \sum_{ij} \bra{j} T \ket{i} \ket{i}^X \ket{j}^Y$. Applied to our basis state and with the help of \cref{eq:novel_single}, we find the coefficient to be

\begin{equation}\label{eq:coeff}
    \bra{i_F} V^{\rightarrow F}_{\{k_1,k_2\}, k_3} \ket{j_{k_3}} * \bra{i_{k_3}} V^{\rightarrow k_3}_{\{k_1\}, k_2} \ket{j_{k_2}} * \bra{i_{k_2}} V^{\rightarrow k_2}_{\emptyset, k_1} \ket{j_{k_1}} * \bra{i_{k_1}} V^{\rightarrow k_1}_{\emptyset, \emptyset}.
\end{equation}

With \cref{eq:basiscomp} and \cref{eq:coeff}, we can now calculate the composition $(\dket{\bar{A}_1} \otimes \dket{\bar{A}_2} \otimes \dket{\bar{A}_3}) * \dket{\bar{V}}$. 

\begin{gather}
\begin{aligned}
    (\dket{\bar{A}_1} &\otimes \dket{\bar{A}_2} \otimes \dket{\bar{A}_3})*\dket{\bar{V}} \\
    \cong & \sum_{(k_1,k_2,k_3)} \sum_{\substack{i_1, i_2, i_3, \\j_1, j_2, j_3, i_F}} \bra{i_F} V^{\rightarrow F}_{\{k_1,k_2\}, k_3} \ket{j_{k_3}} * \bra{i_{k_3}} V^{\rightarrow k_3}_{\{k_1\}, k_2} \ket{j_{k_2}} * \bra{i_{k_2}} V^{\rightarrow k_2}_{\emptyset, k_1} \ket{j_{k_1}} * \bra{i_{k_1}} V^{\rightarrow k_1}_{\emptyset, \emptyset} \\
    &\prod^{3}_{n=1} \dket{A_{k_n}} * (\ket{i_{k_n}}^{A^I_{k_n}} \ket{j_{k_n}}^{A^O_{k_n}}) \ket{i_F} \\
    =& \sum_{(k_1,k_2,k_3)} \prod_{n=1}^{3*} \dket{A_{k_n}} * \sum_{i_{k_{n+1}}, j_{k_n}} \bra{i_{k_{n+1}}} V^{k_{n+1}}_{\C{K}_{n-1}, k_n} \ket{j_{k_n}} (\ket{i_{k_n}}^{A^I_{k_n}} \ket{j_{k_n}}^{A^O_{k_n}}) \\
    =& \sum_{(k_1,k_2,k_3)} \prod_{n=1}^{3*} \dket{A_{k_n}} * \dket{V^{k_{n+1}}_{\C{K}_{n-1}, k_n}} \\
    =& (\dket{A_1} \otimes \dket{A_2} \otimes \dket{A_3}) * \ket{w}.
\end{aligned}
\end{gather}

In the above equation, one should interpret $k_4 = F$ and $\ket{j_{k_0}} = 1$. The star in the upper limit of the products in the second and third equality indicate that these products are actually link products (as we have to contract over the ancillaries). In the third equality, we used again that $\dket{T} = \sum_{ij} \bra{i} T \ket{j} \ket{i}^X \ket{j}^Y$. In the last equality, we used commutativity of the link product and the definition of the process vector of the QC-QC. 

This shows that the QC-QC and process box description of the dynamical switch are equivalent. 

\end{proof}

\symmiso*

\begin{proof}
The proof is somewhat similar to the one we gave for the analogous statement in the case of the dynamical switch. We begin by noting that just like for the dynamical switch the maps as we defined them here conserve the number of messages and thus it is enough to consider what they do to the inner product of two states with the same number of messages. However, it is not immediately clear that $\bar{V}_{n+1}$ conserves orthogonality between states with differing controls. In the proof for the dynamical switch, we used that $V^{\rightarrow k_{n+1} \dagger}_{\C{K}_{n-1}, k_n} V^{\rightarrow k_{n+1}}_{\C{L}_{n-1}, l_n} = 0$ if $\C{K}_{n-1} \cup k_n = \C{L}_{n-1} \cup l_n$ but $k_n \neq l_n$. This does not hold in the general $N$-partite case unless one also sums over all allowed values for $k_{n+1}$ as noted in \cref{lemma:krausiso},

\begin{equation}\label{eq:krausiso}
    \sum_{k_{n+1}} V^{\rightarrow k_{n+1} \dagger}_{\C{K}_{n-1}, k_n} V^{\rightarrow k_{n+1}}_{\C{L}_{n-1}, l_n} = \mathbb{1}^{A^O_{k_{n}}} \delta_{\C{K}_{n-1}, \C{L}_{n-1}} \delta_{k_n, l_n}.
\end{equation}

We therefore have to check that the inner product of two states of the form $\bigotimes_{k_n} \ket{\psi^{k_n}, t=2n+1} \ket{\C{K}^{k_n}, k_n}$ and $\bigotimes_{l_n} \ket{\phi^{l_n}, t=2n+1} \ket{\C{L}^{l_n},l_n}$ with equal number of messages is conserved under application of $\bar{V}_{n+1}$. States of this form span the space of states with a specific number of messages and therefore this is enough by condition 3 of \cref{lemma:isochar}. To do this, we first expand \cref{eq:pb_ext} and write the tensor product of sums as a sum of tensor products. We did the same in the proof of \cref{prop:3switchiso} in \cref{eq:Vapplied2} and \cref{eq:Vapplied3} and the idea is exactly the same. However, for a general process we cannot assume w.l.o.g. which $m$ wires carry messages in an $m$-message state. This makes it more difficult to write down what we sum over. Having $m$ sums where each sum corresponds to one message in $\bigotimes_{k_n} \ket{\psi^{k_n}, t=2n+1} \ket{\C{K}^{k_n}, k_n}$ and $\bigotimes_{l_n} \ket{\phi^{l_n}, t=2n+1} \ket{\C{L}^{l_n},l_n}$ and goes over all agents that this message can be sent to essentially means that we are summing over all the possible ways all messages can be distributed to new agents. We can denote a particular distribution with the help of the Cartesian product $\bigtimes$ as $\bigtimes_{k_n} k_{n+1}^{k_n} \in \bigtimes_{k_n} \C{N}^{k_n}$ where $\C{N}^{k_n} = \C{N} \backslash (\C{K}^{k_n}_{n-1} \cup k_n)$ is the set of agents that could potentially receive a message from agent $k_n$. We can then combine the $n$ sums into a single sum going over $\bigtimes_{k_n} k_{n+1}^{k_n} \in \bigtimes_{k_n} \C{N}^{k_n}$ where $\C{N}^{k_n} = \C{N} \backslash (\C{K}^{k_n}_{n-1} \cup k_n)$. Note that this is exactly what we did implicitly when writing the two sums from \cref{eq:Vapplied2} and the three sums from \cref{eq:Vapplied3} as a single sum in \cref{eq:norm}. We simply did not write it as a Cartesian product. We can thus write

\begin{gather}
\begin{aligned}
    symm&(\bigotimes_{k_n} \sum_{k^{k_n}_{n+1}} V^{\rightarrow k^{k_n}_{n+1}}_{\C{K}^{k_n}_{n-1}, k_n} \ket{\psi^{k_n}}^{A^O_{k_n}} \ket{t=2n+2} \ket{\C{K}^{k_n}_{n-1}, k_n})) \\
    &= symm(\sum_{\bigtimes_{k_n} k^{k_n}_{n+1}} \bigotimes_{k_{n+1}} (\bigotimes_{k^{k_n}_{n+1}: k^{k_n}_{n+1} = k_{n+1}} V^{\rightarrow k_{n+1}}_{\C{K}^{k_n}_{n-1}, k_n} \ket{\psi^{k_n}}^{A^O_{k_n}} \ket{t=2n+2} \ket{\C{K}^{k_n}_{n-1} \cup k_n, k_{n+1}^{k_n}})) \\
    &= \sum_{\bigtimes_{k_n} k^{k_n}_{n+1}} \bigotimes_{k_{n+1}} (\bigodot_{k^{k_n}_{n+1}: k^{k_n}_{n+1} = k_{n+1}} V^{\rightarrow k_{n+1}}_{\C{K}^{k_n}_{n-1}, k_n} \ket{\psi^{k_n}}^{A^O_{k_n}} \ket{t=2n+2} \ket{\C{K}^{k_n}_{n-1} \cup k_n, k_{n+1}^{k_n}}).
\end{aligned}
\end{gather}

Note that the sum over $\bigtimes_{k_n} k^{k_n}_{n+1}$ should be interpreted as a sum over $\bigtimes_{k_n} k^{k_n}_{n+1} \in \bigtimes_{k_n} \C{N}^{k_n}$, essentially following our convention regarding sums over $k_{n+1}$. With this form, we obtain the following expression for the inner product

\begin{gather}\label{eq:inner}
\begin{aligned}
    \sum_{\bigtimes_{k_n} k^{k_n}_{n+1}} & \bigotimes_{k_{n+1}} (\bigodot_{k^{k_n}_{n+1}: k^{k_n}_{n+1} = k_{n+1}} \bra{\psi^{k_n}}^{A^O_{k_n}}  V^{\rightarrow k_{n+1} \dagger}_{\C{K}^{k_n}_{n-1}, k_n} \bra{t=2n+2} \bra{\C{K}^{k_n}_{n-1} \cup k_n, k_{n+1}^{k_n}} \\
    \sum_{\bigtimes_{l_n} l^{l_n}_{n+1}}& \bigotimes_{l_{n+1}} (\bigodot_{l^{k_n}_{n+1}: l^{l_n}_{n+1} = l_{n+1}} V^{\rightarrow l_{n+1}}_{\C{L}^{l_n}_{n-1}, l_n} \ket{\phi^{l_n}}^{A^O_{l_n}} \ket{t=2n+2} \ket{\C{L}^{l_n}_{n-1} \cup l_n, l_{n+1}^{l_n}} \\
    =& \sum_{\substack{\bigtimes_{k_n} k^{k_n}_{n+1}, \\ \bigtimes_{l_n} l^{l_n}_{n+1}, \\ \text{compatible}}} \prod_{k_{n+1}} (\bigodot_{k^{k_n}_{n+1} = k_{n+1}} \bra{\psi^{k_n}}^{A^O_{k_n}}  V^{\rightarrow k_{n+1} \dagger}_{\C{K}^{k_n}_{n-1}, k_n} \bra{\C{K}^{k_n}_{n}, k_{n+1}^{k_n}}) (\bigodot_{l^{l_n}_{n+1} = l_{n+1}} V^{\rightarrow l_{n+1}}_{\C{L}^{l_n}_{n-1}, l_n} \ket{\phi^{l_n}}^{A^O_{l_n}} \ket{\C{L}^{l_n}_{n}, l_{n+1}^{l_n}}.
\end{aligned}
\end{gather}
where we abbreviated the subscript under the symmetric tensor products in the last line and $\C{K}^{k_n}_n = \C{K}^{k_n}_{n-1} \cup k_n$ and $\C{L}^{k_n}_n = \C{L}^{k_n}_{n-1} \cup l_n$. The word ``compatible" in the subscript of the sum in the last line denotes that we only have to consider $\bigtimes_{k_n} k^{k_n}_{n+1}, \bigtimes_{l_n} l^{l_n}_{n+1}$ where for each $k_{n+1}$ the number of times $k_{n+1}$ appears in each Cartesian product is the same. This is because if these numbers differ for some $k_{n+1}$, this term will vanish as the inner product between two states with different numbers of messages on the input wire of $k_{n+1}$ always vanishes.

Let us now look at just the part coming after the product over $k_{n+1}$, i.e. the inner product of what is on the input wire of $k_{n+1}$.



We can use \cref{eq:fock_inner_n} to find this inner product. However, as the messages are not necessarily labeled with $\{1,...,m\}$ use of the permutation group $S^m$ is not directly possible. However, note that $S^m$ is just the set of bijections from $\{1,...,m\}$ to itself and that it therefore has a one-to-one correspondence to any other set of bijections from an $m$-element set to another (not necessarily the same) $m$-element set. 

Going back to \cref{eq:inner}, we therefore define for a given choice of $\bigtimes_{k_n} k^{k_n}_{n+1}, \bigtimes_{l_n} l^{l_n}_{n+1}$ the set $S^{k_{n+1}}$ as the set of bijections from $\{k_n|k^{k_n}_{n+1} = k_{n+1}\}$ to $\{l_n|l^{l_n}_{n+1} = l_{n+1}\}$. We can then write the inner product as

\begin{equation}
    \sum_{\substack{\bigtimes_{k_n} k^{k_n}_{n+1},\\ \bigtimes_{l_n} l^{l_n}_{n+1}, \\ \text{compatible}}} \prod_{k_{n+1}} \sum_{\pi \in S^{k_{n+1}}} \prod_{k_n: k^{k_n}_{n+1} = k_{n+1}} \bra{\psi^{k_n}}^{A^O_{k_n}}  V^{\rightarrow k_{n+1} \dagger}_{\C{K}^{k_n}_{n-1}, k_n} V^{\rightarrow k_{n+1}}_{\C{L}^{\pi(k_n)}_{n-1}, \pi(k_n)} \ket{\phi^{\pi(k_n)}}^{A^O_{\pi(k_n)}} \braket{\C{K}^{k_n}_{n}, k_{n+1}^{k_n}|\C{L}^{\pi(k_n)}_{n}, l_{n+1}^{\pi(k_n)}})
\end{equation}

What we would now like to do is pulling the first sum inside the products again and then use \cref{eq:krausiso} to turn the operators into identities. However, the $k_n$ over which the innermost product goes depends on our choice of $\bigtimes_{k_n} k^{k_n}_{n+1}, \bigtimes_{l_n} l^{l_n}_{n+1}$, which means we cannot exchange the sum and the product as is. We thus need to find some way to get rid of this dependence.

We first exchange the first product with the second sum by using the same expansion trick that we used before $\prod_{k_{n+1}} \sum_{\pi \in S^{k_{n+1}}} = \sum_{\bigtimes_{k_{n+1}} \pi^{k_{n+1}} \in \bigtimes S^{k_{n+1}}} \prod_{k_{n+1}}$. We can then combine the two products to obtain

\begin{gather}
\begin{aligned}
    \sum_{\substack{\bigtimes_{k_n} k^{k_n}_{n+1},\\ \bigtimes_{l_n} l^{l_n}_{n+1}, \\ \text{compatible}}} \sum_{\bigtimes_{k_{n+1}} \pi^{k_{n+1}}} \prod_{k_n} \bra{\psi^{k_n}}^{A^O_{k_n}}  V^{\rightarrow k_{n+1} \dagger}_{\C{K}^{k_n}_{n-1}, k_n} V^{\rightarrow k_{n+1}}_{\C{L}^{\pi^{k^{k_n}_{n+1}}(k_n)}_{n-1}, \pi^{k^{k_n}_{n+1}}(k_n)} \ket{\phi^{\pi^{k^{k_n}_{n+1}}(k_n)}}^{A^O_{\pi^{k^{k_n}_{n+1}}(k_n)}} \\ \cdot \braket{\C{K}^{k_n}_{n}, k_{n+1}^{k_n}|\C{L}^{\pi^{k^{k_n}_{n+1}}(k_n)}_{n}, l_{n+1}^{\pi^{k^{k_n}_{n+1}}(k_n)}})
\end{aligned}
\end{gather}

Now notice that $\bigtimes S^{k_{n+1}}$ is a subset of all the bijections from $\{k_n\}$ to $\{l_n\}$ where $\{k_n\}$ is the set of agents with a non-vacuum state for one vector and $\{l_n\}$ is the analogous set for the other vector. We shall denote the set of all such bijections as $S$. Additionally, if $\pi \in S \backslash \bigtimes S^{k_{n+1}}$, then there exists at least one $k_n$ such that $k^{k_n}_{n+1} \neq l^{\pi(k_n)}_{n+1}$. What this means is that we can replace the sum over $\bigtimes_{k_{n+1}} \pi^{k_{n+1}} \in \bigtimes S^{k_{n+1}}$ with a sum over $\pi \in S$ because any term corresponding to $\pi \not \in \bigtimes S^{k_{n+1}}$ will vanish as there is a factor where $\braket{\C{K}^{k_n}_{n-1} \cup k_n, k_{n+1}^{k_n}|\C{L}^{\pi(k_n)}_{n-1} \cup \pi(k_n), l_{n+1}^{\pi(k_n)}}) = 0$ because $k^{k_n}_{n+1} \neq l^{\pi(k_n)}_{n+1}$. Additionally, we can also drop the requirement that $\bigtimes_{k_n} k^{k_n}_{n+1}, \bigtimes_{l_n} l^{l_n}_{n+1}$ are compatible as terms where this is not the case vanish anyway. We thus have

\begin{gather}
\begin{aligned}
    \sum_{\substack{\bigtimes_{k_n} k^{k_n}_{n+1}, \\ \bigtimes_{l_n} l^{l_n}_{n+1}}} & \sum_{\pi \in S} \prod_{k_n} \bra{\psi^{k_n}}^{A^O_{k_n}}  V^{\rightarrow k^{k_n}_{n+1} \dagger}_{\C{K}^{k_n}_{n-1}, k_n} V^{\rightarrow l^{\pi(k_n)}_{n+1}}_{\C{L}^{\pi(k_n)}_{n-1}, \pi(k_n)} \ket{\phi^{\pi(k_n)}}^{A^O_{\pi(k_n)}} \braket{\C{K}^{k_n}_{n}, k_{n+1}^{k_n}|\C{L}^{\pi(k_n)}_{n}, l_{n+1}^{\pi(k_n)}}) \\
    &= \sum_{\pi \in S} \prod_{k_n} \sum_{k^{k_n}_{n+1}, l^{\pi(k_n)}_{n+1}} \bra{\psi^{k_n}}^{A^O_{k_n}}  V^{\rightarrow k^{k_n}_{n+1} \dagger}_{\C{K}^{k_n}_{n-1}, k_n} V^{\rightarrow l^{\pi(k_n)}_{n+1}}_{\C{L}^{\pi(k_n)}_{n-1}, \pi(k_n)} \ket{\phi^{\pi(k_n)}}^{A^O_{\pi(k_n)}} \braket{\C{K}^{k_n}_{n}, k_{n+1}^{k_n}|\C{L}^{\pi(k_n)}_{n}, l_{n+1}^{\pi(k_n)}}) \\
    &= \sum_{\pi \in S} \prod_{k_n} \sum_{k^{k_n}_{n+1}} \bra{\psi^{k_n}}^{A^O_{k_n}}  V^{\rightarrow k^{k_n}_{n+1} \dagger}_{\C{K}^{k_n}_{n-1}, k_n} V^{\rightarrow k^{k_n}_{n+1}}_{\C{L}^{\pi(k_n)}_{n-1}, \pi(k_n)} \ket{\phi^{\pi(k_n)}}^{A^O_{\pi(k_n)}} \delta_{\C{K}^{k_n}_{n-1} \cup k_n, \C{L}^{\pi(k_n)}_{n-1} \cup \pi(k_n)} \\
    &= \sum_{\pi \in S} \prod_{k_n} \braket{\psi^{k_n}|\phi^{\pi(k_n)}} \delta_{\C{K}^{k_n}_{n-1}, \C{L}^{\pi(k_n)}_{n-1}} \delta_{k_n, \pi(k_n)}
\end{aligned}
\end{gather}

In the last equality, we used \cref{lemma:krausiso}. Now notice that in the last equality the term corresponding to a specific $\pi \in S$ vanishes unless $k_n = \pi(k_n)$ for all $k_n$. However, as $\pi$ is a bijection this is only possible if its domain and codomain are the same and $\pi$ is the identity. Thus, the inner product vanishes if $\{k_n\} \neq \{l_n\}$. We see that $\bar{V}_{n+1}$ conserves the inner product in this case.

Now let us assume that $\{k_n\} = \{l_n\}$. In that case, the inner product becomes

\begin{gather}
\begin{aligned}
    \prod_{k_n}& \braket{\psi^{k_n}|\phi^{k_n}} \delta_{\C{K}^{k_n}_{n-1}, \C{L}^{k_n}_{n-1}} \\
    &= \prod_{k_n} \braket{\psi^{k_n}|\phi^{k_n}} \braket{\C{K}^{k_n}_{n-1}, k_n| \C{L}^{k_n}_{n-1}, k_n} \\
    &= \bigotimes_{k_n} \bra{\psi^{k_n}}^{A^O_{k_n}} \bra{\C{K}^{k_n}_{n-1}, k_n} \bigotimes_{l_n} \ket{\phi^{k_n}}^{A^O_{l_n}} \ket{\C{L}^{l_n}_{n-1}, l_n}.
\end{aligned}
\end{gather}

The inner product is conserved and therefore $\bar{V}_{n+1}$ is an isometry.

\end{proof}

\krausiso*

\begin{proof}
The lemma is a direct consequence of the maps $V_{n+1}: \C{H}^{A^O \alpha_n C_{n}} \rightarrow \C{H}^{A^I \alpha_{n+1} C_{n+1}}$ in the QC-QC picture being isometries, 

\begin{equation}\label{eq:QCQCiso}
V_{n+1}^\dagger V_{n+1} = \mathbb{1}^{A^O \alpha_n C_{n}} = \sum_{\C{K}_{n-1}, k_n} \mathbb{1}^{A^O_{k_n} \alpha_n} \otimes \ket{\C{K}_{n-1}, k_n} \bra{\C{K}_{n-1}, k_n}.
\end{equation}

Using the definition of $V_{n+1}$ given in \cref{eq:defiso},

\begin{gather}
\begin{aligned}
    V_{n+1}^\dagger V_{n+1} &= \sum_{\substack{\C{K}_{n-1},\\ k_n, k_{n+1}}} \sum_{\substack{\C{L}_{n-1},\\ l_n, l_{n+1}}} V^{\rightarrow k_{n+1} \dagger}_{\C{K}_{n-1}, k_n} V^{\rightarrow l_{n+1}}_{\C{L}_{n-1}, l_n} \otimes \ket{\C{K}_{n-1}, k_n} \braket{\C{K}_{n-1} \cup k_n, k_{n+1}|\C{L}_{n-1} \cup l_n, l_{n+1}} \bra{\C{L}_{n-1}, l_n} \\
    &= \sum_{\C{K}_n} \sum_{k_n, l_n \in \C{K}_n} \sum_{k_{n+1}} V^{\rightarrow k_{n+1} \dagger}_{\C{K}_n \backslash k_n, k_n} V^{\rightarrow k_{n+1}}_{\C{K}_n \backslash l_n, l_n} \otimes \ket{\C{K}_n \backslash k_n, k_n} \bra{\C{K}_n \backslash l_n, l_n}.
\end{aligned}
\end{gather}

Comparing this to \cref{eq:QCQCiso}, we notice that terms with an off-diagonal projector over the control $\ket{\C{K}_n \backslash k_n, k_n} \bra{\C{K}_n \backslash l_n, l_n}$, $k_n \neq l_n$, must vanish. This is only possible if $\sum_{k_{n+1}} V^{\rightarrow k_{n+1} \dagger}_{\C{K}_n \backslash k_n, k_n} V^{\rightarrow k_{n+1}}_{\C{K}_n \backslash l_n, l_n} = 0$ for $k_n \neq l_n$. Meanwhile, terms where $k_n = l_n$ must be the identity which implies $\sum_{k_{n+1}} V^{\rightarrow k_{n+1} \dagger}_{\C{K}_n \backslash k_n, k_n} V^{\rightarrow k_{n+1}}_{\C{K}_n \backslash k_n, k_n} = \mathbb{1}^{A^O_{k_n} \alpha_n}$. Combining these, we obtain the statement of the lemma.

\end{proof}

\accept*

\begin{proof}
To prove the proposition, we can show that before each projective measurement the state is already in the respective subspace $\C{H}^n_{accept}$ by induction. The base case is trivial as the global past is the only agent that sends something to the process box during the first time step. It is either trivial or sends exactly one message due to FAA and WSR.

Assume now the projective measurements up to and just before the application of $V_{n+1}$ all yield $accept$. The state at this point is thus in $\C{H}^n_{accept}$ by assumption. If we then apply $V_{n+1}$, the resulting state will be in $\C{H}^{n+1}_{accepted}$.



Next, the local operations act on this state. Due to the LO assumption and the fact that $\C{H}^{n+1}_{accepted}$ only contains zero- and one-message states, at most one agent $k_{n+1}$ will send a non-vacuum state (more precisely, the state can be a superposition and in each term of the superposition only one agent sends a non-vacuum state). This non-vacuum state must be a single message because of the WSR assumption. Finally, the time stamp of all wires is $t=2(n+1)+1$ due to the OSR assumption. Thus, the resulting state is in $\C{H}^{n+1}_{accept}$ and the projective measurement yields $accept$.

\end{proof}

\abortequivalence*

\begin{proof}
Due to \cref{lemma:accept}, the effective Choi matrix is the Choi matrix restricted to $\bigotimes_n (\C{H}^{n}_{accept} \otimes \C{H}^{n}_{accepted})$. This also means that one does not have to consider the projective measurements to find the effective Choi representation as they act as identities. We calculate the Choi vector

\begin{gather}
\begin{aligned}\label{eq:choiabort}
    \dket{V_{N+1}... V_{n+1} ... V_1} = \dket{V_{N+1}} * &... * \dket{V_{n+1}} * ...\dket{V_1} \\
    = \sum_{(k_1,...,k_N)} ((\dket{V^{\rightarrow F}_{\{k_1,...,k_{N-1}\}, k_N}}& \otimes \ket{t=2N+1} \ket{t=2N+2}) \otimes \bigotimes_{i \in \C{N}\backslash k_N} \ket{\Omega, t=2N+1}^{\bar{A}^O_i}) * \\
    ... * ((\dket{V^{\rightarrow k_{n+1}}_{\{k_1,...,k_{n-1}\}, k_n}}& \otimes \ket{t=2n+1} \ket{t=2n+2}) \otimes \bigotimes_{\substack{i \in \C{N}\backslash k_n \\ j \in \C{N}\backslash k_{n+1}}} \ket{\Omega, t=2n+1}^{\bar{A}^O_i} \ket{\Omega, t=2n+2}^{A^I_j}) * \\
    ... * ((\dket{V^{\rightarrow k_1}_{\emptyset, \emptyset}} \otimes \ket{t=1}& \ket{t=2}) \otimes \bigotimes_{i \in \C{N}\backslash k_1} \ket{\Omega, t=2}^{\bar{A}^O_i}).
\end{aligned}
\end{gather}
where we used linearity of the Choi isomorphism to pull out the sums over $k_1,...,k_N$ which then turned into a single sum over the orders $(k_1,...,k_N)$ due to the link product contracting the control system. 

Next, consider a set of equivalent local operations as in \cref{def:local_equivalence}. Composing the process box local operations with the Choi vector \cref{eq:choiabort} yields

\begin{gather}
\begin{aligned}
    \bigotimes_{i=1}^N &\dket{\bar{A}_i} * \dket{V_{N+1}... V_{n+1} ... V_1} \\
    &= \bigotimes_{i=1}^N \bigotimes_{t \in \C{T}^I} (\dket{A_i} + \ket{\Omega}^{\bar{A}^I_i} \ket{\Omega}^{\bar{A}^O_i}) \otimes \ket{t+1} \ket{t}) * \\
    & \sum_{(k_1,...,k_N)} ((\dket{V^{\rightarrow F}_{\{k_1,...,k_{N-1}\}, k_N}} \otimes \ket{t=2N+1} \ket{t=2N+2}) \otimes \bigotimes_{i \in \C{N}\backslash k_N} \ket{\Omega, t=2N+1}^{\bar{A}^O_i}) * \\
    &... * ((\dket{V^{\rightarrow k_{n+1}}_{\{k_1,...,k_{n-1}\}, k_n}} \otimes \ket{t=2n+1} \ket{t=2n+2}) \otimes \bigotimes_{\substack{i \in \C{N}\backslash k_n \\ j \in \C{N}\backslash k_{n+1}}} \ket{\Omega, t=2n+1}^{\bar{A}^O_i} \ket{\Omega, t=2n+2}^{A^I_j}) * \\
    &... * ((\dket{V^{\rightarrow k_1}_{\emptyset, \emptyset}} \otimes \ket{t=1} \ket{t=2}) \otimes \bigotimes_{i \in \C{N}\backslash k_1} \ket{\Omega, t=2}^{\bar{A}^O_i}) \\
    &= \sum_{(k_1,...,k_N)} \bigotimes_{i=1}^N \dket{A_i} * (\dket{V^{\rightarrow F}_{\{k_1,...,k_{N-1}\}, k_N}} * ... * \dket{V^{\rightarrow k_{n+1}}_{\{k_1,...,k_{n-1}\}, k_n}} * ... * \dket{V^{\rightarrow k_1}_{\emptyset, \emptyset}}) \otimes \ket{t=2N+2} \\
    &= \bigotimes_{i=1}^N \dket{A_i} * \ket{w_{\C{N}, F}} \otimes \ket{t=2N+2} \\
    &\cong \bigotimes_{i=1}^N \dket{A_i} * \ket{w_{\C{N}, F}}.
\end{aligned}
\end{gather}

The process box and the QC-QC are thus equivalent.
\end{proof}

\subsection{Proofs for \texorpdfstring{\Cref{sec:characterizing}}{Section 6}}

\PBunitary*

\begin{proof}

We will drop the primes on the additional global past and future, as the global past and future of the initial process box are not relevant to the proof.

As a process box is a causal box, there exists a sequence representation $\bar{V}_1,...,\bar{V}_N$. We then have isometries $\bar{V}_n: \C{H}^{\alpha_{n-1}} \otimes \C{F}^{\C{T}_n}_X \rightarrow \C{F}^{\C{T}_{n-1}}_Y \otimes \C{H}^{\alpha_n}$ for all $n=1,...,N$. Note that $\C{T}_1 = \emptyset$ as $\bigcap_{n=1}^{N+1} \C{T}_n = \emptyset$. The wire $\alpha_N$ is an additional wire compared to the original causal box. We take this wire to be the input wire of one or more global future parties. As it is connected to the last isometry and since it has trivial output space, we can assume the set of input positions $\C{T}^I_F$ to be strictly greater than $\C{T}$. If there are multiple input positions, we can always consider the input position to be part of the message and thus identify the combined input space of multiple global future agents with multiple positions with a larger input space of a single party with a single input position $t_f$. Furthermore, because each map $\bar{V}_n$ is an isometry, it admits a unitary extension to some $\bar{U}_n: \C{H}^{\alpha_{n-1}} \otimes \C{F}^{\C{T}_n}_X \otimes \C{H}^{P_n} \rightarrow \C{F}^{\C{T}_{n-1}}_Y \otimes \C{H}^{\alpha_n}$. The wires $P_n$ are connected to $\bar{U}_n$, which means their output time stamps must be in $\C{T}_n$. They are therefore not a global past.

Our goal now is to change the unitaries $\bar{U}_n$ such that all $P_n$ are connected to $\bar{V}_1$ and their output stamps are therefore $t=t_i$. The proof idea is depicted in \cref{fig:globalpast}.

\begin{figure}
    \centering
    \includegraphics[width=\textwidth]{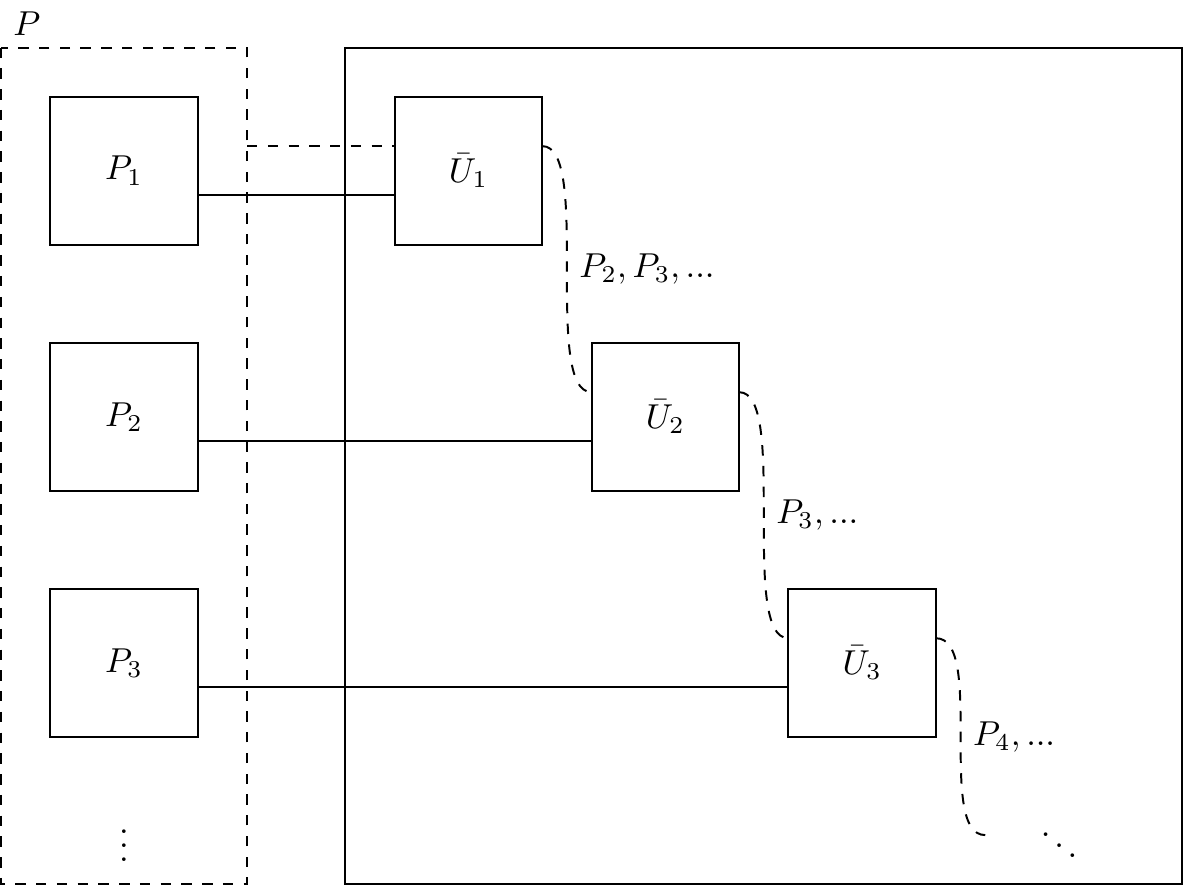}
    \caption{Instead of considering $N$ different $P_n$, each of which sends to $\bar{U}_n$ (solid lines), we can equivalently consolidate them into a single global past $P = \bigotimes_{i=1}^N P_n$ which is only connected to $U_1$ and replace the solid wires with the dashed wires in the figure. The isometry $U_1$ then acts on the part of the message that corresponds to $P_1$ and forwards everything else ($P_2, P_3,...$) along the dashed internal wire to $U_2$ which acts on $P_2$ and so on.}
    \label{fig:globalpast}
\end{figure}

Formally, we define

\begin{gather}
\begin{aligned}
    \bar{U}'_n&: \C{H}^{\alpha_{n-1}} \otimes \C{F}^{\C{T}_n}_X \otimes \C{H}^{P_n} \otimes \bigotimes_{m=i+1}^{N} \C{H}^{P_m} \rightarrow \C{F}^{\C{T}_{n+1}}_Y \otimes \C{H}^{\alpha_i} \otimes \bigotimes_{m=i+1}^{N} \C{H}^{P_m} \\
    \bar{U}'_n &:= \bar{U}_n \otimes \mathbb{1}^{P_{n+1}...P_N}.
\end{aligned}
\end{gather}

As a tensor product of unitaries $\bar{U}'_n$ is again a unitary. Furthermore, note that for the sequence representation defined by $\bar{U}'_1,...,\bar{U}'_N$, the output wires of $P_1,...,P_N$ are connected only to $\bar{U}'_1$ (while the other unitaries have input and output Hilbert spaces $\C{H}^{P_n}$, these correspond to internal wires. They are not the output wires of the agents $P_n$). Hence, we can combine these agents into a single one $\C{H}^{P} := \bigotimes_{n=1}^N \C{H}^{P_n}$ and assign the time stamp $t=t_i$ to them. 

Any sequence representation describes a causal box, but it is not clear if the sequence representation given by $\bar{U}'_1,...,\bar{U}'_N$ also describes a process box. We have to check WSR, LO, OSR and FAA for the global past and future (we have made no changes to the other agents and thus they still fulfill the requirements). 

The trivial input space of the global past and the trivial output space of the global future mean that LO is trivially fulfilled. The global past can only send during $t=1$ while the global future can only receive during the last time step. This means OSR is automatically fulfilled for any local operation. Additionally, WSR and FAA is fulfilled as long as the agents receive and send exactly one message during these time steps. This is the case by definition of the global past's output space and the global future's input space as one-message spaces.

Lastly, we need to show that we actually have a unitary extension of the initial process box. It suffices to consider a pure state $\rho^X = \ket{\psi}^X\bra{\psi}^X$ on the output wires of the agents excluding the global past. We then have

\begin{gather}
\begin{aligned}
    \dket{\bar{U}'} \dbra{\bar{U}'} * (\rho^X \otimes \ket{0}\bra{0}^P) &= \bar{U}' (\ket{\psi}^X\bra{\psi}^X \otimes \ket{0} \bra{0}^P) \bar{U}'^\dagger \\
    &= \bar{U}' (\ket{\psi}^X \otimes \ket{0}^P) (\bra{\psi}^X \otimes \ket{0}^P) \bar{U}'^\dagger \\
    &= \bar{U}_N(\bar{U}_{N-1}(...\bar{U}_1(\ket{\psi}^X \otimes \ket{0}^{P_1}) \otimes \ket{0}^{P_2})... \otimes \ket{0}^{P_{N-1}}) \otimes \ket{0}^{P_N}) (...)\bar{U}_N^\dagger \\
    &= \bar{V}_N(\bar{V}_{N-1}(...\bar{V}_1(\ket{\psi}^X))) (...)\bar{V}_N^\dagger \\
    &= \bar{V} \ket{\psi}^X \bra{\psi}^X \bar{V} \\
    &= \dket{\bar{V}} \dbra{\bar{V}} * \rho^X.
\end{aligned}
\end{gather}

In the first line we used that for any map $\C{M}$, $\rho*M = \C{M}(\rho)$. In the fourth line we used that $\bar{U}_n$ is a unitary extension of $V_i$, i.e. $\bar{U}_n(\ket{\psi} \ket{0}^{P_n}) = V_i \ket{\psi}$. As the above equation holds for any $\rho$, we have $\dket{\tilde{U}} \dbra{\tilde{U}} * \ket{0}\bra{0}^P = \dket{\bar{V}} \dbra{\bar{V}}$. Taking the link product with $\mathbb{1}^F$ is equivalent to tracing out the global future. This shows together with $V$ being a Stinespring representation that $\dket{\bar{U}} \dbra{\bar{U}}$ is in fact a unitary extension of the initial process box.

\end{proof}

\unitaryequivalence*

\begin{proof}

As $\ket{w_{\C{N}, F}}$ is pure, we can assume that $F'$ is trivial. The same is the case for $\dket{\bar{V}}$ and a purification $\dket{\bar{U}}$. Thus, the global future spaces of $\dket{U}$ and $\dket{\bar{U}}$ are isomorphic due to $\ket{w_{\C{N}, F}}$ and $\dket{\bar{V}}$ being equivalent. Assuming then w.l.o.g. that the input and output dimensions of the agents are equal, we then have that the dimension of the global past must be equal to the dimension of the global future as a unitary's domain and codomain have equal dimension. Therefore, the global past spaces used to unitarize $\ket{w_{\C{N}, F}}$ and $\dket{\bar{V}}$ must be isomorphic because the global future spaces are isomorphic.

That $\dket{\bar{U}}$ and $\dket{U}$ are equivalent, then follows from \cref{def:pb_purification} and \cref{def:pm_purification} and equivalence of $\ket{w_{\C{N}, F}}$ and $\dket{\bar{V}}$.

\end{proof}

\subsection{Proofs for \texorpdfstring{\Cref{sec:pbtoqcqc}}{Section 7}}

\relabeling*

\begin{proof}

The fact that $\C{R}$ respects the order relations of $\C{T}$ implies that $\hat{\Phi}'$ is still a valid process box that respects a causality function that is induced by the causality function of $\hat{\Phi}$ via $\C{R}$.


We then define the equivalent local operation to some $\C{M}_{\bar{A_i}}$ as $\C{M}'_{\bar{A_i}} = \C{R} \circ \C{M}_{\bar{A_i}} \circ \C{R}^{-1}$. 

The Choi matrices of the local operations of the agents and the process box are then simply

\begin{gather}
\begin{aligned}
    M'_{\bar{A_i}} = \C{R}(M_{\bar{A_i}}) \\
    \Phi' = \C{R}(\Phi).
\end{aligned}
\end{gather}

The map $\C{R}$ is essentially the identity between the Fock space with set $\C{T}$ and the Fock space with set $\C{T'}$. As can be seen from the definition of the link product \cref{def:linkmat}, applying such a map to both factors has no influence on their link product beyond having to apply the same map to the final result

\begin{gather}
\begin{aligned}
    \bigotimes_{i} M'_{\bar{A_i}} * \Phi' &= \C{R}(\bigotimes_i M_{\bar{A_i}}) * \C{R}(\Phi)\\ 
    &= \C{R}(\bigotimes_i M_{\bar{A_i}} * \Phi).
\end{aligned}
\end{gather}

The process boxes $\Phi$ and $\Phi'$ are thus operationally equivalent. Note, in particular, that if the global past and future are trivial, the link product yields a scalar in which case we have equality in the above equation even without $\C{R}$ in the last line.

\end{proof}

\stretching*

\begin{proof}

Consider an isometry $\bar{V}_n$ of the sequence representation that receives messages at $t=2n-1$ and produces outputs at $t=2n$. If the number of messages sent by $\bar{V}_n$ is at most one, regardless of the input to $\bar{V}_n$, we do nothing and pick another isometry. If there exist inputs for which the isometry sends more than one message, we proceed as follows. First, we relabel all time stamps with $t > 2n$ according to $t \mapsto t+2N-2$ where $N$ is as always the number of agents. Note that now $\C{O}_k(2n) = 2n+2N-1$ for all agents. The resulting process box is still operationally equivalent to the one we started with due to \cref{lemma:relabeling}. 

We then reintroduce the now missing input time stamps $2n+2, 2n+4,..., 2n+2N-2$ and define $\C{D}^I_k: \C{H}^{\bar{A}^I_k} \otimes \ket{t=2n} \rightarrow \C{H}^{\bar{A}^I_k} \otimes \ket{t=2n+2k-2}$ as the identity between these spaces, i.e. $\C{D}^I_k \ket{\psi, t=2n}^{\bar{A}^I_k} = \ket{\psi, t=2n+2k-2}^{\bar{A}^I_k}$. We then define $\bar{V}'_n = \bigotimes_{k=1}^N \C{D}^I_k \circ \bar{V}_n$ which is still an isometry that respects causality. Replacing $\bar{V}_n$ with $\bar{V}'_n$ in the sequence representation thus still yields a valid process box. Additionally, if we define the equivalent local operation of this new process box by applying $\C{D}^I_k$ to the original operation of $A_k$ (which simply replaces the input time $2n$ of $A_k$ with $2n+2k-2$), then the original and the new process box are operationally equivalent. This is because from a mathematical point of view we only relabeled all Hilbert spaces $\C{H}^{\bar{A}^I_k} \otimes \ket{t=2n}$ as $\C{H}^{\bar{A}^I_k} \otimes \ket{t=2n+2k-2}$ which does not change the link product (cf. the proof of \cref{lemma:relabeling}).

Next, we reintroduce the missing output time stamps $2n+1, 2n+3,...,2n+2N-1$ and redefine $\C{O}_k(2n+2k-2) = 2n+2k-1$. We replace $\bar{V}'_n$ with $\bar{V}''_n = \bar{V}'_n \otimes \mathbb{1}^{\bigotimes_{k=1}^{N-1} A^{O, t=2n+2k-1}_k \rightarrow \alpha}$ where $\mathbb{1}^{\bigotimes_{k=1}^{N-1} A^{O, t=2n+2k-1}_k \rightarrow \alpha}: \bigotimes_{k=1}^{N-1} \C{H}^{\bar{A}^O_k} \otimes \ket{t=2n+2k-1} \rightarrow \C{H}^{\alpha}$ is the identity between these spaces and $\alpha$ is the Hilbert space of a new internal wire from $\bar{V}'_n$ to $\bar{V}'_{n+1}$, where $\bar{V}'_{n+1}$ replaces $\bar{V}_{n+1}$. We define the action of $\bar{V}'_{n+1}$ to essentially be the action of $\bar{V}_{n+1}$, i.e. $\bar{V}'_{n+1} = \bar{V}_{n+1} \circ (\mathbb{1}^{A^O_N, t=2n+2N+1} \otimes \mathbb{1}^{\alpha \rightarrow \bigotimes_{k=1}^{N-1} A^{O, t=2n+2N-1}_k})$. Putting all this together means that replacing $\bar{V}'_n$ with $\bar{V}''_n$,  $\bar{V}_{n+1}$ with $\bar{V}'_{n+1}$ and $\C{O}_k(2n+2k-2) = 2n+2N$ with $\C{O}_k(2n+2k-2) = 2n+2k-1$ amounts to relabeling each $\C{H}^{\bar{A}^O_k} \otimes \ket{2n+2N-1}$ as $\C{H}^{\bar{A}^O_k} \otimes \ket{2n+2k-1}$, which does not change the link product. The process box is thus still operationally equivalent to the original.

Finally, we can add all input times $2n+2l-2$ to $\C{T}^I_k$ and all output times $2n+2l-1$ to $\C{T}^O_k$ with $l \neq k$ and define $\C{O}_k(2n+2l-2) = 2n+2l-1$. By construction, the agents can only receive the vacuum during these input times (and thus also only send the vacuum during these output times) which means that adding them does not change the link product.

The final result is that $\bar{V}''_{n}$ can only send a single message during each time step while the overall process box still has the form of \cref{conjecture:simplifying}. We can now use that $\bar{V}''_n$ itself defines a causal box and we can thus find a sequence representation of it. In particular, since $\chi(\C{T}^{\leq t})=\C{T}^{\leq t-1}$ is always a valid causality function for causal boxes on $\C{T} = \{1,...,M\}$, we can find a sequence representation such that each isometry receives messages at a single odd time and sends messages at the subsequent even time. Additionally, since $\bar{V}''_n$ is already pure, we do not need to add any new wires.

We then replace $\bar{V}''_n$ with its sequence representation. This decreases the number of isometries that can send multiple messages by one and as such if we repeat the above procedure for each $n$, we end up with a process box that sends at most a single message during each time step. Due to LO, this then also means that the process box receives only a single message during each time step.

\end{proof}

Before proving \cref{prop:pbtoqcqc}, we prove that we can add the control as discussed in \cref{sec:pbtoqcqc}. This is formalized in the following lemma.

\begin{restatable}[Adding the control]{lemma}{addcontrol}
\label{lemma:addcontrol}
Let $\Phi$ be a process box that is of the form of \cref{lemma:stretching}. Then there exists a sequence representation $\bar{V}_1,...,\bar{V}_M$ of $\Phi$ with the following properties:

\begin{enumerate}
    \item $\bar{V}_{n+1}$ has an incoming internal wire $\C{H}^{C_{\leq n}} = \bigoplus_{m=0}^n \C{H}^{C_m}$ for all $n>0$ and an outgoing internal wire $\C{H}^{C_{\leq {n+1}}} = \bigoplus_{m=0}^{n+1} \C{H}^{C_m}$ for all $n < M-1$ where $\C{H}^{C_n}$ is the space of control systems of length $n$ as defined for QC-QCs.
    \item Given a state $\ket{\psi, t=2n+1}^{\bar{A}^O_{k_m} \alpha_n} \ket{\C{K}_{m-1}, k_m}^{C_m}$, the action of $\bar{V_{n+1}}$ is
    \begin{equation}
        \bar{V}_{n+1} \ket{\psi, t=2n+1}^{\bar{A}^O_{k_m} \alpha_n} \ket{\C{K}_{m-1}, k_m}^{C_m} = \sum_{k_{m+1}} \ket{\phi^{k_{m+1}}, t=2n+2}^{A^O_{k_{m+1}} \alpha_{n+1}} \ket{\C{K}_{m-1} \cup k_m, k_{m+1}}^{C_{m+1}}
    \end{equation}
    for some bipartite $\ket{\phi^{k_{m+1}}}^{\bar{A}^O_{k_m} \alpha_n}$ or 
    \begin{equation}
    \bar{V}_{n+1} \ket{\psi, t=2n+1}^{\bar{A}^O_{k_m} \alpha_n} \ket{\C{K}_{m-1}, k_m}^{C_m} = \ket{\Omega, t=2n+2} \ket{\alpha}^{\alpha_{n+1}} \ket{\C{K}_{m-1}, k_m}^{C_m}
    \end{equation}
    for some $\ket{\alpha}^{\alpha_n} \in \C{H}^{\alpha_n}$ where $\alpha_n$ denotes some additional internal wire.
\end{enumerate}
\end{restatable}


\begin{proof}
Let $\bar{V}_1,...,\bar{V}_M$ be an arbitrary sequence representation of the process box such that $\bar{V}_n$ takes inputs during $t=2n-1$ and produces outputs during $t=2n$. We wish to add the internal wire $\C{H}^{C\leq n}$ between $\bar{V}_n$ and $\bar{V}_{n+1}$. Doing so requires us to define the function $C$ that adds the control $\ket{\C{K}_{n-1}, k_n}$ to a state if that state was in those respective labs. Adding these internal wires must not change the fact that the maps $\bar{V}_n$ are isometries. A sufficient condition for this is that $C$ itself is an isometry. As the map is then invertible on its image, we can define the new isometries of the sequence representations to be $C \circ \bar{V}_n \circ C^{-1}$, which shows that the resulting process box is still the same. On the other hand, due to the form of $C$, it is enough to check that it does not turn non-orthogonal states into orthogonal states to ensure it is an isometry.

We use \cref{lemma:stretching} so that we can assume that during each time step the appropriate isometry $\bar{V}_n$ receives or sends either 0 or 1 messages. We prove by induction over the isometries: In the first time step we essentially just add the lab (or the vacuum if the first isometry does not send anything to the agents). States in different labs are orthogonal already (and so is the vacuum to one-message states) so this is fine. Assume now that we can define $C$ up to the outputs of $\bar{V}_n$. The possible inputs to $\bar{V}_{n+1}$ are then spanned by the vectors $\ket{\psi}^{\bar{A}^O_{k_m} \alpha_n} \ket{\C{K}_{m-1}, k_m}$ where $\alpha_n$ denotes the internal wires of the arbitrary sequence representation we started with and $m \leq n$ with equality if all of the previous isometries produced a non-vacuum output and we dropped the time stamp as it is known by knowing $n$. The isometry outputs either 0 or 1 messages and therefore the output is $\sum_{k_{m+1}} \ket{\phi^{k_{m+1}}}^{A^I_{k_{m+1}} \alpha_{n+1}} + \lambda \ket{\Omega} \ket{\alpha_{n+1}^\Omega}^{\alpha_{n+1}}$ for some $\ket{\phi^{k_{m+1}}}^{A^I_{k_{m+1}} \alpha_{n+1}} \in \C{H}^{A^I_{k_{m+1}} \alpha_{n+1}}, \lambda \in \mathbb{C}, \ket{\alpha_{n+1}^\Omega}^{\alpha_{n+1}} \in \C{H}^{\alpha_{n+1}}$ with $|\lambda|^2 + \sum_{k_{m+1}} \braket{\phi^{k_{m+1}}|\phi^{k_{m+1}}} = 1$ (i.e. the overall state is normalized). Note that $k_{m+1} \in \C{N} \backslash \C{K}_{m-1} \cup k_m$. Otherwise, the process box would violate WSR as our induction assumption implies that the agents given by $\C{K}_{m-1} \cup k_m$ already received a non-vacuum input. We now add the control (i.e. define $C$) 

\begin{gather}
\begin{aligned}
    C(\ket{\phi^{k_{m+1}}}^{A^I_{k_{m+1}} \alpha_{n+1}}) = \ket{\phi^{k_{m+1}}}^{A^I_{k_{m+1}} \alpha_{n+1}} \ket{\C{K}_{m}, k_{m+1}} \\
    C(\ket{\Omega} \ket{\alpha_{m+1}^\Omega}^{\alpha_{m+1}}) = \ket{\Omega} \ket{\alpha_{m+1}^\Omega}^{\alpha_{m+1}} \ket{\C{K}_{m-1}, k_m}.
\end{aligned}
\end{gather}

We now need to check that if two states are assigned different controls, they were already orthogonal. Assume two states were not orthogonal. Then they are not perfectly distinguishable. This means that in some branch of the superposition the following isometries will send some part of them along the same path.\footnote{Note that there can be no destructive interference because the internal operations are isometries. The first state is the sum of a multiple of the second state and something orthogonal to the second state. These two terms cannot cancel each other because the isometries keep them orthogonal} Due to WSR and FAA, the path must contain exactly the agents which have not acted on the systems yet. By induction assumption and how we defined $C$ for the current time step, these agents are exactly those which are not part of the control yet. Therefore, the controls of the two states must contain the same agents, $\C{K}_n \cup k_{n+1} = \C{L}_n \cup l_{n+1}$. As two states are orthogonal if they belong to different labs, we also have $k_{n+1} = l_{n+1}$ and therefore the controls are the same. 

\end{proof}

\pbtoqcqc*

\begin{proof}

We use \cref{lemma:addcontrol} so that we can assume a sequence representation with internal wires that correspond to the control. The isometries $\bar{V}_{n+1}$ are then defined on spaces as in \cref{eq:isomdirectsum}. Note that the $n+1$ subspaces that make up the direct sum that is the domain of $\bar{V}_{n+1}$ in \cref{eq:isomdirectsum} are orthogonal. This is because controls with different lengths are orthogonal.

We further divide each of these subspaces in two, namely the subspace that is mapped to 0-message states which we shall denote with $\C{H}^{m,0}_{t=2n+1} \subseteq \C{H}^{\bar{A}^O}_{t=2n+1} \otimes \C{H}^{\alpha_n} \otimes \C{H}^{C_m}$ and the subspace that is mapped to 1-message states, $\C{H}^{m,1}_{t=2n+1} \subseteq \C{H}^{\bar{A}^O}_{t=2n+1} \otimes \C{H}^{\alpha_n} \otimes \C{H}^{C_m}$. These subspaces are orthogonal because $\bar{V}_{n+1}$, an isometry, maps them, by definition, to orthogonal subspaces. Additionally, as $\bar{V}_{n+1}$ can only send 0- and 1-message states due to WSR and \cref{lemma:stretching}, we can assume $\C{H}^{m,0}_{t=2n+1}  \oplus \C{H}^{m,1}_{t=2n+1}  = \C{H}^{\bar{A}^O}_{t=2n+1} \otimes \C{H}^{\alpha_n} \otimes \C{H}^{C_m}$.

We then define the restriction of $\bar{V}_{n+1}$ to each of these subspaces

\begin{equation}
    \bar{V}^{m, 0/1}_{n+1} \coloneqq \bar{V}_{n+1}|_{\C{H}^{m,0/1}_{t=2n+1}}.
\end{equation}

Note that $\bar{V}^{m, 0/1}_{n+1}$, as a restriction of an isometry, is itself an isometry, on the corresponding restricted subspace. By definition, we then have

\begin{gather}
\begin{aligned}
    \bar{V}^{m, 0}_{n+1}: \C{H}^{m,0}_{t=2n+1} \rightarrow \ket{\Omega, t=2n+2}^I \otimes \C{H}^{\alpha_{n+1}} \otimes \C{H}^{C_{m}} \\
    \bar{V}^{m, 1}_{n+1}: \C{H}^{m,1}_{t=2n+1} \rightarrow \C{H}^{A^I}_{t=2n+2} \otimes \C{H}^{\alpha_{n+1}} \otimes \C{H}^{C_{m+1}}
\end{aligned}
\end{gather}

where $\C{H}^{A^I}_{t=2n+2} \subsetneq \C{H}^{\bar{A}^I}_{t=2n+2}$ is the space of 1-message states with time stamp $t=2n+2$.

Even if we consider subspaces corresponding to different $n$, we find that the spaces $\C{H}^{m,0/1}_{t=2n+1}$ are all orthogonal to each other. We already argued why the subspaces with different $m$ are orthogonal to each other and why the subspaces that are mapped to 0-message states are orthogonal to those mapped to 1-message states. Subspaces with different $n$ are orthogonal because the time stamps differ. 

Additionally, the images $\bar{V}^{m, 0/1}_{n+1}(\C{H}^{m,0/1}_{t=2n+1})$ are also orthogonal to each other. The image of $\C{H}^{m,0}_t$ is per definition orthogonal to the image of $\C{H}^{m', 1}_{t'}$ regardless of the time stamp and control length. If $n$ differs, the time stamps differ, and the images are also orthogonal. The control of the image will have either length $m$ if the space is mapped to the 0-state space or $m+1$ if it is mapped to the 1-state space. Thus if $m$ differs, the images are also orthogonal. 

This means that any direct sum of any $\bar{V}^{m, 0/1}_{n+1}$ is an isometry because it fulfills condition 4 of \cref{lemma:isochar}.

We now define $V^{m, 0/1}_{n+1}$ by simply dropping the time stamp. We thus have

\begin{equation}
    \bar{V}^{m, 0/1}_{n+1} = V^{m, 0/1}_{n+1} \otimes \ket{t=2n+2} \bra{t=2n+1}.
\end{equation}

The input and output spaces of $V^{m, 0/1}_{n+1}$ are then

\begin{gather}
\begin{aligned}
    V^{m, 0}_{n+1}: \C{H}^{m,0}_{n} \rightarrow \ket{\Omega}^I \otimes \C{H}^{\alpha_{n+1}} \otimes \C{H}^{C_{m}} \\
    V^{m, 1}_{n+1}: \C{H}^{m,1}_{n} \rightarrow \C{H}^{A^I} \otimes \C{H}^{\alpha_{n+1}} \otimes \C{H}^{C_{m+1}}
\end{aligned}
\end{gather}

where $\C{H}^{A^{I/O}}$ are the input/output spaces of the agents in the QC-QC picture and $\C{H}^{m,0/1}_{n} \subseteq (\ket{\Omega}^O \oplus \C{H}^{A^O}) \otimes \C{H}^{C_m}$ such that $\C{H}^{m,0/1}_{t=2n+1} = \C{H}^{m,0/1}_{n} \otimes \ket{t=2n+1}$.

Note that $V^{m, 0/1}_{n+1}$ is still an isometry. While we previously used the orthogonality of the time stamps to argue that $\bar{V}^{m, 0/1}_{n+1}$ and $\bar{V}^{m, 0/1}_{n'+1}$ for $n \neq n'$ are orthogonal, we can now use the internal wires $\alpha_n$. Different Hilbert spaces $\C{H}^{\alpha_n}$ are orthogonal in the same sense that the Hilbert spaces of different agents are orthogonal. It is in principle possible that two subsequent isometries are not internally connected i.e., there are no $\alpha_n$ wires. However, in this case we can add one-dimensional wires $\alpha_n$ and still obtain a sequence representation of the same process box. 

By the same argument as for $\bar{V}^{m, 0/1}_{n+1}$, any direct sum of isometries $V^{m, 0/1}_{n+1}$ is thus also an isometry.

Let us now define the isometries $V_{m+1}$ in the QC-QC picture. This isometry should take as inputs states with a control of length $m$ and should output states with control of length $m+1$. In particular, this means it should never output 0-message states. What $V_{m+1}$ should do then is applying all $V^{m, 1}_{n+1}$ for $n > m$. However, this leaves the action on states that are sent to 0-message states, which need not be 0-message states themselves, undefined. Let us consider a state in $\C{H}^{m, 0}_{n}$ and let us for simplicity assume that it is a separable state $\ket{\psi}^{A^O_{k_m}} \ket{\phi}^{\alpha_n} \ket{\C{K}_{m-1}, k_m}$.\footnote{It should be noted that the subspace $\C{H}^{m, 0}_{n}$ does not necessarily contain a separable state or is spanned by separable states. For full generality, one should thus consider linear combinations of the form $\sum_{\C{K}_{m-1}, k_m} \ket{\psi^{\C{K}_{m-1}, k_m}}^{A^O_{k_m} \alpha_n} \ket{\C{K}_{m-1}, k_m}$, where $\ket{\psi^{\C{K}_{m-1}, k_m}}^{A^O_{k_m} \alpha_n}$ is a bipartite state on the output and ancillary wire. However, the arguments are the same due to linearity.} The only isometry that is well-defined on this state is $V^{m, 0}_{n+1}$, but then the output is a 0-message state, $\ket{\Omega}^I \ket{\phi'}^{\alpha_{n+1}} \ket{\C{K}_{m-1}, k_m}$ for some $\ket{\phi'}^{\alpha_{n+1}} \in \C{H}^{\alpha_{n+1}}$. However, due to LO, we know that the state on the output wires in the next time step will always be $\ket{\Omega}^O \ket{\phi'}^{\alpha_{n+1}} \ket{\C{K}_{m-1}, k_m}$. We can thus just skip this time step and look at what the process box does in the next time step. The state $\ket{\Omega}^O \ket{\phi'}^{\alpha_{n+1}} \ket{\C{K}_{m-1}, k_m}$ is a linear combination of states in $\C{H}^{m, 0}_{n+1}$ and $\C{H}^{m, 1}_{n+1}$ (remember that the process box does not have to map 0-message states to 0-messages because only the local agents need to respect LO). Due to linearity, we can instead consider what should happen if the state is an element of one or the other. If it is an element of $\C{H}^{m, 1}_{n+1}$, we apply $V^{m, 1}_{n+2}$ to it and obtain a 1-message state in $\C{H}^{m+1, 1}_{n+2}$, i.e. the control has length $m+1$, which is precisely what we wanted. If it is an element of $\C{H}^{m, 0}_{n+1}$, we apply $V^{m, 0}_{n+2}$ obtain a 0-message state again. In this case, we can repeat this procedure. If $m<N$ we obtain a 1-message state at some point due to FAA and if $m=N$, then all agents acted already, and the procedure will stop when it reaches the last time step and outputs something to the global future (which we can take to be a 1-message state).

There is therefore some time step during which $V_{m+1}$ receives a one-message state from the agents and another time step, which may be the same or a later one, during which $V_{m+1}$ sends a one-message state to the agents. We can write this idea compactly as an equation

\begin{gather}\label{eq:isomcompact}
\begin{aligned}
    V_{m+1} =  \bigoplus_{n'_m \geq m} \bigoplus_{n_{m+1} \geq n'_m} V^{m, 1}_{n_{m+1}} \ket{\Omega}^O \bra{\Omega}^I V^{m, 0}_{n_{m+1}-1} \ket{\Omega}^O \bra{\Omega}^I ... \ket{\Omega}^O \bra{\Omega}^I V^{m, 0}_{n'_m+1})
\end{aligned}
\end{gather}


Note now that the map in \cref{eq:isomcompact} is still an isometry. Each term in the direct sum is an isometry because compositions of isometries are again isometries. Furthermore, two different terms take orthogonal subspaces as inputs and send them to orthogonal subspaces. Thus, condition 4 of \cref{lemma:isochar} applies.

The isometries $V_{m+1}$ then fulfill the conditions for being internal operations in the QC-QC framework. They take as inputs states $\ket{\psi}^{A^O_{k_m}} \ket{\tilde{\alpha}_m}^{\tilde{\alpha}_m} \ket{\C{K}_{m-1}, k_m}$ (and linear combinations of such states) and send them in an isometric fashion to states of the form $\sum_{k_{m+1}} \ket{\phi^{k_{m+1}}}^{A^I_{k_{m_1}} \tilde{\alpha}_{m+1}} \ket{\C{K}_{m-1} \cup k_m, k_{m+1}}$. Here $\tilde{\alpha}_m$ and $\tilde{\alpha}_{m+1}$ are new internal wires which are constructed from the internal wires $\alpha_n$ from the process box.

Let us now check that this construction is operationally equivalent to the original process box. Using that

\begin{gather}
\begin{aligned}
    \bar{V}_{n+1} &= \bigoplus_{m \leq n} (\bar{V}^{m, 0}_{n+1} \oplus \bar{V}^{m, 1}_{n+1}) \\
\end{aligned}
\end{gather}

by definition on the WSR restricted space, we find that the Choi vector is

\begin{gather}
\begin{aligned}
    \dket{\bar{V}} &= \dket{\bar{V}_1} * ... *\dket{\bar{V}_{|\C{T}^I|+1}} \\
    &= \sum_{(n_1,...,n_N)} \dket{V^{0, 0}_{1}} * ... * \dket{\bar{V}^{0, 1}_{n_1}} * \dket{\bar{V}^{1, 0}_{n_1+1}} *... * \dket{\bar{V}^{N-1, 1}_{n_N}} * \dket{\bar{V}^{N, 0}_{n_N+1}} *...* \dket{\bar{V}^{N, 1}_{|\C{T}^I|+1}}
\end{aligned}
\end{gather}
where $(n_1,...,n_N)$ is an ordered subset of $\{1,..., |\C{T}^I|\}$ such that $n_k < n_{k+1}$. In other words, $(n_1,...,n_N)$ signifies that $\bar{V}_{n_k}$ is the $k$-th isometry to send a one-message state, $k=1,...,N$. 

\begin{gather}
\begin{aligned}
    (\dket{\bar{A}_1} \otimes ... \otimes \dket{\bar{A}_N}) * \dket{\bar{V}} =& \bigotimes_{k=1}^N \bigotimes_{n=1}^{|\C{T}^I|} (\dket{A_k} \otimes \ket{t=2n} \ket{t=2n+1} + \ket{\Omega, t=2n} \ket{\Omega, t=2n+1}) * \\
    & \sum_{(n_1,...,n_N)} \dket{\bar{V}^{0, 0}_{1}} * ... * \dket{\bar{V}^{0, 1}_{n_1}} * \dket{\bar{V}^{1, 0}_{n_1+1}} *... * \dket{\bar{V}^{N-1, 1}_{n_N}} * \dket{\bar{V}^{N, 0}_{n_N+1}} *...* \dket{\bar{V}^{N, 0}_{|\C{T}^I|+1}} \\
    =& \sum_{(n_1,...,n_N)} \bigotimes_{k=1}^N \dket{A_k} 
    * \dket{V^{0, 0}_{1}} * \ket{\Omega}^I \ket{\Omega}^O * ... * \dket{V^{0, 1}_{n_1}} * \dket{V^{1, 0}_{n_1+1}} *\\&... * \dket{V^{N-1, 1}_{n_N}} * \dket{V^{N, 0}_{n_N+1}} * \ket{\Omega}^I \ket{\Omega}^O * ...* \ket{\Omega}^I \ket{\Omega}^O * \dket{V^{N, 0}_{|\C{T}^I|+1}}.
\end{aligned}
\end{gather}

On the other hand, for the process vector of the QC-QC we find

\begin{gather}
\begin{aligned}
    \ket{w} &= \dket{V_{N+1}} * ... * \dket{V_1} \\
    &= \sum_{(n_1, n'_1, ..., n_N, n_N')} \dket{V^{0, 0}_{1}} * \ket{\Omega}^I \ket{\Omega}^O * \\ &... * \dket{V^{0, 1}_{n_1}} * \dket{V^{1, 0}_{n_1'+1}} *... * \dket{V^{N, 1}_{n_N}} * \dket{V^{N+1, 0}_{n_N'+1}} * \ket{\Omega}^I \ket{\Omega}^O * ...* \ket{\Omega}^I \ket{\Omega}^O * \dket{V^{N+1, 0}_{|\C{T}^I|+1}} \\
    &= \sum_{(n_1, ..., n_N)} \dket{V^{0, 0}_{1}} * \ket{\Omega}^I \ket{\Omega}^O * \\ &... * \dket{V^{0, 1}_{n_1}} * \dket{V^{1, 0}_{n_1+1}} *... * \dket{V^{N-1, 1}_{n_N}} * \dket{V^{N, 0}_{n_N+1}} * \ket{\Omega}^I \ket{\Omega}^O * ...* \ket{\Omega}^I \ket{\Omega}^O * \dket{V^{N, 0}_{|\C{T}^I|+1}}
\end{aligned}
\end{gather}
where we used that $\dket{V^{m, 0/1}_{n}}*\dket{V^{m+1, 0/1}_{n'+1}} = \delta_{n,n'} \dket{V^{m, 0/1}_{n}}*\dket{V^{m+1, 0/1}_{n+1}}$.

This shows that the process box and the QC-QC are operationally equivalent.
\end{proof}

\section{Justification of \texorpdfstring{\Cref{conjecture:simplifying}}{Conjecture 1}}\label{sec:justification}

\simplifying*

We give here some more technical justification in addition to what we already stated in \cref{sec:simplifying} before making the conjecture. The arguments here should capture most of the ideas necessary for a complete proof, except for explicitly defining the necessary modifications to the internal operations of the process box and finally showing that the final result is still a process box. 

As we already mentioned, both the set $\C{T}$ and the internal operations of the process box influence which inputs to the process box can influence which outputs of the process box. We can thus introduce new order relationships $t_1 < t_2$ between previously unrelated elements $t_1, t_2 \in \C{T}$ without changing anything. If we change nothing else about the process box, then these new relationships are irrelevant as no message with time stamp $t_1$ can influence a message with time stamp $t_2$. 

It is then always possible to convert a partially ordered set into a totally ordered set by introducing suitable new order relationships between all previously unrelated elements \cite{szpilrajn1930extension}. We then obtain a finite, discrete and totally ordered set which can thus be taken to have the form as in property 1 of \cref{conjecture:simplifying} as relabeling the elements can always be done.

For property 2, we first consider agents with trivial inputs. The idea is similar to what we did in the proof of \cref{prop:PBunitary} when arguing that the agents $P_n$ can be assumed to output before all other agents. Here we have the added complication that the agent potentially has multiple valid output times. Let us assume that the agent can send during $m$ time steps in $\C{T}'$. As their output cannot be influenced by any input from the process box, we can assume that the agent outputs at $t=1$ and outputs an additional $m$-dimensional system that tells the process box when to act on the message. We then modify the isometries $V_{n+1}$ by adding additional internal wires to forward the message to the right isometry. 

This then also allows us to consolidate all of these agents into the global past.

Next, we consider agents with non-trivial inputs and as a first step, we argue that $\C{T}'^I \cap \C{T}'^O = \emptyset$. Assume there was $t \in \C{T}'^I \cap \C{T}'^O$. An agent receiving a message at time $t$ cannot send anything until after $t$ as $\C{O}_i(t) > t$. The first time the agent could send something is $t+1$ (assuming $t+1 \in \C{T}^O$). Therefore, we can relabel $t$ in $\C{T}^I$ with $t+1/2$. We then relabel $t+1/2 \rightarrow t+1$ and also $t' \rightarrow t' + 1$ for all $t' > t+1/2$. 

At this point, $\C{T}'$ consists of sequences of elements in $\C{T}'^I$ alternating with sequences of elements in $\C{T}'^O$. The idea is now to consider each of these sequences as a single time step which gives us property 2. 

However, consolidating multiple time steps into one is potentially problematic due to $\C{O}_i$ being a bijection while also having to fulfill $\C{O}_i(t) > t$. We thus need to redefine it which gives us both property 3 and the part of property 2 that states $\C{T}'^{I/O}_i = \C{T}'^{I/O}$ and $\C{O}_i(t) = \C{O}(t) = t+1$ for all agents. We can use a similar argument as we did when justifying that the local operations of the agents can be assumed to be time-independent. The original sending time becomes a part of the message. We replace input states $\ket{\psi} \ket{t}$ with $\ket{\psi} \ket{t} \ket{t'}$ where the first two states correspond to the message while $t'$ is the time stamp that $t$ is consolidated into. We proceed similarly for the output, replacing an output $\ket{\phi} \ket{\C{O}_i(t)}$ with $\ket{\phi} \ket{\C{O}_i(t)} \ket{t'+1}$. The process box then acts the same as it did before (this can be achieved by using $\ket{\C{O}_i(t)}$ to decide whether a given isometry should act on the message or if it should forward it to the next isometry via an internal wire). 

This then also allows us to assume that $\C{T}'^{I/O}_i = \C{T}'^{I/O}$. We have that $t \not \in \C{T}'^I$ if and only if $t+1 \not \in \C{T}'^O_i$ due to how we defined $\C{O}_i$. But then adding $t$ to the domain of $\C{O}_i$ with $\C{O}_i(t) = t+1$ maintains the properties of $\C{O}_i$. \footnote{Compare this to the following example: $\C{T}^I_i = \{2\}$ and $\C{T}^O_i = \{5\}$ with $\C{O}_i(2) = 5$. Let us say, $\C{T}^I = \{2, 4, 6\}$ while $\C{T}^O = \{3, 5, 7\}$. Then there is no way to extend $\C{T}^I_i$ to all possible input times and $\C{T}^O_i$ to all possible output times without violating $\C{O}_i(t) > t$. Clearly, we have to define $\C{O}_i(6) = 7$, but then the only possible choice for when the agent outputs for a message received at $t=4$ would be $\C{O}_i(4) = 3 < 4$.} Note that during $t$, the agent only ever receives the vacuum, which means they only ever send the vacuum during $t+1$. But there is no issue with considering them to be input/output times for the agent.

\section{License information}

\Cref{fig:QCFO,fig:QCCC,fig:QCQC,fig:new_QCQC} were taken from \cite{Wechs_2021}, which is covered by a Creative Commons Attribution 4.0 International license. The license can be found \href{https://creativecommons.org/licenses/by/4.0/legalcode}{here}. A human-readable version can be found \href{https://creativecommons.org/licenses/by/4.0/}{here}.

\Cref{fig:sequencerep} was taken from \cite{Portmann_2017}. The relevant license was obtained from \href{https://ieeexplore.ieee.org/document/7867830}{here} and is reproduced below.

\begin{figure}[h]
    \centering
    \includegraphics[width=\textwidth]{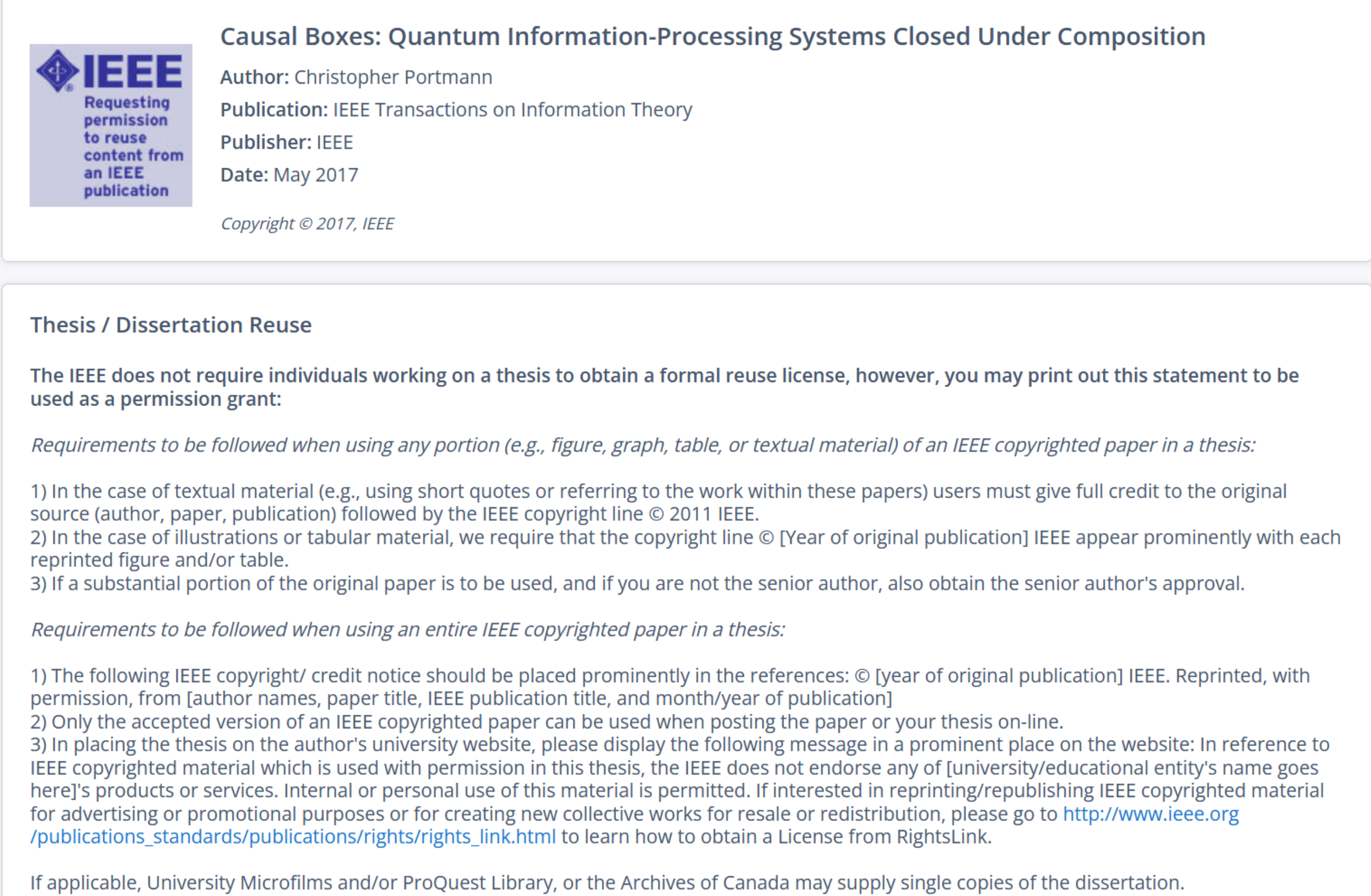}
\end{figure}

\end{document}